\documentclass[11pt]{article}
\usepackage[hmargin=2.4cm,vmargin=2.4cm]{geometry}
\usepackage{url}
\usepackage{amssymb,latexsym,amsmath,amsthm,psfrag,array,cite,comment,bigints,mathrsfs,bbm,setspace}
\usepackage{amsfonts}
\usepackage{array}
\usepackage{csquotes}
\usepackage{graphicx, caption,array}
\usepackage[colorlinks,linkcolor={blue},urlcolor ={blue},citecolor={blue}]{hyperref}
\usepackage{authblk}
\usepackage{booktabs}
\usepackage[longnamesfirst]{natbib}
\usepackage{graphics}
\usepackage{pdfpages}
\usepackage{mathtools}
\usepackage{booktabs,colortbl,tabularx}
\usepackage{algorithm2e}
\usepackage{graphicx,subcaption,float}
\usepackage{adjustbox}
\usepackage{xr}
\usepackage{subcaption} 
\usepackage{etoolbox}
\AtBeginEnvironment{align}{
  \setlength{\abovedisplayskip}{3pt} 
  \setlength{\belowdisplayskip}{3pt} 
}
\AtBeginEnvironment{align*}{
  \setlength{\abovedisplayskip}{3pt} 
  \setlength{\belowdisplayskip}{3pt} 
}

\AtBeginEnvironment{eqnarray}{
  \setlength{\abovedisplayskip}{3pt} 
  \setlength{\belowdisplayskip}{3pt} 
}
\AtBeginEnvironment{eqnarray*}{
  \setlength{\abovedisplayskip}{3pt} 
  \setlength{\belowdisplayskip}{3pt} 
}

\usepackage{titlesec}
\titlespacing*{\section}
  {0pt}{8pt plus 2pt minus 2pt}{6pt plus 2pt minus 2pt}  
\titlespacing*{\subsection}
  {0pt}{8pt plus 2pt minus 2pt}{4pt plus 2pt minus 2pt}
\usepackage{placeins}
\DeclareMathOperator*{\argmin}{arg\,min}
\DeclareMathOperator*{\argmax}{arg\,max}
\DeclareMathOperator*{\median}{median}
\SetAlgoCaptionSeparator{}

\newtheorem{thm}{Theorem}

\newtheorem{aplemma}{Lemma}[section]

\theoremstyle{definition}

\newtheorem{rem}{Remark}
\numberwithin{equation}{section}
\newcommand{\E}{\mathrm{E}}
\newcommand{\vvec}{\mathrm{vec}}

\begin{document}

\title{Unsupervised Learning in a General Semiparametric Clusterwise Index Distribution Model 
}

\author{Jen-Chieh Teng}
\affil{Data Science Degree Program, National Taiwan University, Taipei, Taiwan\\D09948011@ntu.edu.tw}
\author{Chin-Tsang Chiang}
\affil{Institute of Applied Mathematical Sciences, National Taiwan University, Taipei, Taiwan \\chiangct@ntu.edu.tw} 

\date{\today}%
\maketitle
 
\begin{abstract}
This study introduces a general semiparametric clusterwise index distribution model to analyze how latent clusters affect the covariate-response relationships. By employing sufficient dimension reduction to account for the effects of covariates on the cluster variable, we develop a distinct method for estimating model parameters. Building on a subjectwise representation of the underlying model, the proposed separation penalty estimation method clusters data points and estimates cluster index coefficients. We propose a convergent algorithm for this estimation procedure and incorporate a heuristic initialization to expedite optimization.
The resulting clustering estimator is subsequently used to fit the cluster membership model and to construct an optimal clustering rule, with both procedures iteratively updating the clustering and parameter estimators. Another key contribution of our method is the development of two consistent semiparametric information criteria for selecting the number of clusters. In line with principles of clustering and estimation in supervised learning, the estimated cluster structure is consistent and optimal, and the parameter estimators possess the oracle property. Comprehensive simulation studies and empirical data analyses illustrate the effectiveness of the proposed methodology.
\end{abstract}
\begin{spacing}{1.9}
\begin{section}{Introduction}

In fields such as real estate economics, biomedical research, and epidemiology, identifying latent clusters is fundamental to characterizing heterogeneity in covariate-response relationships. Clustered regression methodologies aim to identify unobserved structures, thereby enabling more refined estimation of covariate effects within distinct subpopulations. For instance, in a real estate valuation study in New Taipei City, the nonparametric framework of \citet{li2023nonparametric} revealed substantial variation in how factors such as transaction date, building age, proximity to transit infrastructure, access to amenities, and location influenced house prices. Notably, whether a property was located on the left or right riverbank emerged as a key axis of heterogeneity, reflecting the spatial complexity of urban housing markets. In biomedical research, latent structure models have similarly proven effective in delineating clinically relevant subgroups. The method developed by \citet{wei2013latent}, applied to the Cleveland Heart Disease dataset, identified distinct clusters characterized by heterogeneous relationships between maximum heart rate and covariates spanning demographic and biological domains. These clusters were closely aligned with the heart disease status, illustrating the capacity of such a model to uncover biological variation not readily apparent through traditional analyses.

We introduce a semiparametric clusterwise index distribution model for the unsupervised detection of latent structure within heterogeneous subpopulations. A dedicated estimation procedure is developed for both
model parameters and cluster assignments. Motivated by real estate and biomedical research applications, we revisit empirical studies where previous analyses indicate latent clustering related to the left or right riverbank and heart disease status, respectively. 
By including these factors as covariates, the proposed model remains effective in capturing the latent structure underlying the 
covariate-response relationships. The framework is sufficiently adaptable to capture nuanced forms of dependence and can be applied effectively in various scientific contexts.

Due to missing, unmeasured, or unobservable characteristics, observations cannot be reliably assigned to inherent subgroups. This limitation necessitates inferring latent cluster memberships and the corresponding covariate-response relationships. A substantial body of work has addressed this challenge by developing models that uncover latent structures and accommodate unobserved heterogeneity in regression settings. This modeling framework---commonly referred to as switching regression \citep[]{quandt1972new}, clusterwise regression \citep[]{charles1977regression}, or regression clustering \citep[]{devijver2017model}---has been widely adopted in diverse fields. To handle the complexities of unsupervised learning in such contexts, researchers
have employed mixture and subjectwise regression models
to identify latent clusters and capture variation in the covariate-response relationships across
heterogeneous subpopulations.
 
Building on a subjectwise representation of the clusterwise linear model, \citet{spath1981correction} recast the problem of cluster membership identification as a combinatorial optimization task and proposed an exchange algorithm to obtain a suboptimal clustering of data points. \citet{ma2017concave} and \citet{ma2020exploration} employed a pairwise fusion strategy with concave penalties to consistently recover latent group structures. As a concave alternative, \citet{chen2021identifying} proposed an adaptive fusion method based on the $\ell_{1}$ penalty.
More recently, \citet{han2024robust} extended the linear transformation model to incorporate latent cluster effects, integrating a pairwise fusion penalty within a rank-based estimation procedure for subject index coefficients.
To address the computational complexity associated with pairwise fusion, \citet{tang2021individualized} introduced a multidirectional separation penalty, and \citet{he2023center} proposed a center-augmented regularization framework. 
Despite these advances, subjectwise regression methods often struggle to recover the true clustering accurately and generally involve substantial computational complexity.

In contrast to subjectwise regression methods, mixture regression models introduce clusterwise distribution and cluster membership models for capturing latent cluster structures. Early contributions by \citet{quandt1958estimation,quandt1960tests,quandt1972new} and \citet{chow1960tests} investigated maximum likelihood estimation in switching regression models under covariate-independent cluster membership. \citet{kiefer1978discrete} later proposed a more efficient estimation strategy within this setting. \citet{desarbo1988maximum} developed a conditional mixture method that estimates cluster assignments and model parameters to account for individual heterogeneity.
Subsequent advances introduced covariate-dependent structures for cluster membership. Building on the mixture-of-experts framework of \citet{jacobs1991adaptive}, \citet{jordan1994hierarchical} proposed a hierarchical model combining clusterwise generalized linear models with a logistic gating function. \citet{jiang1999hierarchical} further investigated the theoretical properties of this model class. Numerous methods have been proposed to determine the presence or number of latent clusters, including those of \citet{peng1996bayesian} and \citet{shen2015inference}. It is worth noting that mixture regression models are sensitive to misspecification in the cluster membership structure.

Building on the single-index framework, the proposed clusterwise index distribution model enhances our understanding of the underlying data structure. Its semiparametric nature reduces the risk of parametric misspecification and obviates the complexities of identifying cluster-specific dimension-reduction subspaces. Estimation under the corresponding mixture distribution is challenging, owing to the difficulty in estimating the unknown distribution function and the constraints imposed on the cluster membership model. We adopt a subjectwise model representation to address this issue and facilitate cluster structure identification and parameter estimation. The proposed method comprises three key components. First, we combine a pseudo least integrated squares criterion with a separation penalty to clustering data points and estimate the clusterwise index distribution model. 
The corresponding clustering estimator informs the estimation of the cluster membership model. Second, we construct an optimal clustering rule while iteratively updating the clustering and parameter estimators. Third, we propose two consistent semiparametric information criteria for selecting the number of clusters. In line with principles of clustering and estimation in supervised learning, the clustering estimator consistently recovers the true clusters and maximizes the probability of correct membership, and the parameter estimators possess the oracle property.

For separation penalty estimation, we employ the alternating direction method of multipliers algorithm of \citet{boyd2011distributed} to ensure numerical convergence and incorporate the difference of convex functions programming \citep{an2005dc} to accommodate the separation penalty. 
A heuristic initialization scheme is provided to efficiently minimize the objective function. In the refined estimation phase, sufficient dimension reduction \citep{li1991sliced} is used to characterize the influence of covariates on cluster membership, mitigating sensitivity to misspecification in parametric and semiparametric models. Building on the clustering estimator from the separation penalty estimation procedure, we propose a pseudo least squares estimator of the central subspace \citep{cook1994interpretation}. 
The proposed criteria for selecting the number of clusters yield consistent estimators and exhibit favorable finite-sample performance with distinct advantages.

The article is organized as follows. Section \ref{sec:background} presents a general semiparametric framework for modeling clusterwise and subjectwise distributions, along with a flexible specification for cluster membership.
Section \ref{sec:estimation I} describes an estimation-clustering procedure for unsupervised settings. Section \ref{sec:estimation II} details the computational algorithm for the separation penalty estimation and the implementation of the proposed estimation procedure. Section \ref{sec:simulation} presents simulation studies to assess the finite-sample performance of the competing estimators. Section \ref{sec:applications} illustrates the applicability of the proposed methodology with empirical data. Section \ref{sec:discussion} concludes with a summary of the main findings and a discussion of future research directions.

\end{section}

\begin{section}{Modeling Conditional Distributions and Cluster Membership}
\label{sec:background}

Let $Y$ be a quantitative response variable with support $\mathcal{Y}$, and let $X$ be a $p \times 1$ covariate vector with support $\mathcal{X}$.
Denote $k$ as the number of latent clusters, $C$ as the corresponding cluster variable, taking values in $\mathcal{C} =\{1, \dots, k\}$, and $Z$ as the augmented covariate vector $(1, X^{\top})^{\top}$.

\begin{subsection}{Clusterwise and Subjectwise Index Distribution Models} \label{subsec:model_prop}

To investigate how latent clusters influence the conditional distribution of a response variable given covariates, we introduce the clusterwise index distribution (CID) model for the conditional distribution $F(y|x,\ell)$ of $Y$ given $(X, C)$:
\begin{align}
F(y|x, \ell) = G\big(y,\gamma_{\ell}^{\top}z\big), (x, \ell, y) \in \mathcal{X} \times \mathcal{C} \times \mathcal{Y}, \label{fn:semi_dist}
\end{align}
where $\gamma_{\ell} = (\gamma_{\ell 0}, \gamma_{\ell 1}, \dots, \gamma_{\ell p})^{\top}$ is a vector of index coefficients for Cluster $\ell$,
and $G(y,v)$ is an unknown bivariate function that defines a valid conditional distribution in $y$
for each fixed $v$. The associated conditional density $g(y,v)$ satisfies 
$g\big(y,\gamma_{\ell_{1}}^{\top}z\big)\neq g\big(y,\gamma_{\ell_{2}}^{\top}z)$ for almost every $(x, y)$ whenever $\ell_{1}\neq \ell_{2}$.  For identifiability, we assume $\gamma_{11}\neq 0$ 
and normalize $(\gamma_{10}, \gamma_{11})=(0,1)$.
The proposed semiparametric framework encompasses clusterwise versions of linear, generalized linear, and linear transformation models as special cases. The CID model implies that the conditional mean $E[Y|x,\ell]$ of $Y$ given $(X, C)$ satisfies
\begin{align}
E[Y|x, \ell] = m\big(\gamma_{\ell}^{\top}z\big), \ell \in \mathcal{C}, \label{fn:semi_mean}
\end{align}
where $m(v)=\int y \,d_{y}G(y,v)$ is an unknown smooth function. When $\gamma_{\ell}$ comprises cluster-specific coefficients $\gamma_{(1)\ell}$ of length $1\leq q\leq p+1$ and cluster-invariant coefficients $\gamma_{(2)}$ corresponding to covariates $Z_{(1)}$ and $Z_{(2)}$, respectively, the cluster index $\gamma_{\ell}^{\top}Z$ can be written as $\gamma_{(1)\ell}^{\top}Z_{(1)}+\gamma_{(2)}^{\top}Z_{(2)}$, $\ell\in\mathcal{C}$. 
Accordingly, the CID model can be expressed more explicitly as
\begin{align}
F(y|x, \ell) = G\big(y,\gamma_{(1)\ell}^{\top}z_{(1)}+\gamma_{(2)}^{\top}z_{(2)}\big), \ell \in \mathcal{C}. \label{fn:semi_dist1}
\end{align}

Given unsupervised data $\{(X_i, Y_i)\}_{i=1}^n$, we initiate estimation of the CID model using its representation as a subjectwise index distribution (SID) model. This formulation characterizes the conditional distribution 
$F(y|x_i, \beta_i)$ of $Y_i$ given $X_i$ and the subject index coefficients $\beta_i$ as
\begin{align}
F(y|x_i,\beta_i) = G\big(y, \beta_i^{\top} z_i\big), i = 1, \dots, n,\text{ with } \beta_1, \dots, \beta_n \in \{ \gamma_1, \dots,\gamma_k\}. \label{fn:semi_dist_ind}
\end{align}
Each $\beta_i$ can be decomposed into subject-specific coefficients $\beta_{(1)i}\in \{ \gamma_{(1)1}, \dots,\gamma_{(1)k}\}$ and subject-invariant coefficients $\beta_{(2)}=\gamma_{(2)}$. An alternative strategy based on the mixture distribution model
$\sum_{\ell = 1}^k G(y, \gamma^{\top}_{\ell} z)\pi(\ell | x)$, where $\pi(\ell | x)$ denotes the conditional probability of $\{C=\ell\}$ given $X$, 
is hindered by the complexity of estimating $G(y,v)$ and the constraints imposed on $\pi(\ell | x)$, $c\in\mathcal{C}$. In contrast, the SID-based formulation avoids these limitations and provides clusters of data points that facilitates estimation under the proposed cluster membership model.

\end{subsection}

\begin{subsection}{Cluster Membership Model}  \label{subsec:model_variants}

In the unsupervised learning setting, we refine the initial estimation of the CID model by introducing the cluster membership probabilities $\pi(\ell | x)$, $\ell\in\mathcal{C}$. To characterize the dependence of $C$ on $X$, we incorporate sufficient dimension reduction (SDR) via the specification

\vspace{-0.1in}
\begin{align}
\pi( \ell | x) = \pi_{\ell}\big(A^{\top}_{d}x\big), (x, \ell) \in \mathcal{X} \times \mathcal{C}, \label{fn:SDR}
\end{align}
where $A_{d} = (\alpha_1, \dots, \alpha_d)$ is a $p\times d$ full-rank matrix spanning the central subspace (CS) and $\pi_{\ell}(u)$ is an unknown function satisfying $0\leq \pi_{\ell}(u)\leq 1$, $\ell=1,\dots, k$, and $\sum^{k}_{\ell=1}\pi_{\ell}(u)=1$ for all $u\in\mathbb{R}^d$. 
The structural dimension $d$ governs model complexity: when $d = p$, the model is fully nonparametric; for $d < p$, it encompasses semiparametric and parametric forms, including the single-index model of \citet{cosslett1983distribution} and several classical models for categorical outcomes \citep[cf.][]{nelder1972generalized}.
In the special case where $C$ and $X$ are independent, the SDR-based model (\ref{fn:SDR}) reduces to the covariate-independent form
$\pi(\ell|x) = \pi_{\ell}$ for all $\ell \in \mathcal{C}$ and $x$.
To ensure the identifiability of $A_{d}$ in the semiparametric setting, we adopt a local coordinate system on the Grassmann manifold \citep{borisenko1991grassmann}, parameterizing the basis as $A_{d}=(I_{d},A^{\top}_{d*})^{\top}$, where $I_{d}$ is the $d\times d$ identity matrix and $A_{d*}=(\alpha_{1*},\dots,\alpha_{d*})$ is a $(p-d)\times d$ parameter matrix. The $i$th component of $\alpha_{j*}$ captures the relative contribution of the $(i + j)$th covariate, compared to the $j$th, along the $j$th direction of the CS, for $j = 1, \dots, d$.

Building on the CID model, the conditional density $f(y|x,\ell)$ of $Y$ given $(X,C)$ is
\vspace{-0.1in}
\begin{align}
f(y|x, \ell) = g\big(y,\gamma_{\ell}^{\top}z\big), (x, \ell, y) \in \mathcal{X} \times \mathcal{C} \times \mathcal{Y}. \label{fn:semi_den}
\end{align}
Together with the SDR-based cluster membership model, this yields the posterior probability of belonging to cluster $\ell$ as
\begin{align}
\pi(\ell |x, y) = \frac{ g\big(y, \gamma^{\top}_{\ell}z\big) \pi_{\ell}\big(A^{\top}_{d}x\big)}{\sum_{\ell_1 = 1}^k  g\big(y, \gamma^{\top}_{\ell_1}z\big) \pi_{\ell_{1}}\big(A^{\top}_{d}x\big)}, (x,y,\ell) \in \mathcal{X} \times \mathcal{Y} \times \mathcal{C}. \label{fn:posterior} 
\end{align}
These posterior probabilities form the basis of the optimal clustering rule that maximizes the probability of correct membership.

\end{subsection}

\end{section}
\begin{section}{Oracle and Refined Estimation with Cluster Selection}

\label{sec:estimation I}

Let $\mathcal{G} = \{ \mathcal{G}_1, \dots, \mathcal{G}_k \}$ denote a partition of the individual-based index set $\{1, \dots, n\}$, and let $\beta$ be a column vector composed of 
$\beta_{1},\dots,\beta_{n}$ in the SID model. When each $\beta_{i}$ comprises $\beta_{(1)i}$ and $\beta_{(2)}$, $\beta$ is written as 
$\beta=\big(\beta^{\top}_{(1)},\beta^{\top}_{(2)}\big)^{\top}$, with $\beta_{(1)}=\big(\beta^{\top}_{(1)1},\dots,\beta^{\top}_{(1)n}\big)^{\top}$. 
The set of underlying clusters is denoted by $\mathcal{G}^{\text{o}}$, with the corresponding true coefficients $\beta^{\text{o}}$, where $\beta^{\text{o}}_i=\gamma^{\text{o}}_{\ell}$ if $i\in \mathcal{G}^{\text{o}}_{\ell}$ for $i=1,\dots, n$ and $\ell\in\mathcal{C}$. The assumptions and proofs of the main results are provided in Appendices \ref{sec:A.Assumptions}--\ref{sec:A.4}.

\begin{subsection}{Separation Penalty Estimation for the CID and SID Models}
\label{subsec:PSISP}

Let $K_{r}(v)$ be an 
$r$th-order kernel function satisfying
$\int K_{r}(v)dv$ $=1$, $\int v^{s}K_{r}(v)dv=0$, $s=1,\dots,r-1$, and $\int u^{r}K_{r}(v)dv\neq 0$.
We define the scaled kernel function $K_{r,h}(v)$ as
$K_{r}(v/h)/h$ and specify $K_{r}(v)$ to be symmetric, twice continuously differentiable, compactly supported, and such that $\int K^{(2)}_{r}(v)dv=0$, where $h$ denotes a generic bandwidth and $K^{(m)}_r(v)$ denotes the $m$th derivative of $K_r(v)$. 
 
Based on unsupervised data $\{(X_{i},Y_{i})\}^{n}_{i=1}$, $G(y,v)$ given $\beta$ is estimated by the kernel-weighted empirical estimator
\begin{align}
\widehat{G}_h(y, v; \beta) = \frac{\sum_{i=1}^{n} I(Y_i \leq y) K_{2,h}\big(\beta_i^{\top}Z_i - v\big)}{\sum_{i=1}^{n} K_{2,h}\big(\beta_i^{\top}Z_i - v\big)}. \label{fn:G_n_hat}
\end{align}
The subject and cluster coefficients are estimated by minimizing the pseudo sum of integrated squares (PSIS) with a separation penalty:
\begin{eqnarray}
\textsc{psis}_{\textsc{sp}}\big(\beta,\gamma_{(1)};\lambda\big) =  \frac{1}{2} \sum_{i=1}^n \int \big( I(Y_i \leq y)  - \widehat{G}_{h}^{-i}\big(y,\beta_{i}^{\top} Z_i; \beta\big) \big)^2 d \widehat{F}(y)
 + \lambda\sum^{n}_{i=1}\min_{\ell} \big\|\beta_{(1)i} -\gamma_{(1)\ell}\big\|_{1},~\label{fn:psisp}
\end{eqnarray}
where $\widehat{G}^{-i}_{h}(y, v; \beta)$ denotes the leave-one-out version of $\widehat{G}_{h}(y, v; \beta)$, $i=1,\dots, n$, $\widehat{F}(y)$ is the empirical distribution of $Y$, and $\|\cdot\|_{1}$ denotes the $\ell_{1}$-norm of a vector.  The tuning parameter $\lambda\geq 0$ controls within-cluster heterogeneity: larger values encourage closer alignment of subject-specific coefficients with their corresponding cluster centers, whereas smaller values permit greater deviation from the centers. If $\beta_{i}$ does not include subject-invariant coefficients, the penalty term is modified by replacing $\big(\beta_{(1)i},\gamma_{(1)\ell}\big)$
with $\big(\beta_{i},\gamma_{\ell}\big)$ for all $i=1,\dots, n$ and $\ell \in \mathcal{C}$. The criterion employs the $\ell_{1}$-type penalty rather than the $\ell_{2}$-type penalty, as the latter does not ensure the oracle property of the resulting estimator.

For fixed values of $h$ and $\lambda$, the separation penalty (SP) estimator is defined as the minimizer $\big(\widehat{\beta}^{\lambda},\widehat{\gamma}_{(1)}^{\lambda}\big)$ of $\textsc{psis}_{\textsc{sp}}\big(\beta,\gamma_{(1)};\lambda\big)$.
The clustering estimator is given by $\widehat{\mathcal{G}}^{\lambda} = (\widehat{\mathcal{G}}^{\lambda}_1, \dots, \widehat{\mathcal{G}}^{\lambda}_{k})$, where the $i$th data point is assigned to $\widehat{\mathcal{G}}^{\lambda}_{\ell}$ if $\ell\in\argmin_{\ell_1} \| \widehat{\beta}^{\lambda}_{(1)i} - \widehat{\gamma}^{\lambda}_{(1)\ell_1} \|_1$ for $i=1,\dots, n$ and $\ell\in\mathcal{C}$.
We define the objective function $\textsc{psis}\big(\gamma,h;\mathcal{G}\big)$ as
\begin{align}
\textsc{psis}\big(\gamma,h;\mathcal{G}\big)=  \frac{1}{2}\sum_{i=1}^n \int \bigg(I(Y_i \leq y) -\sum^{k}_{\ell=1}I(i\in \mathcal{G}_{\ell}) \widehat{G}^{-i}_h\big(y, \gamma^{\top}_{\ell}Z_i; \mathcal{G},\gamma\big) \bigg)^2 d\widehat{F}(y),\label{fn:psis1}
\end{align}
where
\begin{align}
\widehat{G}_{h}(y,v; \mathcal{G},\gamma) =  \frac{\sum_{i=1}^n I(Y_i \leq y)\sum^{k}_{\ell=1}I\big(i\in \mathcal{G}_{\ell}\big) K_{2,h}\big(\gamma^{\top}_{\ell}Z_{i} - v\big)}{\sum_{i=1}^n \sum^{k}_{\ell=1}I\big(i\in \mathcal{G}_{\ell}\big) K_{2,h}\big(\gamma^{\top}_{\ell}Z_i - v\big)}. \label{fn:G_n_hat_subject}
\end{align}
To ensure the $\sqrt{n}$-consistency of the SP estimator, the bandwidth selector $\hat{h}$ is obtained through the score iteration procedure described in Section \ref{subsec:hsmethod}. The tuning parameter $\hat{\lambda}$ is chosen to minimize the criterion $\textsc{psis}(\widehat{\gamma}^{\lambda}, \tilde{h}^{\lambda}; \widehat{\mathcal{G}}^{\lambda})$, with $\tilde{h}^{\lambda}\in\argmin_{h}\textsc{psis}(\widehat{\gamma}^{\lambda}, h; \widehat{\mathcal{G}}^{\lambda})$, over a pre-specified grid of positive values.

Let $\widehat{\gamma}$ denote the minimizer of $\textsc{psis}\big(\gamma,h;\mathcal{G}^{\text{o}}\big)$ and define
 the oracle estimator of $\beta^{\text{o}}$ as $\widehat{\beta}^{or} = \sum_{\ell = 1}^k (I(1 \in \mathcal{G}^{\text{o}}_{\ell}), \dots, I(n \in \mathcal{G}^{\text{o}}_{\ell}) )^{\top} \otimes \widehat{\gamma}_{\ell}$, where $I(\cdot)$ is the indicator function and $\otimes$ denotes the Kronecker product. We denote $\theta_{*}$ as the unknown components of a generic parameter vector $\theta$ and
 $W$ as $(I(C=1), \dots, I(C=k))^{\top} \otimes Z$ with support $\mathcal{W}$, $G^{[m]}\big(y, w, \gamma^{\top}w\big)$ as the uniformly convergent function of $\partial_{\gamma_{*}}^m\widehat{G}_h\big(y, \gamma^{\top}w; \gamma\big)$ for $m =0, 1, 2$, and
$\Gamma$ as the parameter space of $\gamma_{*}$, with the true parameter $\gamma^{\text{o}}_{*}$ lying in its interior. The following theorem establishes the oracle property of $(\widehat{\beta}^{\lambda},\widehat{\gamma}^{\lambda})$:

\begin{thm} \label{Thm3.1}
Under assumptions {A1}--{A7},
\begin{align*}
P\Big(\widehat{\beta}^{\lambda}= \widehat{\beta}^{or}\Big) {\longrightarrow} 1\text{ and }P\Big(\widehat{\gamma}_{(1)}^{\lambda} = \widehat{\gamma}_{(1)}\Big) {\longrightarrow} 1\text{ as }n {\longrightarrow} \infty.
\end{align*} 
\begin{proof}
See Appendix \ref{sec:A.1}.
\end{proof}
\end{thm}
\noindent Theorem \ref{Thm3.1} implies that $\widehat{\mathcal{G}}^{\lambda}$ consistently recovers $\mathcal{G}^{\text{o}}$ as the sample size increases. Let $\widehat{\beta}$, $\widehat{\gamma}$, $\widehat{\mathcal{G}}$, and $\tilde{h}$ denote $\widehat{\beta}^{\hat{\lambda}}$, $\widehat{\gamma}^{\hat{\lambda}}$, $\widehat{\mathcal{G}}^{\hat{\lambda}}$, $\tilde{h}^{\hat{\lambda}}$, respectively. The asymptotic normality of $\widehat{\gamma}_{*}$ follows from that of $\widehat{\gamma}_{*}$ and is given by
\begin{align}
\sqrt{n}\big(\widehat{\gamma}_{*} - \gamma^{\text{o}}_{*}\big)\stackrel{d}{\longrightarrow} N\big(0,V_{2}^{-1}V_{1}V_{2}^{-1}\big) \text{ as }n \longrightarrow \infty, \label{normality:psispe}
\end{align}
where $V_{1}= E[(\int (I(Y\leq y)-G(y,\gamma^{\text{o}\top}W))G^{[1]}(y, W, \gamma^{\text{o}\top}W)dF(y))^{\otimes 2} \big]$, $V_{2}= E[\int ( G^{[1]}\big(y, W,\gamma^{\text{o}\top}$ $W))^{\otimes 2}dF(y)]$, and $(\cdot)^{\otimes 2}$ denotes the second-order Kronecker power. The asymptotic variance of $\widehat{\gamma}_{*}$ is estimated by $\widehat{V}_{2}^{-1}\widehat{V}_{1}\widehat{V}_{2}^{-1}$, with
\begin{align*}
\widehat{V}_{1}=&\frac{1}{n} \sum_{i=1}^n \sum^{k}_{\ell=1}I(i\in \widehat{\mathcal{G}}_{\ell}) \bigg[\int \Big( I(Y_i \leq y)  - \widehat{G}^{-i}_{\tilde{h}}\big(y, \widehat{\gamma}^{\top}_{\ell}Z_i; \widehat{\mathcal{G}},\widehat{\gamma}\big)  \Big)  \partial_{\gamma_{*}} \widehat{G}^{-i}_{\tilde{h}}\big(y, \widehat{\gamma}^{\top}_{\ell}Z_i; \widehat{\mathcal{G}},\widehat{\gamma}\big) d\widehat{F}(y)\bigg]^{\otimes 2}\text{ and}\\
\widehat{V}_{2}=& \frac{1}{n} \sum_{i=1}^n \sum^{k}_{\ell=1}I(i\in \widehat{\mathcal{G}}_{\ell}) \int \Big(\partial_{\gamma_{*}} \widehat{G}^{-i}_{\tilde{h}}\big(y, \widehat{\gamma}^{\top}_{\ell}Z_i; \widehat{\mathcal{G}},\widehat{\gamma}\big) \Big)^{\otimes 2}d\widehat{F}(y).
\end{align*}

\begin{rem}
\label{rem1}
An alternative to the PSIS in (\ref{fn:psisp}) is the negative log pseudo-likelihood for the response variable. However, this method is more sensitive to extreme values in the density estimates and requires more restrictive bandwidth conditions when applied to continuous outcomes.
The SP term may also be replaced by a pairwise fusion penalty (FP) of the form $\sum_{i < j} p(\| \beta_{(1)i} - \beta_{(1)j} \|; \lambda)$, where  $p(t, \lambda)$ is a non-decreasing and concave function. While SP has a computational complexity of order $O(kn)$, FP entails higher-order $O(n^2)$. Numerical results suggest that the two methods yield comparable empirical performance.
\end{rem}

\end{subsection}

\begin{subsection}{Estimation for the Cluster Membership Model}
\label{subsec:GMestimation}

Define $\mathcal{K}_{r,h_{d}}(u)=\prod^{d}_{j=1}K_{r,h_{d}}(u_{j})$
for $r > \max\{(4+d)/4,(1+d)/2\}$.
Under the SDR-based cluster membership model, the conditional expectation
$E[I(C=\ell)|A^{\top}_{d}X=u]$ coincides with the cluster membership probability $\pi_{\ell}(u)$. Given $\mathcal{G}$ and $A_{d}$,  
$\pi_{\ell}(u)$ is estimated by
\begin{align}
\widehat{\pi}_{\ell,h_{d}}(u;\mathcal{G},A_{d})=\frac{\sum^{n}_{i=1}I\big(i\in \mathcal{G}_{\ell}\big)\mathcal{K}_{r,h_{d}}\big(A^{\top}_{d}X_{i}-u\big)}{\sum^{n}_{i=1}\mathcal{K}_{r,h_{d}}\big(A^{\top}_{d}X_{i}-u\big)},\ell \in \mathcal{C}.\label{fn:est1_GM}
\end{align}
For fixed values of $d$ and $h_{d}$, the pseudo least squares (PLS) estimator $\widehat{A}_{d}$ is defined as the minimizer of the following pseudo sum of squares: 
\begin{align}
\textsc{pss}(A_{d},h_{d})=\frac{1}{2}\sum^{n}_{i=1}\sum^{k-1}_{\ell=1}\Big(I\big(i\in \widehat{\mathcal{G}}_{\ell}\big)- \widehat{\pi}^{-i}_{\ell,h_{d}}\big(A^{\top}_{d}X_{i};\widehat{\mathcal{G}},A_{d}\big)\Big)^{2},  \label{fn:PSS_l}
\end{align}
where $\widehat{\pi}^{-i}_{\ell,h_{d}}(u;\widehat{\mathcal{G}},A_{d})$ is computed as in (\ref{fn:est1_GM}), excluding the $i$th observation.
As noted in Remark \ref{rem1}, the pseudo maximum likelihood estimation may be unstable when the estimated cluster membership probabilities approach zero or take negative values.

To estimate the true basis matrix $A^{\text{o}}$ and its structural dimension $d^{\text{o}}$, we first compute the bandwidth
$\tilde{h}_{d}\in\arg\min_{h_{d}} \textsc{pss}(\widehat{A}_{d},h_{d})$, which achieves the convergence rate $O_{p}(n^{-1/(2r+d)})$. The estimator $\widehat{d}$ of $d^{\text{o}}$ is subsequently defined as the minimizer of the criterion
\begin{align}
\textsc{pss}(d) =\frac{2}{n}\textsc{pss}\Big(\widehat{A}_{d},\tilde{h}_{d}\Big)+\frac{\log n}{2n^{\frac{2r}{2r+d}}}\text{ for }d\geq 1.
 \label{fn:PSS_d}
\end{align}
The resulting estimator $\widehat{A}$ of $A^{\text{o}}$ is obtained accordingly. In practice, the procedure is implemented via a forward search over $d$, in which  $\textsc{pss}(d)$ is evaluated sequentially, starting from $d=1$ until a local minimum is reached.
For $\hat{d} = 1$, a second-order kernel function is employed for $1 \leq d \leq 2$, while a fourth-order kernel is used for $1 \leq d \leq 6$ when $2 \leq \hat{d} \leq 5$.

For a generic matrix $A$, define its projection matrix as $P_{A}=A(A^{\top}A)^{-1}A^{\top}$. Let $\pi^{[m]}_{\ell}({x}, A^{\top}_{d}x)$ denote the uniformly convergent function of $\partial^{m}_{\text{vec}(A_{d*})}\widehat{\pi}_{\ell,h_{d}}(A^{\top}_{d}x;\widehat{\mathcal{G}},A_{d})$ for $m=0, 1, 2$, where $\text{vec}(\cdot)$ denotes the column-stacking vectorization operator, with $\pi^{[0]}_{\ell}({x},A^{\text{o} \top}x) = \pi_{\ell}(A^{\text{o} \top}x)$ for $\ell\in\mathcal{C}$.
The following theorem establishes the consistency of $\hat{d}$ to $d^{\text{o}}$ and the asymptotic normality of $\widehat{A}$:
\begin{thm} 
\label{Thm3.3_1}
Under assumptions {A1}--{A7} and {B1}--{B5}, 
\begin{align*}
\hat{d}\stackrel{p}{\longrightarrow}d^{\text{o}} \text{ and }\sqrt{n}\text{vec}\big(P_{\widehat{A}}- P_{A^{\text{o}}}\big)\stackrel{d}{\longrightarrow} N(0,\Sigma) \text{ as }n \longrightarrow \infty,
\end{align*}
where 
\begin{align*}
&\Sigma=\Big[\big(\text{vec}\big((I_{p}-P_{A^{\text{o}}})(I_{d^{\text{o}}},V^{\top}_{0})^{\top}\big)(A^{\text{o}\top}A^{\text{o}})^{-1}A^{\text{o}\top}
+A^{\text{o}}(A^{\text{o} \top}A^{\text{o}})^{-1}(I_{d^{\text{o}}},V^{\top}_{0})(I_{p}-P_{A^{\text{o}}})\big)^{\otimes 2}\Big],\text{ with}\\
&\text{vec}(V_{0})=\big(\sum^{k-1}_{\ell=1}E\big[ \big(\pi^{[1]}_{\ell}({X}, A^{\text{o} \top}X)\big)^{\otimes 2}\big]\big)^{-1}
\big(\sum^{k-1}_{\ell=1}\big(I(C=\ell)-\pi_{\ell}\big({X}, A^{\text{o}\top}X\big)\big)\pi^{[1]}_{\ell}\big({X}, A^{\text{o}\top}X\big)\big).
\end{align*}
\begin{proof}
See Appendix \ref{sec:A.3.1}.
\end{proof}
\end{thm}
\begin{rem}
Existing SDR methods for categorical responses, including those by
\cite{shin2014probability}, \cite{yao2019covariate}, and \cite{zhang2020maximum}, are typically limited to binary responses or continuous covariates. In contrast, the proposed method recovers the central subspace while accommodating certain discrete and categorical covariates, provided the indices meet the continuity condition.
\end{rem}

\end{subsection}

\begin{subsection}{Optimal Clustering and Refined Estimation}
\label{subsec:optimal_estimation}

Let $C^{\text{o}}$ denote the true cluster label associated with a new observation $(X_{0},Y_{0})$. Define $\widehat{C}$ as the Bayes classifier specified by
$\widehat{C}=\ell_{0}$ if $\pi(\ell_{0}|X_{0},Y_{0})=\max_{\ell}\pi(\ell |X_{0},Y_{0}), \ell_{0}\in\mathcal{C}$.
For an arbitrary classifier $\bar{C}\in\mathcal{C}$ constructed from $(X_{0},Y_{0})$,
the probability of correct membership is
\begin{align}
P(\bar{C}=C^{\text{o}})=E\bigg[\sum^{k}_{\ell=1} I(\bar{C}=\ell)E[I(C^{\text{o}}=\ell)|X_{0},Y_{0}]\bigg]= E\bigg[\sum^{k}_{\ell=1}I(\bar{C}=\ell)\pi(\ell|X_{0},Y_{0})\bigg].\label{pcc}
\end{align}
By the definition of the Bayes classifier, it holds that
\begin{align}
\pi (\ell_{1} |X_{0},Y_{0})-\pi(\ell_{2} | X_{0},Y_{0})\geq 0 \text{ for } \widehat{C}=\ell_{1} \text{ and } \bar{C}=\ell_{2},\text{ with } \ell_{1}, \ell_{2} \in \mathcal{C}. \label{prob_ineq}
\end{align}
Combining (\ref{pcc}) and (\ref{prob_ineq}), the classifier $\widehat{C}$ is optimal in the sense that 
$P(\widehat{C}=C^{\text{o}})\geq P(\bar{C}=C^{\text{o}})$.

In the clusterwise index density framework associated with the CID model, we estimate $g(y,v)$ for continuous responses by
\begin{align}
&\widehat{g}_{\hat{h}_{1},\hat{h}_{2}}\big(y,v;\widehat{\mathcal{G}},\widehat{\gamma}\big)=\int K_{2,\hat{h}_{1}}(y^{*}-y)d_{y^{*}}\widehat{G}_{\hat{h}_{2}}\big(y^{*},v;\widehat{\mathcal{G}},\widehat{\gamma}\big),
\label{gestimate}
\end{align}
with $\big(\hat{h}_{1},\hat{h}_{2}\big)=\argmax_{\{h_{1},h_{2}\}}\sum^{n}_{i=1} \sum^{k}_{\ell =1}I\big(i\in\widehat{\mathcal{G}}_{\ell} \big)\log \widehat{g}_{h_{1},h_{2}}\big(Y_{i},\widehat{\gamma}^{\top}_{\ell}Z_{i};\widehat{\mathcal{G}},\widehat{\gamma}\big)$.
This data-driven selection ensures optimal asymptotic convergence of $\widehat{g}_{\hat{h}_{1},\hat{h}_{2}}\big(y,v;\widehat{\mathcal{G}},\widehat{\gamma}\big)$ regarding the mean integrated squared error. The resulting estimator enables estimation of the posterior cluster probabilities in (\ref{fn:posterior}), given by
\begin{align}
\widehat{\pi}(\ell |x, y) = \frac{ \widehat{g}_{\hat{h}_{1},\hat{h}_{2}}\big(y,\widehat{\gamma}^{\top}_{\ell}z;\widehat{\mathcal{G}},\widehat{\gamma}\big)\widehat{\pi}_{\ell,\tilde{h}_{d}}\big(\widehat{A}^{\top}x;\widehat{\mathcal{G}},\widehat{A}\big)}{\sum_{\ell = 1}^{k}  \widehat{g}_{\hat{h}_{1},\hat{h}_{2}}\big(y,\widehat{\gamma}^{\top}_{\ell}z;\widehat{\mathcal{G}},\widehat{\gamma}\big) \widehat{\pi}_{\ell,\tilde{h}_{d}}\big(\widehat{A}^{\top}x; \widehat{\mathcal{G}},\widehat{A}\big)}, (x,y,\ell) \in \mathcal{X} \times \mathcal{Y} \times \mathcal{C}. \label{fn:posterior_estimate} 
\end{align}
For discrete responses, the posterior cluster probabilities are estimated by $\widehat{g}_{\tilde{h}}\big(y,v;\widehat{\mathcal{G}},\widehat{\gamma}\big)$, derived from $\widehat{G}_{\tilde{h}}\big(y,v;\widehat{\mathcal{G}},\widehat{\gamma}\big)$.
Given these estimators, the optimal clustering rule assigns $(X_{0},Y_{0})$ to Cluster $\ell_{0}$ if
$\ell_{0} \in \argmax_{\ell } \widehat{\pi}(\ell |X_{0}, Y_{0})$.
Applying this clustering rule to the sample $\{(X_{i},Y_{i})\}^{n}_{i=1}$ yields an updated clustering estimator $\widetilde{\mathcal{G}}$, leading to the refined SP method. 
The cluster index coefficient estimator $\widehat{\gamma}$ is updated to $\widetilde{\gamma}$, obtained by minimizing 
$\textsc{psis}\big(\gamma, \hat{h}; \widetilde{\mathcal{G}}\big)$.
Substituting $\widetilde{\mathcal{G}}$ for $\widehat{\mathcal{G}}$ in (\ref{fn:est1_GM})--(\ref{fn:PSS_d})
yields the refined estimator $\widetilde{A}$ of $\widehat{A}$. The posterior cluster probabilities are then re-estimated as

\vspace{-0.1in}
\begin{align}
&\widetilde{\pi}(\ell |x, y) = \frac{ \widehat{g}_{\hat{h}_{1},\hat{h}_{2}}\big(y,\widetilde{\gamma}^{\top}_{\ell}z;\widetilde{\mathcal{G}},\widetilde{\gamma}\big)\widehat{\pi}_{\ell,\tilde{h}_{d}}\Big(\widetilde{A}^{\top}x; \widetilde{\mathcal{G}},\widetilde{A}\Big)}{\sum_{\ell = 1}^{k}   \widehat{g}_{\hat{h}_{1},\hat{h}_{2}}\big(y,\widetilde{\gamma}^{\top}_{\ell}z;\widetilde{\mathcal{G}},\widetilde{\gamma}\big)\widehat{\pi}_{\ell,\tilde{h}_{d}}\Big(\widetilde{A}^{\top}x;\widetilde{\mathcal{G}},\widetilde{A}\Big)} \text{ for } (x,y,\ell) \in \mathcal{X} \times \mathcal{Y} \times \mathcal{C}. \label{fn:posterior_estimate1} 
\end{align}
The clustering estimator $\widetilde{\mathcal{G}}$ may be further refined by iteratively updating cluster assignments according to the optimal clustering rule, with posterior probabilities re-estimated at each iteration. The procedure is repeated until convergence.

\end{subsection}

\begin{subsection}{Selecting the Number of Clusters} \label{subsec:SIC}

By adopting the artificial likelihood approach of \cite{wang2009wilcoxon} for the binary process
$\{ I(Y \leq y): y \in \mathcal{Y}\}$, we construct a semiparametric information criterion for selecting the number of clusters:
\begin{align}
\textsc{spic}_{1}(k) = \frac{2}{n} \textsc{psis}\big(\widetilde{\gamma},\tilde{h};\widetilde{\mathcal{G}}\big) + \frac{k \log n }{2 n^{\frac{4}{5}}}, k \geq 1. \label{fn:sic1}
\end{align}
The penalty term in (\ref{fn:sic1}) reflects the discrepancy between each candidate model's estimated and limiting distribution functions.
Motivated by the Bayesian information criterion of \cite{S1978} for linear models with normally distributed errors, we propose an alternative criterion based on the integrated log pseudo-likelihood for the binary process $\{I(Y \leq y): y \in \mathcal{Y} \}$:
\begin{align}
\textsc{spic}_{2}(k) = & -  \frac{1}{n} \sum^{n}_{i=1}  \int \Bigg\{ \sum_{\ell = 1}^{k}I\big(i \in \widetilde{\mathcal{G}}_{\ell}\big) \bigg[ 
I(Y_i \leq y) \log \widehat{G}^{-i}_{0,\tilde{h}}\big(y, \widetilde{\gamma}^{\top}_{\ell} Z_i; \widetilde{\mathcal{G}},\widetilde{\gamma}\big)\nonumber \\  
& +I(Y_i > y) \log \Big(1-  \widehat{G}^{-i}_{1,\tilde{h}}\big(y, \widetilde{\gamma}^{\top}_{\ell} Z_i; \widetilde{\mathcal{G}},\widetilde{\gamma}\big)\Big) \bigg]  \Bigg\} d \widehat{F}(y) + \frac{k \log n }{ n^{\frac{4}{5}}}, k \geq 1,  \label{fn:sic2}
\end{align}
where
$\widehat{G}^{-i}_{s,\tilde{h}}\big(y, \widetilde{\gamma}^{\top}_{\ell} Z_i; \widetilde{\mathcal{G}},\widetilde{\gamma}\big) =\widehat{G}^{-i}_{\tilde{h}}\big(y, \widetilde{\gamma}^{\top}_{\ell} Z_i; \widetilde{\mathcal{G}},\widetilde{\gamma}\big) +(-1)^{s}I\big( \widehat{G}^{-i}_{\tilde{h}}\big(y, \widetilde{\gamma}^{\top}_{\ell} Z_i; \widetilde{\mathcal{G}},\widetilde{\gamma}\big) = s \big)$
for $i=1,\dots, n$ and $s=0,1$. The penalty term in (\ref{fn:sic2}) is chosen to mitigate the risk of overestimating the number of clusters, a tendency observed under the alternative penalty $k \log n/2n^{4/5}$. 

Assume that 
$\lim_{n\longrightarrow\infty}P\big(\widetilde{\mathcal{G}}=\mathcal{G}^{\text{o}}\big)=1$ for $k \geq 1$.
Let $k^*$ denote the true number of clusters and define $\mathcal{G}^{*\text{o}} = \{\mathcal{G}^{*\text{o}}_{1}, \dots, \mathcal{G}^{*\text{o}}_{k^*}\}$ as the true clusters of the individual-based index set $\{1,\dots, n\}$. Let $C^{*}$ be the latent cluster variable corresponding to $\mathcal{G}^{*\text{o}}$, taking values in $\mathcal{C}^{*} = \{1, \dots, k^{*}\}$, and
set $\gamma^{*\text{o}} =(\gamma^{*\text{o}}_{1}, \dots, \gamma^{*\text{o}}_{k^*})$, with $\gamma^{*\text{o}}_{\ell}=\gamma^{\text{o}}_{\ell}$ for all $\ell\in\mathcal{C}^{*}$.
The criteria (\ref{fn:sic1}) and (\ref{fn:sic2}) yield consistent cluster number selection, as the minimizer $\hat{k}^{*}_{s}$ of $\textsc{SPIC}_s$ consistently estimates $k^{*}$ for $s=1,2$.

\begin{thm} \label{Thm3.4}
Under assumptions {A1}--{A4} and {A5} $($with $\mathcal{G}^{\text{o}} \in \mathcal{C}_{\mathcal{G}}=\{ \mathcal{G}^{\text{o}}: \mathcal{G}_{\ell}^{*\text{o}} = \cup_{\{\ell_1:\mathcal{G}^{\text{o}}_{\ell_1} \subseteq \mathcal{G}_{\ell}^{*\text{o}},\gamma^{\text{o}}_{\ell_1}=\gamma^{*\text{o}}_{\ell}\}}\mathcal{G}^{\text{o}}_{\ell_1}$ $\forall \ell\in\mathcal{C}^{*} \}$$)$, 
$\hat{k}^{*}_{s} \stackrel{p}{\longrightarrow} k^{*}$ as $n {\longrightarrow} \infty$, $s=1, 2$. 
 \begin{proof}
See Appendix \ref{sec:A.4}.
\end{proof}
 \end{thm}
\noindent Numerical results show that $\text{SPIC}_{2}$ tends to select more clusters than $\text{SPIC}_{1}$. While the log pseudo-likelihood may replace the integrated log pseudo-likelihood in (\ref{fn:sic2}) for the response variable, such a substitution leads to degraded performance due to instability from extreme values in the density estimates.
\end{subsection}

\end{section}

\begin{section}{Estimation Implementation}
\label{sec:estimation II}

This section introduces a heuristic initialization strategy, details the computational algorithm for the SP estimation, and describes the overall implementation procedure. Appendix \ref{spcode-flowcharts} provides the pseudocode for the SP estimation procedure and the flowcharts illustrating the proposed method and its variants.

\begin{subsection}{A Heuristic Initialization Strategy}

\label{subsec:hsmethod}

To minimize the objective function $\textsc{psis}_{\textsc{sp}}(\beta,\gamma;\lambda)$, we begin by fitting a single-index distribution model to approximate $F(y|x)$. Based on this model, we compute the residual processes
$\big\{ I(Y_i \leq y) - \widehat{G}^{-i}_{h}\big(y,\widehat{\gamma}_1^{\top} Z_i;\widehat{\gamma}_{1}\big): y \in \mathcal{Y}\big\}_{i = 1}^n$
or the residuals
$\big\{Y_i - \int y d_y\widehat{G}_{h}^{-i}\big(y, \widehat{\gamma}_1^{\top} Z_i;\widehat{\gamma}_{1}\big)\big\}_{i = 1}^n$,
 where $\widehat{G}_h^{-i}(y, v; \gamma_{1})$ is the leave-one-out version of 
 \begin{align}
\widehat{G}_{h}(y,v; \gamma_1) =\frac{\sum_{i=1}^n I(Y_i \leq y) K_{2,h}\big(\gamma^{\top}_{1}Z_{i} - v\big)}{\sum_{i=1}^n K_{2,h}\big(\gamma^{\top}_{1}Z_i - v\big)}, i=1,\dots,n.
 \label{fn:G_k_hat}
\end{align} 
The index coefficient vector $\gamma_{1}$ is estimated by a minimizer $\widehat{\gamma}_{1}$ of the objective function
\begin{align}
\textsc{psis}(\gamma_{1},h) = \frac{1}{2}\sum_{i=1}^n \int \Big(I(Y_i \leq y) - \widehat{G}^{-i}_h\big(y, \gamma^{\top}_{1}Z_i; \gamma_{1}\big) \Big)^2 d\widehat{F}(y).\label{fn:psis}
\end{align}
We then apply a distribution-free clustering method, such as  $k$-means or hierarchical clustering, to these quantities to construct an initial clustering estimator $\bar{\mathcal{G}}^{\ell} = \{\bar{\mathcal{G}}^{\ell}_{1}, \dots, \bar{\mathcal{G}}^{\ell}_{\ell} \}$ for $\ell$ ranging from 2 to $\bar{k}$, where $\bar{k}$ is set, for instance, to $\lfloor \sqrt{n/\log n} \rfloor$. Given this clustering estimator, we compute the corresponding pseudo least integrated squares (PLIS) estimator $\bar{\gamma}^{\ell}$ by minimizing the criterion in (\ref{fn:psis1}), with $\mathcal{G}$ replaced by $\bar{\mathcal{G}}^{\ell}$. The resulting subject index coefficient estimator is denoted by $\bar{\beta}^{\ell}$.

The PLIS estimator is computed through the following steps:
\begin{itemize}
\item[] Step 1. Obtain the preliminary estimator $\big(\bar{\gamma}_{1}, \bar{h}\big)$ by minimizing $\textsc{psis}(\gamma_{1},h)$.
\item[] Step 2. Set $\hat{h} = n^{3/80} \bar{h}$ and compute $\widehat{\gamma}_{1}$ as the minimizer of $\textsc{psis}\big(\gamma_{1},\hat{h}\big)$.
\end{itemize}
To enhance the asymptotic mean integrated squared error performance of $\widehat{G}_h\big(y,v; \widehat{\gamma}_{1}\big)$, $\hat{h}$  is updated to $\tilde{h}$, the minimizer of
$\textsc{psis}\big(\widehat{\gamma}_{1},h\big)$.
For numerical implementation, the initial value of $\gamma_{1}$ in Step 1 is obtained using inverse regression methods \citep[cf.][]{li1991sliced, cook1991sliced}. The PLIS estimator is subsequently computed via iterative score-based updates. 

Since $(\mathcal{G},\gamma)$ and $\beta$ contain equivalent information, the estimator $\widehat{G}_{h}(y,v;\beta)$ in (\ref{fn:G_n_hat}) can be written as $\widehat{G}_h\big(y, v; \mathcal{G},\gamma\big)$ in (\ref{fn:G_n_hat_subject}).
Although empirical results indicate that residual-based clustering (RC) estimates may serve as effective initial values for optimizing the objective function in the SP estimation, the RC procedure does not offer a direct benefit within the proposed model.
Instead, our method refines the preliminary clustering estimator $\bar{\mathcal{G}}^{\ell}$ and its associated cluster index coefficient estimator $\bar{\gamma}^{\ell}$, for $\ell=2,\dots,\bar{k}$, by reassigning data points according to their integrated squared deviations:
\begin{align*}
\textsc{pis}_i(\ell_1) = \int \big(I(Y_i \leq y) - \widehat{G}^{-i}_{\tilde{h}}\big(y,\bar{\gamma}^{\ell\top}_{\ell_{1}}Z_i; \bar{\mathcal{G}}^{\ell}, \bar{\gamma}^{\ell}\big) \big)^2 d\widehat{F}(y), i = 1, \dots, n, \ell_{1} = 1, \dots, \ell.
\end{align*}
A new clustering estimator $\check{\mathcal{G}}^{\ell}$ is obtained by assigning the $i$th data point to $\check{\mathcal{G}}^{\ell}_{\ell^\text{o}}$ if 
$\textsc{pis}_i(\ell^{\text{o}}) = \min_{\{1 \leq \ell_{1} \leq \ell\}}$ $\textsc{pis}_i(\ell_{1})$, $i = 1, \dots, n$.
The updated pair $\big(\check{\mathcal{G}}^{\ell}, \check{\gamma}^{\ell}\big)$ replaces
$\big(\bar{\mathcal{G}}^{\ell}, \bar{\gamma}^{\ell}\big)$ if
$\textsc{psis}\big(\check{\gamma}^{\ell},\hat{h}; \check{\mathcal{G}}^{\ell}\big) < \textsc{psis}\big(\bar{\gamma}^{\ell},\hat{h}; \bar{\mathcal{G}}^{\ell}\big)$;
otherwise, $\big(\check{\mathcal{G}}^{\ell},\check{\gamma}^{\ell}\big)$ is set to $\big(\bar{\mathcal{G}}^{\ell}, \bar{\gamma}^{\ell}\big)$, 
where $\check{\gamma}^{\ell}$ is obtained by minimizing $\textsc{psis}\big(\gamma^{\ell},\hat{h};\check{\mathcal{G}}^{\ell}\big)$ with respect to $\gamma^{\ell}$. 
This refinement procedure is iterated until the objective function attains its minimum or a specified convergence threshold is met. 
Based on the heuristic solution (HS) estimator $\big(\check{\mathcal{G}}^{\ell}, \check{\gamma}^{\ell}\big)$, the corresponding subject index coefficient estimators are denoted by $\check{\beta}^{\ell}$.

\begin{rem} To mitigate the computational burden of minimizing the objective function in the SP estimation, the refined estimation procedure in Section \ref{subsec:optimal_estimation} can be directly employed to update the HS estimator for each $\ell\in\{2,\dots,\bar{k}\}$. 
Although the resulting clustering estimator derived from the refined HS method may not consistently recover $\mathcal{G}^{\text{o}}$, our numerical experiments demonstrate that it substantially outperforms the SP method with RC-based initialization and performs comparably to, or only slightly worse than, the refined SP method under certain conditions.
\end{rem}
\end{subsection}

\begin{subsection}{Computational Algorithm for the SP Estimation}
 
\label{subsec:PSISSPmethod}

We recast the minimization of $\textsc{psis}_{\textsc{sp}}(\beta, \gamma_{(1)}; \lambda)$ as the constrained minimization of
\begin{align*}
\textsc{psis}_{\textsc{sp}}(\beta, \eta, \gamma_{(1)}; \lambda)= \frac{1}{2} \sum_{i=1}^n \int \big( I(Y_i \leq y)  - \widehat{G}_{\hat{h}}^{-i}\big(y,\beta_{i}^{\top} Z_i; \beta\big) \big)^2 d \widehat{F}(y) 
+ \lambda \sum^{n}_{i=1}\min_{\ell} \|\eta_{i} -\gamma_{(1)\ell}\|_{1} 
\end{align*}
subject to $\beta_{(1)} = \eta$. This formulation separates the objective function into two components: the first depends solely on $\beta$ and the second on $(\eta, \gamma_{(1)})$, enabling efficient optimization via the alternating direction method of multipliers (ADMM) algorithm \citep{boyd2011distributed}.
To implement this strategy, we adopt the augmented Lagrangian method, yielding the modified objective:
\begin{align}
\textsc{psis}_{\textsc{sp}}(\beta,\eta, \gamma_{(1)},\nu; \lambda) = \textsc{psis}_{\textsc{sp}}(\beta, \eta, \gamma_{(1)}; \lambda) + \nu^{\top} \big(\beta_{(1)} - \eta\big) + \frac{1}{2} \big\| \beta_{(1)} - \eta \big\|^2, \label{psisp_sep_2}
\end{align}
where $\nu$ is a vector of Lagrange multipliers. The minimization proceeds by iteratively updating the parameters $\beta$, $\eta$, 
$\gamma_{(1)}$, and $\nu$ for a given value of $\lambda$. At iteration $m\geq 0$, we denote the updated values by $\widehat{\beta}^{(m)}$, $\widehat{\eta}^{(m)}$, 
$\widehat{\gamma}^{(m)}$, and $\widehat{\nu}^{(m)}$. The algorithm terminates when the primal residual
$\big\|\widehat{\beta}^{(m+1)}_{(1)}-\widehat{\eta}^{(m+1)}\big\|$ falls below a predefined tolerance. 
Initialization is performed by setting $\widehat{\beta}^{(0)}$ to the HS estimates $\{\check{\beta}^{\ell}: \ell=k,\dots,\bar{k}\}$ in Section \ref{subsec:hsmethod}. The algorithm converges efficiently across a sequence of regularization parameters
$\lambda_{1} < \dots < \lambda_{J}$, where $\lambda_{1}>0$ and $\lambda_{J}=\max_{\{i,\ell_{1}\in\mathcal{C}\}}\|\check{\beta}^{\ell}_{(1)i}-\check{\gamma}^{k}_{(1)\ell_{1}}\|_{1}$.

To minimize (\ref{psisp_sep_2}) with respect to $(\eta, \gamma)$, we reformulate the objective function as a difference of convex functions in $\eta$ given $\{\beta_{(1)}, \nu, \gamma_{(1)}, \lambda\}$:
\begin{align}
\bigg( \frac{1}{2} \big\| \beta_{(1)} - \eta + \nu \big\|^2 + \lambda \sum^{n}_{i=1}\sum_{\ell = 1}^{k} \big\|\eta_{i} -\gamma_{(1)\ell}\big\|_{1} \bigg)  - \lambda \sum_{i = 1}^n \max_{\ell} \sum_{ \{ \ell_1: \ell_1 \neq \ell \}} \big\| {\eta}_i  - {\gamma}_{(1)\ell_1} \big\|_1 . \label{sep_eta_1}
\end{align}
Minimization of (\ref{sep_eta_1}) with respect to $\eta$ and $\gamma_{(1)}$ proceeds via a block coordinate descent algorithm,
wherein the update for $\eta$ employs a difference of convex functions programming \citep{an2005dc}.
Although it involves computational challenges due to the nonconvex structure of the problem, these difficulties are effectively addressed through a strategy similar to that used in the $\beta$-update step. 

The SP estimation procedure is implemented via the following iterative algorithm:
\begin{itemize}
\item \textbf{($\boldsymbol{\beta}$-minimization)} The update $\widehat{\beta}^{(m+1)}$ is obtained via the iterative scheme:
\begin{align}
 \beta_{(1)*}^{(s+1)}  = &\beta_{(1)*}^{(s)} - \Big( I_1\big(\beta^{(s)}\big) + I_{*} \Big)^{-1} \Big[ S_1\Big(\beta^{(s)}\Big) + \Big(\beta^{(s)}_{(1)*} - \widehat{\eta}^{(m)}_{*} + {\widehat{\nu}^{(m)}_{*}}\Big)\Big],  \nonumber\\
\beta_{(2)*}^{(s+1)} =& \beta_{(2)*}^{(s)} - I^{-1}_2 \Big(\beta^{(s)}\Big)S_2\Big(\beta^{(s)}\Big), s\geq 0,  \label{alg:SP_beta}
\end{align}
where $\beta^{(0)}= \widehat{\beta}^{(m)}$, $\widehat{\eta}^{(0)}= \widehat{\beta}^{(0)}_{(1)}$, $\widehat{\nu}^{(0)} = 0$, $I_{*}$ denotes the identity matrix of dimension $nq \times nq$ with rows and columns corresponding to known components in $\beta_{(1)}$ removed,
\begin{align*}
S_t(\beta)  = &- \sum_{i=1}^n \int \big( I(Y_i \leq y) - \widehat{G}_{\hat{h}}^{-i}\big(y,\beta_i^{\top}Z_i; \beta\big) \big)  \partial_{\beta_{(t)*}} \widehat{G}_{\hat{h}}^{-i}\big(y,\beta_i^{\top}Z_i; \beta\big)  d\widehat{F}(y),\\
I_t(\beta)  =& \sum_{i=1}^n \int \big( \partial_{\beta_{(t)*}} \widehat{G}_{\hat{h}}^{-i}\big(y,\beta_i^{\top}Z_i; \beta\big) \big)^{{\otimes}2} d\widehat{F}(y), t=1,2.
\end{align*}
\item \textbf{($\boldsymbol{(\eta,\gamma_{(1)})}$-minimization)} The updates $\widehat{\eta}^{(m+1)}$ and $\widehat{\gamma}^{(m+1)}_{(1)}$ are computed via the iterative updates 
\begin{align}
&\eta^{(s+1)}=   \argmin_{\eta} \Big[ \frac{1}{2} \Big\| \widehat{\beta}^{(m+1)}_{(1)} - \eta+ \widehat{\nu}^{(m)} \Big\|^2 + \lambda \sum^{n}_{i=1}\sum_{\ell = 1}^{k} \Big\|\eta_{i} - \gamma^{(s)}_{(1)\ell}\Big\|_{1}- \lambda \partial_{\eta}^{\top} S\Big(\eta^{(s)}, \gamma^{(s)}_{(1)} \Big) \big(\eta- \eta^{(s)}\big) \Big]  \nonumber \\
&\gamma^{(s+1)}_{(1) \ell} =  \bigg( \median_{\big\{i: \ell \in \argmin\limits_{\ell_1} \big\| {\eta}^{(s+1)}_{i} - {\gamma}^{(s)}_{(1)\ell_1}  \big\|_1 \big\} } {\eta}^{(s+1)}_{ij} \bigg),\, \ell \in \mathcal{C}, s\geq 0, \label{alg:SP_gamma}
\end{align}
with $\big(\eta^{(0)},\gamma^{(0)}_{(1)} \big) =\big(\widehat{\eta}^{(m)},\widehat{\gamma}^{(m)}_{(1)} \big)$ and
$\partial_{\eta_{ij}} S (\eta, \gamma_{(1)} ) = \sum_{\big\{\ell: \ell \neq \argmax\limits_{\ell_1} \sum\limits_{ \{ \ell_2: \ell_2 \neq \ell_1 \} }\| \eta_{i} - \gamma_{(1) \ell_2}\|_1 \big\} }\text{sgn} \big(\eta_{ij} - \gamma_{(1) \ell j} \big)$.

The update $\eta^{(s+1)}$ is computed via the following iterative subroutine, indexed by $s_1\geq 0$, with $i=1,\dots, n$ and $\ell\in\mathcal{C}$:\begin{itemize}
  \item \textbf{($\boldsymbol{\eta}$-minimization)} 
  \begin{eqnarray}
  \hspace{-0.1in}\bar{\eta}^{(s_1+1)}=  \frac{1}{k+1} \Big(\widehat{\beta}^{(m+1)}_{(1)} + \widehat{\nu}^{(m)}+ \sum_{\ell = 1}^{k} \Big(\gamma^{(s)}_{(1)\ell} \otimes 1_{n} +  \delta_{\ell}^{(s_1)} 
   - u_{\ell}^{(s_1)}\Big)+ \lambda \partial_{\eta}^{\top} S\Big( \eta^{(s)}, \gamma^{(s)}_{(1)} \Big) \Big),  \label{alg:SP_eta}
  \end{eqnarray}
  where $1_{n}$ denotes the $n$-dimensional vector of ones, $\delta_{i\ell}^{(0)} = \eta_{i}^{(s)} - \gamma_{(1)\ell}^{(s)}$, and $u^{(0)}_{\ell}=0$.
  \item \textbf{($\boldsymbol{\delta}$-minimization)} 
  
  \begin{align}
  &\delta_{i \ell }^{(s_1+1)} = \text{diag}\left(\max \Bigg\{0, 1 - \frac{\lambda}{\big\| \big(\bar{\eta}_i^{(s_1+1)} - \gamma_{(1)\ell}^{(s)} + u_{i \ell}^{(s_1)}\big)_j \big\|}\Bigg\}\right) \Big(\bar{\eta}_{i}^{(s_1+1)} - \gamma_{(1)\ell}^{(s)} + u_{i \ell}^{(s_1)}\Big). \label{alg:SP_delta}
  \end{align}
\item \textbf{(dual variable update)} 
  \begin{align}
  u^{(s_1+1)}_{i \ell } = u^{(s_1)}_{i \ell } + \bar{\eta}^{(s_1+1)}_{i} - \gamma^{(s)}_{(1)\ell} - \delta^{(s_1+1)}_{i \ell}.  \label{alg:SP_u}
  \end{align}
  \end{itemize}
\item \textbf{(dual variable update)} The update $\widehat{\nu}^{(m+1)}$ is given by
\begin{align}
\widehat{\nu}^{(m+1)}= \widehat{\nu}^{(m)} + \widehat{\beta}^{(m+1)}_{(1)} - \widehat{\eta}^{(m+1)}. \label{alg:SP_nu}
\end{align}
\end{itemize}
Since $I_1(\beta)$ is at least positive semi-definite, $I_1(\beta) + I_{*}$ is guaranteed to be positive definite, ensuring the numerical stability of the $\beta$-minimization step. Notably, the update for $\beta_{(2)}^{(s+1)}$ is omitted when $\beta = \beta_{(1)}$. 

\end{subsection}

\begin{subsection}{Computational Procedure}
\label{subsec:comp_procedure}
 
We implement the proposed estimation procedure in the following steps:
\begin{description}
\item Step 1. Fit a single-index distribution model $G(y,\gamma^{\top}_{1}z)$ to approximate $F(y|x)$, estimating $G(y,v)$ and $\gamma_{1}$ using (\ref{fn:G_k_hat}) 
and (\ref{fn:psis}), respectively. 
\item Step 2. Construct an initial clustering estimate $\bar{\mathcal{G}}^{\ell}$ by applying a clustering method to $\big\{ I(Y_i \leq y) - \widehat{G}^{-i}_{\tilde{h}}\big(y,\widehat{\gamma}_1^{\top} Z_i;\widehat{\gamma}_{1}\big): y \in \mathcal{Y}\big\}_{i = 1}^n$ or 
$\big\{Y_i - \int y\, d_y\widehat{G}_{\tilde{h}}^{-i}\big(y, \widehat{\gamma}_1^{\top} Z_i;\widehat{\gamma}_{1}\big)\big\}_{i = 1}^n$, 
and obtain the corresponding cluster index coefficient estimate $\bar{\gamma}^{\ell}$ using (\ref{fn:psis1}), with $\mathcal{G}$ replaced by $\bar{\mathcal{G}}^{\ell}$, for each $\ell\in\{2,\dots, \bar{k}\}$.
\item Step 3.  Refine the RC estimate $\big(\bar{\mathcal{G}}^{\ell}, \bar{\gamma}^{\ell}\big)$ to obtain the HS estimate $\big(\check{\mathcal{G}}^{\ell}, \check{\gamma}^{\ell}\big)$ by
iteratively applying the refinement procedure for each $\ell\in\{2,\dots, \bar{k}\}$.
\item Step 4. Obtain the SP estimate $\big(\widehat{\mathcal{G}}, \widehat{\gamma}\big)$ using the algorithm in Section \ref{subsec:PSISSPmethod}, with initial values provided by the HS estimates $\big\{\big(\check{\mathcal{G}}^{\ell}, \check{\gamma}^{\ell}\big): \ell =k ,\dots, \bar{k}\big\}$.
\item Step 5. Estimate the posterior cluster probabilities in (\ref{fn:posterior_estimate}) by fitting the SDR-based model using (\ref{fn:est1_GM})--(\ref{fn:PSS_d}) and the clusterwise index density model using (\ref{gestimate}).
\item Step 6. Update $\widehat{\mathcal{G}}$ to $\widetilde{\mathcal{G}}$ by applying the optimal clustering rule in Section \ref{subsec:GMestimation}, and obtain the corresponding cluster index coefficient estimate $\widetilde{\gamma}$ using (\ref{fn:psis1}), with $\mathcal{G}$ replaced by $\widetilde{\mathcal{G}}$.
\item Step 7. Select the number of clusters as $\hat{k}^{*}_{1}$ using $\text{SPIC}_{1}$ in (\ref{fn:sic1}), or as $\hat{k}^{*}_{2}$ 
using $\text{SPIC}_{2}$ in (\ref{fn:sic2}). 
\end{description} 

\end{subsection}

\end{section}

\begin{section}{Monte Carlo Simulations} 
\label{sec:simulation}
We conducted simulation studies with sample sizes of 100, 200, 300, and 400 to assess the proposed method's finite-sample performance compared to its variants. Each scenario was replicated 500 times to ensure result stability. Supplementary tables are provided in Appendix \ref{table-figure}.

\begin{subsection}{Simulation Design and Performance Metrics} \label{sec:simulation_settings}
Data were generated from the CID model to assess the proposed methodology.
The simulation employed the following bivariate function: 
\begin{align*}
G(y,v) = \Phi\bigg(\frac{y - (v+5)^\frac{3}{5}}{\sigma}\bigg),
\end{align*}
where $\Phi(u)$ denotes the standard normal cumulative distribution function, and $\sigma>0$ is a scale parameter. The cluster-invariant coefficients were set to $\gamma_{(2)} = (1,-1,1,-1,1)^{\top}/\sqrt{5}$ across all simulation settings. The first four components of $Z_{(2)}$ were drawn from a multivariate normal distribution $N_{4}(0_{4}, 0.8 I_4 + 0.2 1_4 1_4^{\top})$, and the fifth component was independently drawn from $\{-0.5,0.5\}$ with equal probability. 

The first simulation setting was designed to reflect characteristics observed in a real estate valuation study. The number of clusters was set to $k = 2$, and the covariate $Z_{(1)} = 1$ served as the baseline for cluster-specific intercepts. Three CID models were considered, differing in the cluster-specific coefficients and the scale parameter: 
\begin{align*}
\text{M1.} ~ \gamma_{(1)} = (0, 2.5)^{\top},\sigma = 0.1;~
\text{M2.} ~ \gamma_{(1)}= (0,1.5)^{\top},\sigma = 0.1; ~
\text{M3.}~  \gamma_{(1)}= (0, 2.5)^{\top},\sigma = 0.15.
\end{align*}
Relative to model {M1}, model {M2} induces smaller differences in the cluster-specific intercepts, whereas model {M3} introduces greater variability in the response variable. Consequently, {M1} yields the clearest separation between the conditional distributions of the two clusters. Two cluster membership mechanisms were considered:
\begin{align*}
&\text{Covariate-independent membership with } (\pi_1,\pi_2) = (0.7,0.3); \\
&\text{Covariate-dependent membership via } \pi(1|x)=0.25+\frac{2.5\phi\big(2.8\alpha^{\top}_{1}x\big)}{1+\phi\big(2.8\alpha^{\top}_{1}x\big)}, \text{ with }\alpha_1 = (1,0,0,0,0.29)^{\top}, 
\end{align*}
where $\phi(u)$ denotes the standard normal density function.

In the second simulation scenario, the structure from Scenario I was extended to accommodate $k = 3$ clusters. The cluster-specific intercepts and scale parameter in the CID model were set to
\begin{align*}
\text{M4.}~ \gamma_{(1)}=(0, 2.5, 5)^{\top},\sigma=0.125. 
\end{align*}
Two cluster membership models were considered:
\begin{align*}
&\text{Covariate-independent membership with } \pi_1 = \pi_2 = \pi_3; \\
&\text{Covariate-dependent membership via } \pi( \ell | x) = \frac{\exp\big(\alpha_{\ell}^{*\top}z\big)}{1+\sum_{\ell_1 = 1}^{2}\exp\big(\alpha_{\ell_1}^{*\top}z\big)},\, \ell=1,2,\text{ with}\\
& \alpha^{*}_1 = (0,-0.12,0.12,-0.12,0.12,-0.12)^{\top}\text{ and }\alpha^{*}_2 = \alpha^{*}_{1}/2.
\end{align*}

The simulated settings in the third scenario were motivated by data from the Cleveland Heart Disease study and the AIDS Clinical Trials Group study. The number of clusters was set to $k = 2$, and cluster-specific intercepts and treatment effects were incorporated. The covariates $Z_{(1)}$ represented treatment groups and were drawn from a trinomial distribution with equal probabilities over the vectors $(1, 0, 0)^\top$, $(1, 1, 0)^\top$, and $(1, 0, 1)^\top$.
The cluster-specific coefficients and scale parameter in the CID model were specified as
\begin{align*}
\text{M5.}~ \gamma_{(1)}=(0, 0.1, 0.1, 1.5,0.3,1)^{\top},\sigma=0.1. 
\end{align*}
Two cluster membership mechanisms were considered:
\begin{align*}
&\text{Covariate-independent membership with } (\pi_1 , \pi_2) = (0.5,0.5); \\
&\text{Covariate-dependent membership via } \pi(1| x) = \frac{\exp\big(\alpha_{1}^{*\top}z\big)}{1+\exp\big(\alpha_{1}^{*\top}z\big)},\text{ with }
\alpha^{*}_1 = (0,0.1,0.1,0,0,0,0,0)^{\top}.
\end{align*}

We assessed the performance of several estimation strategies, including the RC and HS methods in Section \ref{subsec:hsmethod}, the SP method in Section \ref{subsec:PSISP}, and the refined SP (RSP) methods in Section \ref{subsec:optimal_estimation}. Results for the SP method with RC-based initialization, the fusion penalty (FP) method---shown to yield comparable performance to the SP method---and the refined HS (RHS) method are presented in the supplementary tables. To quantify the accuracy of the recovered cluster assignments, we computed the Rand index (RI) \cite[]{rand1971objective} between the true clusters and those induced by each estimator.
The performance of a generic estimator $\breve{\theta}$ of the coefficients $\theta^{\text{o}}$ in the CID or SID model was assessed using the normalized root squared error (RSE)
$\sqrt{(\breve{\theta}-\theta^{\text{o}})^{\top}(\breve{\theta}-\theta^{\text{o}})/|\theta^{\text{o}} |}$. For a generic estimator $\breve{\pi}(x)$ of the cluster membership probabilities, performance was measured by $\sqrt{\sum^{n}_{i=1}\sum^{k-1}_{\ell=1}\big(\breve{\pi}_{\ell}(X_{i}) - \pi(\ell|X_{i})\big)^{2}/n}$.

\end{subsection}

\begin{subsection}{Estimator and Cluster Selection Performance} \label{sec: assessment}

Tables \ref{tab:RI} and \ref{tab:S_RI} summarize the clustering performance of the methods under consideration. The RC method does not directly improve the estimation of the true clusters. In Scenario I, where cluster membership is independent of covariates, the HS method achieves a higher RI than the SP method with RC-based initialization when $n=100$; however, the difference diminishes with increasing sample size, and the RI values become comparable or slightly lower for $n \geq 200$. When cluster membership depends on covariates, the HS method exhibits better performance. Across all settings, the SP method with HS-based initialization consistently outperforms the HS method and the SP method with RC-based initialization, though the gains are modest for $n \geq 200$. In Scenario II, the HS method substantially outperforms the RC method and the SP method with RC-based initialization. It achieves performance comparable to the SP method with HS-based initialization. A similar pattern emerges in Scenario III, where the RI values associated with the RC method and the SP method with RC-based initialization are markedly lower than those of the HS-based methods. The SP method with HS-based initialization slightly exceeds the HS method in RI value. Overall, these findings highlight the effectiveness of the SP method with HS-based initialization.

\begin{table}[htbp]
  \centering
  \caption{Means (standard deviations) of 500 RI values of the true clusters and clustering estimates from various methods under different CID and cluster membership models.}
      {\small
      \begin{adjustbox}{max width=\linewidth}
    \begin{tabular}{ccccccccccc}
    \toprule
    \multicolumn{2}{c}{Method} & RC    & HS    & SP   & RSP & & RC    & HS    & SP    & RSP \\
        \cmidrule(rl){3-6} \cmidrule(rl){8-11}  
    CID   & $n$     & \multicolumn{4}{c}{Covariate-independence} & & \multicolumn{4}{c}{Covariate-dependence}\\
    \midrule          
    M1    & 100   & 0.920 (0.0714) & 0.934 (0.0693) & 0.940 (0.0638) & 0.962 (0.0564) &       & 0.873 (0.1017) & 0.942 (0.0947) & 0.955 (0.0849) & 0.967  (0.0810)\\
          & 200   & 0.970 (0.0332) & 0.971 (0.0308) & 0.975 (0.0247) & 0.987 (0.0170) &       & 0.939 (0.0682) & 0.980 (0.0497) & 0.986 (0.0447)& 0.990 (0.0379) \\
          & 300   & 0.981 (0.0182) & 0.982 (0.0176)& 0.983 (0.0141)& 0.990 (0.0107)&       & 0.965 (0.0500)& 0.989 (0.0331)& 0.992 (0.0323)& 0.994 (0.0258)\\
          & 400   & 0.986 (0.0126)& 0.987 (0.0122)& 0.987 (0.0112)& 0.992 (0.0079)&       & 0.972 (0.0480)& 0.991 (0.0360)& 0.992 (0.0409)& 0.995 (0.0276)\\
           \midrule
    M2    & 100   & 0.847 (0.0909)& 0.862 (0.0876)& 0.871 (0.0829)& 0.895 (0.0784)&       & 0.809 (0.0981)& 0.892 (0.0813)& 0.903 (0.0700)& 0.918 (0.0683)\\
          & 200   & 0.911 (0.0537)& 0.915 (0.0486)& 0.921 (0.0436)& 0.942 (0.0335)&       & 0.885 (0.0638)& 0.938 (0.0371)& 0.944 (0.0313)& 0.952 (0.0287)\\
          & 300   & 0.930 (0.0407)& 0.931 (0.0393)& 0.938 (0.0328)& 0.955 (0.0219)&       & 0.912 (0.0529)& 0.951 (0.0252)& 0.957 (0.0202)& 0.960 (0.0200) \\
          & 400   & 0.939 (0.0361)& 0.940 (0.0354)& 0.946 (0.0221)& 0.958 (0.0174)&       & 0.921 (0.0432)& 0.955 (0.0203)& 0.959 (0.0164)& 0.961 (0.0157)\\
           \midrule
    M3    & 100   & 0.880 (0.0849)& 0.891 (0.0833)& 0.897 (0.0813)& 0.920 (0.0759)&       & 0.820 (0.1056) & 0.891 (0.1043)& 0.906 (0.0939)& 0.925 (0.0931)\\
          & 200   & 0.931 (0.0499)& 0.933 (0.0492)& 0.938 (0.0417)& 0.960 (0.0306)&       & 0.883 (0.0842)& 0.938 (0.0720)& 0.946 (0.0647) & 0.956 (0.0607)\\
          & 300   & 0.946 (0.0390)& 0.947 (0.0386)& 0.953 (0.0270)& 0.968 (0.0178) &       & 0.907 (0.0786)& 0.952 (0.0604)& 0.955 (0.0630)& 0.965 (0.0553)\\
          & 400   & 0.956 (0.0323)& 0.956 (0.0318)& 0.960 (0.0207)& 0.971 (0.0146)&       & 0.922 (0.0695)& 0.961 (0.0513) & 0.964 (0.0525)& 0.969 (0.0448)\\
          \midrule
    M4 & 100   & 0.823 (0.0817)& 0.885 (0.0879) & 0.889 (0.0850)& 0.902 (0.0885)&       & 0.821 (0.0832)& 0.889 (0.0838)& 0.894 (0.0805)& 0.921 (0.0806)\\
          & 200   & 0.890 (0.0616)& 0.953 (0.0474)& 0.956 (0.0428)& 0.973 (0.0336)&       & 0.894 (0.0609)& 0.957 (0.0445)& 0.961 (0.0414)& 0.977 (0.0311)\\
          & 300   & 0.924 (0.0443)& 0.972 (0.0270)& 0.976 (0.0214)& 0.983 (0.0142)&       & 0.923 (0.0526)& 0.976 (0.0291)& 0.976 (0.0245)& 0.985 (0.0157)\\
          & 400   & 0.933 (0.0481)& 0.980 (0.0161)& 0.980 (0.0164)& 0.986 (0.0086)&  & 0.936 (0.0356)& 0.981 (0.0151)& 0.982 (0.0136)& 0.987 (0.0085)\\
          \midrule
   M5 & 100   & 0.753 (0.1131) & 0.800 (0.1572)& 0.808 (0.1568)& 0.824 (0.1566)&       & 0.766 (0.1046)& 0.805 (0.1529)& 0.806 (0.1592)& 0.817 (0.1579)\\
          & 200   & 0.883 (0.0707)& 0.949 (0.0601)& 0.954 (0.0462)& 0.962 (0.0572)&       & 0.867 (0.0767) & 0.940 (0.0688)& 0.944 (0.0666)& 0.957 (0.0635)\\
          & 300   & 0.913 (0.0610)& 0.967 (0.0253)& 0.971 (0.0184)& 0.975 (0.0146)&       & 0.904 (0.0616)& 0.963 (0.0351)& 0.968 (0.0228)& 0.972 (0.0193)\\
          & 400   & 0.931 (0.0464)& 0.974 (0.0178)& 0.977 (0.0134)& 0.979 (0.0126)&  & 0.928 (0.0467)& 0.975 (0.0151)& 0.977 (0.0126)& 0.979 (0.0100)\\
          \bottomrule
    \end{tabular}%
    \end{adjustbox}
    }
  \label{tab:RI}%
\end{table}%

In Scenario I, the RHS method yields RI values slightly lower than or comparable to the RSP method, with both refined methods substantially outperforming the alternatives. Across all methods, higher RI values are consistently observed under the CID model {M1} compared to models {M2} and {M3}. In Scenario II, the RHS and RSP methods yield similar RI values and outperform the SP method with HS-based initialization. While differences in RI values are more evident for smaller sample sizes ($n\leq 200$), they diminish as $n\geq 300$, reflecting increased estimator stability with larger samples. For $n=400$, the clustering estimators obtained by the HS method, the SP method with HS-based initialization, and the RSP method closely approximate the true clusters. In Scenario III, the RHS method yields modest improvements over the SP method initialized with HS, yet it continues to underperform relative to the RSP method. These results highlight the incremental gains achieved through refinement procedures and their comparative advantage under specific settings. Overall, the decline in RI values across models is attributable to reduced separation between conditional distributions, increased response variability, and a higher number of cluster index coefficients, all of which have important practical
implications for clustering estimation. Tables \ref{tab:RI} and  \ref{tab:S_RI} also reveal a clear upward trend in RI values as sample size increases.

Tables \ref{tab:RMSE_b} and \ref{tab:S_RMSE_b} present the performance of the subject-specific coefficient estimators. The results demonstrate that the HS estimator consistently attains a smaller RSE than the RC estimator and the SP estimator with RC-based initialization. The patterns in RSE observed across methods in Scenario I are consistent with the corresponding RI values. In Scenario II, the SP estimator with HS-based initialization yields a slightly smaller RSE than the HS estimator, and its refined version achieves a marginal further reduction in RSE relative to the RHS estimator. The RSP estimator improves upon the SP estimator with HS-based initialization for $n\leq 300$, with a modest advantage persisting at $n=400$. In Scenario III, the SP method with HS-based initialization exhibits a slight improvement over the HS method at $n=100$, with the reduction in RSE becoming more pronounced as the sample size increases to $n\geq 200$. The RSE of the SP estimator with HS-based initialization exceeds that of the RHS estimator, which in turn is marginally larger than that of the RSP estimator. Across all scenarios, the RSE of each estimator consistently decreases as the sample size increases.

\begin{table}[htbp]
  \centering
  \caption{Means of 500 RSEs of subject-specific coefficient estimates across subjectwise representations corresponding to CID models {M1}--{M5} from various methods under different cluster membership models.}
    {\small
     \begin{adjustbox}{max width=0.6\linewidth}
         \begin{tabular}{ccccccccccc}
    \toprule
    \multicolumn{2}{c}{Method} & RC    & HS    & SP   & RSP & & RC    & HS    & SP    & RSP \\
        \cmidrule(rl){3-6} \cmidrule(rl){8-11} 
    CID   & $n$     & \multicolumn{4}{c}{Covariate-independence} & & \multicolumn{4}{c}{Covariate-dependence}\\
    \midrule  
    M1    & 100   & 0.559 & 0.509 & 0.477 & 0.394 &       & 0.848 & 0.658 & 0.594 & 0.554 \\
          & 200   & 0.340 & 0.330 & 0.310 & 0.241 &       & 0.693 & 0.571 & 0.543 & 0.460 \\
          & 300   & 0.266 & 0.260 & 0.239 & 0.189 &       & 0.520 & 0.412 & 0.393 & 0.366 \\
          & 400   & 0.220 & 0.215 & 0.209 & 0.166 &       & 0.489 & 0.410 & 0.381 & 0.308 \\
          \midrule
    M2    & 100   & 0.450 & 0.427 & 0.410 & 0.372 &       & 0.589 & 0.465 & 0.445 & 0.391 \\
          & 200   & 0.327 & 0.320 & 0.313 & 0.265 &       & 0.429 & 0.329 & 0.327 & 0.286 \\
          & 300   & 0.290 & 0.287 & 0.277 & 0.239 &       & 0.333 & 0.249 & 0.238 & 0.225 \\
          & 400   & 0.268 & 0.257 & 0.253 & 0.227 &       & 0.312 & 0.237 & 0.227 & 0.219 \\
          \midrule
    M3    & 100   & 0.714 & 0.687 & 0.672 & 0.602 &       & 1.018 & 0.853 & 0.807 & 0.712 \\
          & 200   & 0.498 & 0.492 & 0.478 & 0.390 &       & 0.853 & 0.712 & 0.689 & 0.639 \\
          & 300   & 0.426 & 0.423 & 0.414 & 0.340 &       & 0.707 & 0.558 & 0.546 & 0.477 \\
          & 400   & 0.385 & 0.385 & 0.370 & 0.324 &       & 0.703 & 0.554 & 0.468 & 0.398 \\
          \midrule
    M4    & 100   & 1.487 & 1.335 & 1.212 & 1.096 &  & 1.358 & 1.164 & 1.140 & 1.073 \\
          & 200   & 1.220 & 1.028 & 0.983 & 0.871 &       & 1.044 & 0.814 & 0.805 & 0.742 \\
          & 300   & 0.954 & 0.798 & 0.760 & 0.720 &  & 0.914 & 0.774 & 0.751 & 0.700 \\
          & 400   & 0.908 & 0.738 & 0.728 & 0.701 &       & 0.896 & 0.732 & 0.720 & 0.675 \\
          \midrule
    M5    & 100   & 0.992 & 0.918 & 0.910 & 0.886 &  & 0.788 & 0.731 & 0.717 & 0.704 \\
          & 200   & 0.433 & 0.351 & 0.266 & 0.250 &       & 0.458 & 0.382 & 0.327 & 0.299 \\
          & 300   & 0.353 & 0.295 & 0.211 & 0.202 &  & 0.342 & 0.266 & 0.222 & 0.210 \\
          & 400   & 0.325 & 0.274 & 0.170 & 0.156 &       & 0.281 & 0.216 & 0.150 & 0.141 \\
          \bottomrule
    \end{tabular}%
      \end{adjustbox}
   }
  \label{tab:RMSE_b}%
\end{table}%

Table \ref{tab:RMSE_g} compares the proposed estimators and their oracle counterparts for the cluster index coefficients and cluster membership probabilities. The RSEs of both estimators decrease consistently with increasing sample size, with the proposed and oracle estimators exhibiting convergence in magnitude. For each CID model, the RSE of the RSP estimator of the cluster index coefficients approaches that of the oracle estimator as the sample size increases. An exception arises in Scenario III, where substantially larger samples are needed for the RSEs of the proposed estimators to attain comparable levels. In contrast, the RSE of the refined pseudo least squares (RPLS) estimator for the cluster membership probabilities remains consistently close to that of the oracle estimator across all sample sizes.

\begin{table}[htbp]
  \centering
   \caption{Means of 500 RSEs of proposed and oracle estimates of cluster index coefficients and cluster membership probabilities across various CID and cluster membership models.}
  \setlength{\tabcolsep}{4pt} 
  \begin{subfigure}[t]{0.48\textwidth}
    \centering
    \caption*{Cluster Index Coefficients}
    \begin{adjustbox}{max width=0.9\linewidth}
      \begin{tabular}{cc cc c cc}
        \toprule
        \multicolumn{2}{c}{Method} & RSP & Oracle & & RSP & Oracle \\
        \cmidrule(lr){3-4} \cmidrule(lr){6-7}
        CID & $n$ & \multicolumn{2}{c}{Covariate-independence} & & \multicolumn{2}{c}{Covariate-dependence} \\
        \midrule
        M1 &  100 & 0.171 & 0.110 & & 0.292 & 0.198 \\
           &  200 & 0.110 & 0.072 & & 0.181 & 0.091 \\
           &  300 & 0.087 & 0.064 & & 0.160 & 0.090 \\
           &  400 & 0.053 & 0.054 & & 0.117 & 0.085 \\
        \midrule
        M2 &  100 & 0.129 & 0.106 & & 0.164 & 0.109 \\
           &  200 & 0.068 & 0.052 & & 0.119 & 0.052 \\
           &  300 & 0.045 & 0.043 & & 0.043 & 0.042 \\
           &  400 & 0.041 & 0.037 & & 0.042 & 0.037 \\
        \midrule
        M3 &  100 & 0.285 & 0.150 & & 0.400 & 0.215 \\
           &  200 & 0.144 & 0.104 & & 0.295 & 0.126 \\
           &  300 & 0.098 & 0.085 & & 0.190 & 0.120 \\
           &  400 & 0.091 & 0.080 & & 0.130 & 0.087 \\
        \midrule
        M4 &  100 & 0.735 & 0.563 & & 0.705 & 0.486 \\
           &  200 & 0.678 & 0.590 & & 0.534 & 0.449 \\
           &  300 & 0.648 & 0.554 & & 0.521 & 0.426 \\
           &  400 & 0.621 & 0.533 & & 0.498 & 0.401 \\
        \midrule
        M5 &  100 & 0.555 & 0.455 & & 0.579 & 0.449 \\
           &  200 & 0.294 & 0.224 & & 0.281 & 0.193 \\
           &  300 & 0.197 & 0.109 & & 0.196 & 0.097 \\
           &  400 & 0.165 & 0.093 & & 0.114 & 0.085 \\
        \bottomrule
      \end{tabular}
    \end{adjustbox}
  \end{subfigure}
  \hfill
  \begin{subfigure}[t]{0.48\textwidth}
    \centering
    \caption*{Cluster Membership Probabilities}
    \begin{adjustbox}{max width=0.9\linewidth}
      \begin{tabular}{cc cc c cc}
        \toprule
        \multicolumn{2}{c}{Method} & RPLS & Oracle & & RPLS & Oracle \\
        \cmidrule(lr){3-4} \cmidrule(lr){6-7}
        CID & $n$ & \multicolumn{2}{c}{Covariate-independence} & & \multicolumn{2}{c}{Covariate-dependence} \\
        \midrule
        M1 &  100 & 0.192 & 0.186 & & 0.284 & 0.280 \\
           &  200 & 0.138 & 0.136 & & 0.249 & 0.247 \\
           &  300 & 0.107 & 0.110 & & 0.225 & 0.225 \\
           &  400 & 0.099 & 0.098 & & 0.217 & 0.220 \\
        \midrule
        M2 &  100 & 0.216 & 0.186 & & 0.293 & 0.280 \\
           &  200 & 0.147 & 0.136 & & 0.255 & 0.247 \\
           &  300 & 0.112 & 0.110 & & 0.238 & 0.225 \\
           &  400 & 0.101 & 0.098 & & 0.224 & 0.220 \\
        \midrule
        M3 &  100 & 0.202 & 0.186 & & 0.292 & 0.280 \\
           &  200 & 0.143 & 0.136 & & 0.261 & 0.247 \\
           &  300 & 0.117 & 0.110 & & 0.240 & 0.225 \\
           &  400 & 0.101 & 0.098 & & 0.230 & 0.220 \\
        \midrule
        M4 &  100 & 0.244 & 0.236 & & 0.243 & 0.238 \\
           &  200 & 0.168 & 0.168 & & 0.170 & 0.165 \\
           &  300 & 0.142 & 0.139 & & 0.143 & 0.141 \\
           &  400 & 0.119 & 0.113 & & 0.123 & 0.122 \\
        \midrule
        M5 &  100 & 0.210 & 0.200 & & 0.207 & 0.203 \\
           &  200 & 0.158 & 0.158 & & 0.165 & 0.164 \\
           &  300 & 0.138 & 0.131 & & 0.137 & 0.142 \\
           &  400 & 0.118 & 0.117 & & 0.121 & 0.114 \\
        \bottomrule
      \end{tabular}
    \end{adjustbox}
  \end{subfigure}
  \label{tab:RMSE_g}
\end{table}

Tables \ref{tab:est_k} and \ref{tab:S_est_k} compare $\textsc{SPIC}_1$ and $\textsc{SPIC}_2$ in selecting the number of clusters. Overall, $\textsc{SPIC}_1$ tends to select fewer or comparable numbers of clusters than $\textsc{SPIC}_2$. For a given criterion, the HS method, the SP method with HS-based initialization, and their refined counterparts exhibit similar performance, all outperforming the RC method and the SP method with RC-based initialization. In Scenario I, $\textsc{SPIC}_2$ systematically overestimates the number of clusters in models M2 and M3. Although the refined estimation methods partially mitigate this overestimation, $\textsc{SPIC}_1$ generally provides more accurate cluster number estimates. When $n=100$, $\textsc{SPIC}_1$ displays a slight tendency toward underestimation, which is alleviated by the refined methods compared to other methods. As the sample size increases, $\textsc{SPIC}_1$ shows substantial improvement in estimation accuracy.
Nevertheless, under the CID model {M3}, $\textsc{SPIC}_2$ continues to overestimate the number of clusters even as the sample size increases. When conditional distributions exhibit limited separation, $\textsc{SPIC}_1$ is more susceptible to underestimating the number of clusters. The overestimation observed with $\textsc{SPIC}_2$ is especially pronounced when conditional distributions are only marginally distinct, the response variable exhibits high variability, or considerable heterogeneity exists within clusters, underscoring the primary factors driving this tendency.
In Scenario II, $\textsc{SPIC}_1$ notably underestimates the number of clusters when $n=100$, improves at $n=200$, and accurately estimates the number of clusters for $n\geq 300$. In contrast, $\textsc{SPIC}_2$ exhibits slight underestimation at $n=200$ across the HS method, the SP method with HS-based initialization, and the RSP method, whereas it achieves accurate estimation for $n\geq 300$. In Scenario III, results emphasize the distinction between $\textsc{SPIC}_1$ and $\textsc{SPIC}_2$ when $n=100$, where both criteria underestimate the number of clusters. This contrast is particularly informative for understanding the behavior of the selection criteria in small-sample settings. For $n\geq 200$, both criteria accurately estimate the number of clusters, indicating that their performance stabilizes with larger samples. Among the methods considered for estimating the number of clusters, the RSP method consistently demonstrates superior performance.

\begin{table}[htbp]
  \centering
  \caption{Means of 500 estimated number of clusters using distinct criteria across various methods under different CID and cluster membership models.}
 \setlength{\tabcolsep}{4pt} 
  \begin{subfigure}[t]{0.48\textwidth}
    \centering
    \caption*{$\textsc{SPIC}_1$}
    \begin{adjustbox}{max width=0.9\linewidth}
      \begin{tabular}{ccccccccccc}
        \toprule
      \multicolumn{2}{c}{Method} & RC    & HS    & SP   & RSP & & RC    & HS    & SP    & RSP \\
      \cmidrule(rl){3-6} \cmidrule(rl){8-11}  
       CID   & $n$     & \multicolumn{4}{c}{Covariate-independence} & & \multicolumn{4}{c}{Covariate-dependence}\\
      \midrule  
    M1    & 100   & 1.80  & 1.85  & 1.87  & 1.92  &       & 1.83  & 1.94  & 1.96  & 1.97 \\
          & 200   & 2.00  & 2.00  & 2.00  & 2.00  &       & 1.99  & 2.00  & 2.00  & 2.00 \\
          & 300   & 2.00  & 2.00  & 2.00  & 2.00  &       & 2.00  & 2.00  & 2.00  & 2.00 \\
          & 400   & 2.00  & 2.00  & 2.00  & 2.00  &       & 2.00  & 2.00  & 2.00  & 2.00 \\
          \midrule
    M2    & 100   & 1.38  & 1.45  & 1.47  & 1.58  &       & 1.54  & 1.84  & 1.88  & 1.91 \\
          & 200   & 1.99  & 1.99  & 1.99  & 2.00  &       & 1.99  & 2.00  & 2.00  & 2.00 \\
          & 300   & 2.00  & 2.00  & 2.00  & 2.00  &       & 2.00  & 2.00  & 2.00  & 2.00 \\
          & 400   & 2.00  & 2.00  & 2.00  & 2.00  &       & 2.00  & 2.00  & 2.00  & 2.00 \\
          \midrule
    M3    & 100   & 1.67  & 1.74  & 1.77  & 1.85  &       & 1.76  & 1.91  & 1.95  & 1.96 \\
          & 200   & 2.00  & 2.00  & 2.00  & 2.00  &       & 2.00  & 2.00  & 2.00  & 2.00 \\
          & 300   & 2.00  & 2.00  & 2.01  & 2.00  &       & 2.01  & 2.00  & 2.00  & 2.00 \\
          & 400   & 2.00  & 2.00  & 2.00  & 2.00  &       & 2.02  & 2.01  & 2.01  & 2.00 \\
          \midrule
           \multicolumn{1}{c}{M4} & 100   & 1.78  & 1.94  & 1.98  & 1.99  &       & 1.74  & 1.91  & 1.96  & 1.99 \\
                 & 200   & 2.13  & 2.44  & 2.44  & 2.48  &       & 2.22  & 2.35  & 2.35  & 2.35 \\
                 & 300   & 2.83  & 3.00  & 3.00  & 3.00  &       & 2.86  & 2.94  & 2.95  & 2.98 \\
                 & 400   & 2.97  & 3.00  & 3.00  & 3.00  &       & 2.99  & 3.00  & 3.00  & 3.00 \\
          \midrule
           \multicolumn{1}{c}{M5} & 100   & 1.55  & 1.71  & 1.74  & 1.77  &       & 1.52  & 1.71  & 1.74  & 1.77 \\
                 & 200   & 1.98  & 1.99  & 2.00  & 2.00  &       & 1.97  & 2.00  & 2.00  & 2.00 \\
                 & 300   & 2.00  & 2.00  & 2.00  & 2.00  &       & 2.00  & 2.00  & 2.00  & 2.00 \\
                 & 400   & 2.00  & 2.00  & 2.00  & 2.00  &       & 2.00  & 2.00  & 2.00  & 2.00 \\
        \bottomrule
      \end{tabular}
    \end{adjustbox}
  \end{subfigure}
  \hfill
  \begin{subfigure}[t]{0.48\textwidth}
    \centering
    \caption*{$\textsc{SPIC}_2$}
    \begin{adjustbox}{max width=0.9\linewidth}
           \begin{tabular}{ccccccccccc}
        \toprule
      \multicolumn{2}{c}{Method} & RC    & HS    & SP   & RSP & & RC    & HS    & SP    & RSP \\
      \cmidrule(rl){3-6} \cmidrule(rl){8-11}    
       CID   & $n$     & \multicolumn{4}{c}{Covariate-independence} & & \multicolumn{4}{c}{Covariate-dependence}\\
        \midrule
    M1    & 100   & 1.92  & 1.95  & 1.92  & 1.97  &       & 1.94  & 2.01  & 1.99  & 1.99 \\
          & 200   & 2.02  & 2.02  & 2.01  & 2.00  &       & 2.03  & 2.02  & 2.01  & 2.00 \\
          & 300   & 2.03  & 2.02  & 2.03  & 2.00  &       & 2.02  & 2.01  & 2.01  & 2.00 \\
          & 400   & 2.02  & 2.02  & 2.03  & 2.00  &       & 2.02  & 2.01  & 2.02  & 2.01 \\
          \midrule
    M2    & 100   & 1.83  & 1.89  & 1.85  & 1.91  &       & 1.88  & 1.97  & 1.96  & 1.98 \\
          & 200   & 2.01  & 2.01  & 2.02  & 2.00  &       & 2.02  & 2.00  & 2.00  & 2.00 \\
          & 300   & 2.03  & 2.03  & 2.01  & 2.00  &       & 2.02  & 2.00  & 2.00  & 2.00 \\
          & 400   & 2.10  & 2.10  & 2.03  & 2.03  &       & 2.04  & 2.00  & 2.00  & 2.00 \\
          \midrule
    M3    & 100   & 1.95  & 1.97  & 1.89  & 1.95  &       & 1.95  & 2.02  & 1.99  & 1.99 \\
          & 200   & 2.06  & 2.06  & 2.03  & 2.02  &       & 2.04  & 2.03  & 2.01  & 2.00 \\
          & 300   & 2.15  & 2.15  & 2.11  & 2.07  &       & 2.08  & 2.04  & 2.05  & 2.03 \\
          & 400   & 2.39  & 2.40  & 2.29  & 2.34  &       & 2.10  & 2.06  & 2.07  & 2.02 \\
          \midrule
           \multicolumn{1}{c}{M4} & 100   & 2.07  & 2.21  & 2.21  & 2.23  &  & 2.10  & 2.21  & 2.21  & 2.21 \\
                 & 200   & 2.68  & 2.94  & 2.93  & 2.98  &  & 2.74  & 2.92  & 2.90  & 2.98 \\
                 & 300   & 2.95  & 3.00  & 3.00  & 3.00  &  & 2.97  & 2.99  & 2.99  & 3.00 \\
                 & 400   & 3.02  & 3.01  & 3.02  & 3.00  &  & 3.00  & 3.00  & 3.00  & 3.00 \\
          \midrule
           \multicolumn{1}{c}{M5} & 100   & 1.72  & 1.88  & 1.93  & 1.97  &  & 1.72  & 1.85  & 1.86  & 1.90 \\
                 & 200   & 2.01  & 2.02  & 2.01  & 2.01  &  & 2.02  & 2.03  & 2.02  & 2.01 \\
                 & 300   & 2.03  & 2.01  & 2.00  & 2.00  &  & 2.05  & 2.01  & 2.02  & 2.01 \\
                 & 400   & 2.00  & 2.00  & 2.01  & 2.00  &  & 2.04  & 2.01  & 2.00  & 2.00 \\
        \bottomrule
      \end{tabular}
    \end{adjustbox}
  \end{subfigure}
  \label{tab:est_k}
\end{table}

\end{subsection}

\end{section}

\begin{section}{Applications} 
\label{sec:applications}

This section demonstrates the methodological applicability using three benchmark datasets from the UCI Machine Learning Repository: a real estate valuation study, the Cleveland Heart Disease study, and the AIDS Clinical Trials Group study (see Appendix \ref{subsec:ACTG}). All continuous variables were standardized to have a mean of zero and a standard deviation of one to facilitate comparability across analyses. 
Supplementary figures are provided in Appendix \ref{table-figure}.

\begin{subsection}{Application in Real Estate Valuation: A Case Study in New Taipei City} \label{subsec:Sindian}

The dataset used in this study, obtained from the Taiwan Ministry of the Interior’s Public Database, comprises 414 residential property transactions in the Sindian District of New Taipei City, recorded 
between June 2012 and May 2013. Each observation includes the unit-area transaction price (\textit{price}), transaction date, property age (\textit{age}), walking distance to the nearest Mass Rapid Transit (MRT) station (\textit{dist}), number of nearby convenience stores (\textit{store}), and geographic coordinates (latitude (\textit{lat}) and longitude (\textit{long})). Further details on the dataset and its construction are provided in \cite{yeh2018building}.

Building on the unsupervised learning framework of \cite{li2023nonparametric}, which identifies latent heterogeneity in house prices between the left and right riverbanks, we model \textit{price} as the response variable.
The covariates considered are \textit{age}, \textit{dist}, \textit{store}, \textit{lat}, \textit{long}, and a binary riverbank indicator (\textit{bank}; 0 = right bank, 1 = left bank). To improve numerical stability, we excluded seven observations that exhibited extreme values in either the response variable or the covariates, resulting in a final sample of 407 properties. The CID model, incorporating cluster-specific intercepts and riverbank effects, was applied in this data analysis. The number of clusters was selected using $\textsc{SPIC}_1$ and $\textsc{SPIC}_2$, consistently favoring two clusters across the RC, HS, SP, and RSP methods. The corresponding clustering estimates exhibited a high degree of consistency across methods: the RI value between the RC and RSP methods was 0.830, while all pairwise comparisons among the HS, SP, and RSP methods achieved an RI value of 1.000, indicating the same clustering and parameter estimates among these three methods.

Using the RSP method, we estimate the cluster index coefficients for the CID model incorporating cluster-specific intercepts and riverbank effects as follows:
\begin{table}[H]
  \centering
  \footnotesize
    \begin{tabular}{cccccccc}
    \toprule
    Cluster & intercept& \textit{bank} & \textit{age}   & \textit{dist}  & \textit{store}   & \textit{lat}   & \textit{long}   \\
    \midrule
    1     & 0    & 4.39 (1.483)& 1  & 3.27 (0.944)  & -0.74 (0.061) & -1.58 (0.338) & 0.28 (0.079) \\
    2     & -3.68 (0.482)      & 4.27 (0.619)& 1   & 3.27 (0.944) & -0.74 (0.061)& -1.58 (0.338) & 0.28 (0.079)  \\
          \bottomrule
    \end{tabular}%
\end{table}%
\noindent Standard errors, reported in parentheses, are computed using the asymptotic variance matrix estimator in Section \ref{subsec:PSISP}. Significant differences in cluster-specific intercepts and significant adverse effects associated with \textit{bank} underscore the effectiveness of the cluster number selection criteria. The results further suggest that latent clusters and \textit{bank} do not exhibit significant interaction effects on \textit{price}. Figure \ref{fig:Sindian_map} indicates that the selected clusters are not solely driven by geographic location.
Among the cluster-invariant coefficients, \textit{store} and \textit{lat} exhibit significant associations with \textit{price}, whereas \textit{age}, \textit{dist}, and \textit{long} are adversely associated. A second set of estimates, obtained using the RSP method for the CID model with cluster-specific intercepts only, is presented below.
\begin{table}[H]
  \centering
  \footnotesize
    \begin{tabular}{cccccccc}
    \toprule
    Cluster & intercept& \textit{bank} & \textit{age}   & \textit{dist}  & \textit{store}   & \textit{lat}   & \textit{long}   \\
    \midrule
    1     & 0   & 4.32 (0.811)& 1 & 3.27 (0.807) & -0.74 (0.073)& -1.58 (0.308) & 0.28 (0.070) \\
    2     & -3.70 (0.589)    & 4.32 (0.811) & 1& 3.27 (0.807) & -0.74 (0.073)& -1.58 (0.308) & 0.28 (0.070) \\
          \bottomrule
    \end{tabular}%
\end{table}%
\noindent 
The estimated coefficients are nearly identical to those from the initial fitting, demonstrating the stability of the model across refinements. Figure \ref{fig:Sindian_mean_D} depicts the fitted response function based on the clusterwise index mean model. A declining trend is observed across most of the index support, except for the interval $(11,16)$, where the curve levels off. This pattern indicates that Cluster 1 is associated with lower \textit{price} than Cluster 2. Moreover, properties located on the right bank of the Sindian River tend to command higher \textit{price} than those on the left bank. These findings are further corroborated in Figure \ref{fig:Sindian_pdf_D}, which presents the estimated density functions stratified by cluster and riverbank side (left vs. right).

Based on the estimated number of clusters and the corresponding clustering estimate, the structural dimension of the SDR-based cluster membership model was identified as one, with Cluster 2 serving as the reference group. The table below presents the estimated coefficients for the CS direction, accompanied by their standard errors (in parentheses).
\begin{table}[H]
  \centering
  \footnotesize
    \begin{tabular}{cccccc}
    \toprule
    \textit{bank} & \textit{age}   & \textit{dist}  & \textit{store}   & \textit{lat}   & \textit{long}   \\
    \midrule
    -1.13 (0.391)& 1 & 2.02 (0.177)& -0.18 (0.099)& -0.11 (0.068) & 2.17 (0.089)\\
          \bottomrule
    \end{tabular}%
\end{table}%
\noindent The significant effects of $\textit{age}$, $\textit{dist}$, and $\textit{long}$ in one direction and $\textit{bank}$ in the opposite direction on the probability of being in Cluster 1 support the suitability of a covariate-dependent cluster membership model. For comparison, we also estimated a logistic cluster membership model. Figure 
\ref{fig:Sindian_prob} illustrates that the logistic model does not adequately capture the structure of the cluster membership probabilities, underscoring its limitations compared to the SDR-based method.

\begin{figure}[htbp]
\centering
\begin{subfigure}[t]{0.3\linewidth} 
    \centering
    \includegraphics[width=\linewidth]{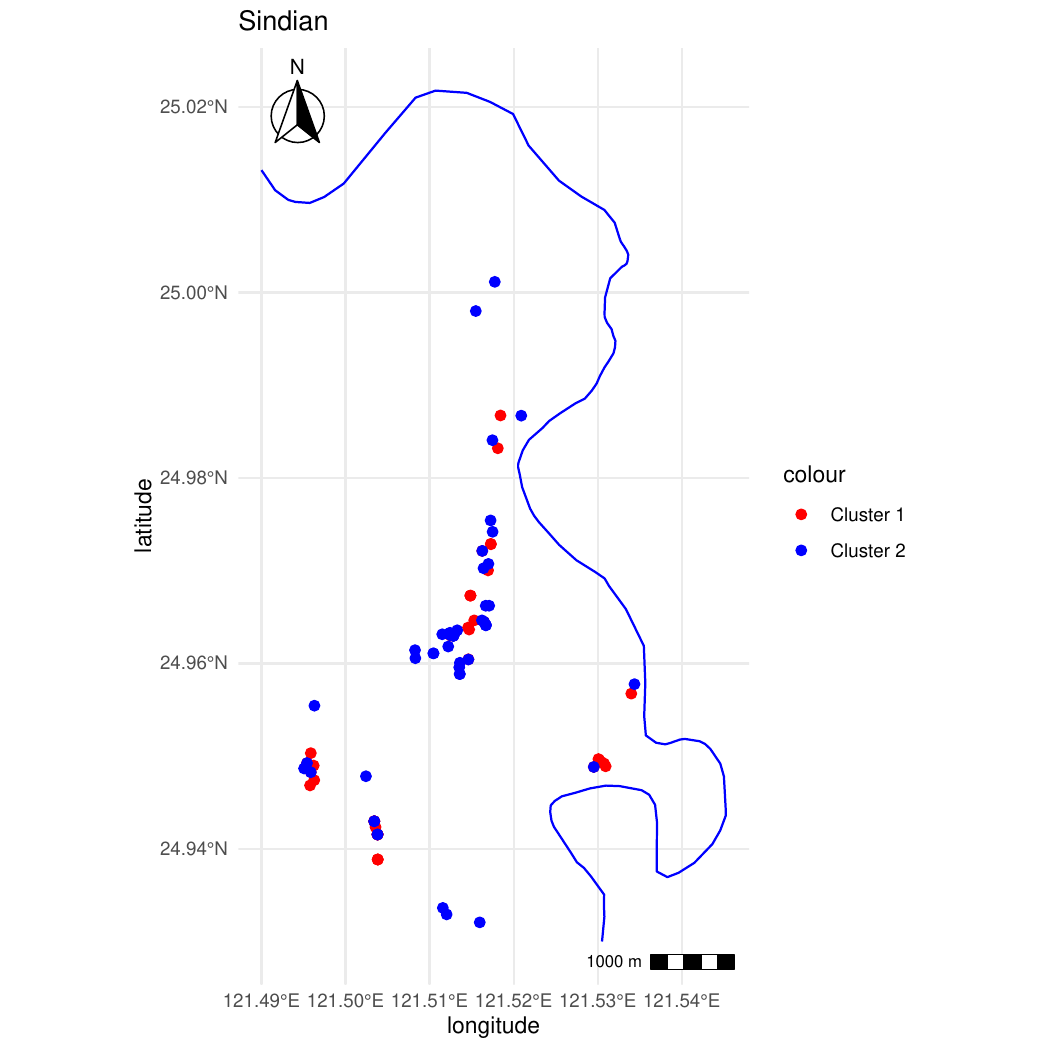} 
    \caption{{\footnotesize Left Bank}}
\end{subfigure}
\begin{subfigure}[t]{0.3\linewidth} 
    \centering
    \includegraphics[width=\linewidth]{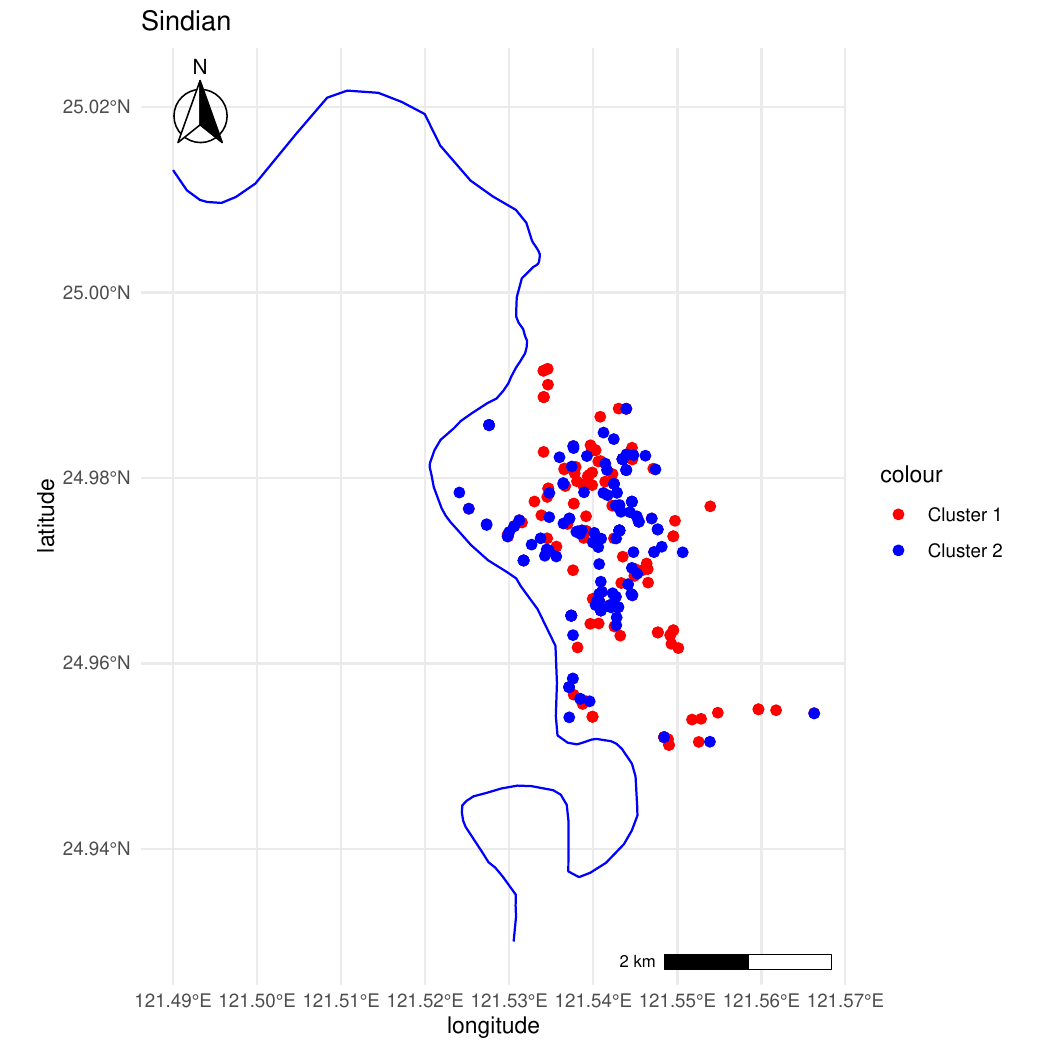} 
    \caption{{\footnotesize Right Bank}}
\end{subfigure}

\caption{{\footnotesize Scatter plots of houses in Clusters 1 and 2, grouped by the left and right banks of the Sindian River.}}
\label{fig:Sindian_map}
\end{figure}

\begin{figure}[htbp]
\centering
 \includegraphics[width=12cm,height=6cm]{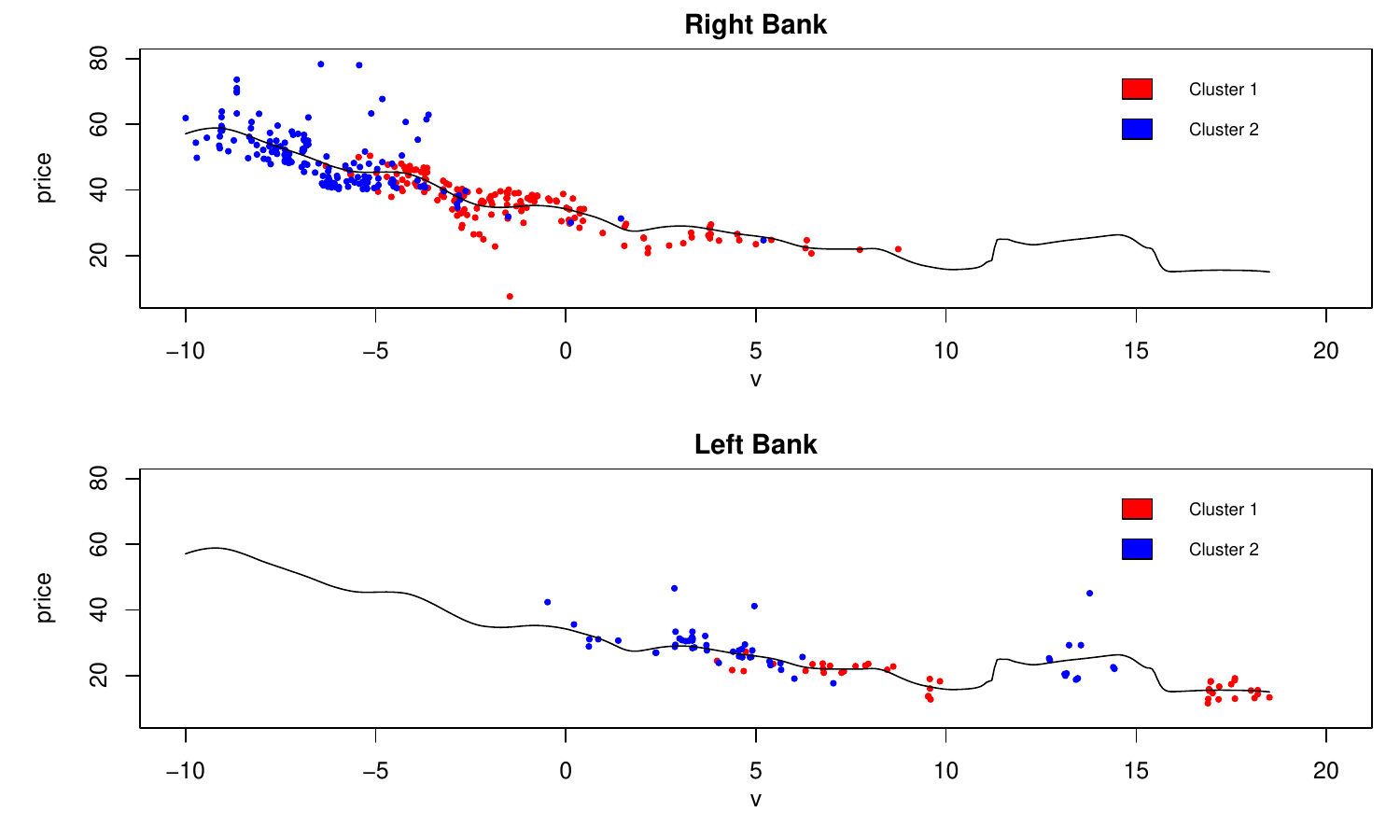}
 \caption{{\footnotesize Scatter plots of house prices from a real estate valuation study overlaid with the estimated mean function from the clusterwise index mean model incorporating cluster-specific intercepts.}} \label{fig:Sindian_mean_D} 
\end{figure}

\end{subsection}

\begin{subsection}{Application in Heart Disease Research} \label{subsec:CLE}

The Cleveland Clinic Foundation Heart Disease dataset comprises measurements on 13 clinical variables from 303 patients, originally compiled to aid in diagnosing heart disease. \citet{lauer1999impaired} identified the maximum heart rate achieved during exercise as a prognostic indicator of cardiac mortality, and \citet{wei2013latent} subsequently proposed its use as a surrogate endpoint for heart disease status. This study models maximum heart rate (\textit{mhr}) as the response variable, using the same set of covariates as \citet{ma2017concave}---age (\textit{age}), gender (\textit{gender}; 0 = female, 1 = male), resting blood pressure (\textit{rbp}), serum cholesterol (\textit{sc}), fasting blood sugar (\textit{fbs}), and resting electrocardiographic results (\textit{re})---with heart disease status (\textit{hds}; 0 = no heart disease, 1 = heart disease) included as an additional covariate. The analysis examines whether latent heterogeneity affects the association between \textit{mhr} and \textit{hds} or whether the relationship is consistent across the population. Further details on the dataset are available in \citet{detrano1989international}.

Seven individuals with extreme values in the response variable or covariates were excluded to improve numerical stability, resulting in a final sample of 296 participants. The CID model, incorporating cluster-specific intercepts and heart disease effects, was fitted to the data. According to $\textsc{SPIC}_1$ and $\textsc{SPIC}_2$, the RC method consistently identified more clusters than the HS, SP, and RSP methods, each of which selected two clusters. The resulting RI values were 0.547 between the RC and RSP methods, 0.624 between the HS and RSP methods, and 1.000 between the SP and RSP methods, indicating that the clustering estimates from the RC and HS methods differ substantially from that of the RSP method. In contrast, the SP and RSP methods produce the same clustering and parameter estimates.

In the RSP method, the estimated cluster index coefficients for the CID model incorporating cluster-specific intercepts and disease effects are reported below.
\begin{table}[H]
  \centering
  {\footnotesize
   \begin{adjustbox}{max width=\linewidth}
    \begin{tabular}{ccccccccc}
    \toprule
    Cluster & intercept &\textit{hds}  & \textit{age}   & \textit{gender} & \textit{rbp}   & \textit{sc}   & \textit{fbs}   & \textit{re} \\
    \midrule
    1     &   0    & 2.38  (0.140)& 1  & -0.16 (0.063)& 0.02 (0.024) & 0.00  (0.027) & -0.51 (0.079)& 0.17 (0.033)\\
    2     & -4.37 (0.120)& 1.51 (0.171) & 1 & -0.16 (0.063)& 0.02 (0.024) & 0.00  (0.027) & -0.51 (0.079)& 0.17 (0.033)\\
          \bottomrule
    \end{tabular}%
     \end{adjustbox}
     }
\end{table}%
\noindent The significant variation in cluster-specific intercepts, coupled with the significant adverse effects of \textit{hds}, highlights the effectiveness of the proposed cluster selection strategy in capturing meaningful heterogeneity. These results indicate a significant interaction between the latent cluster structure and disease status in shaping \textit{mhr}.
Within the CID model, the cluster-invariant coefficients of \textit{gender} and \textit{fbs} are significantly associated with \textit{mhr}, while \textit{age} and \textit{re} exhibit significant adverse effects. In Figure \ref{fig:CLE_mean_D}, the estimated mean function displays a clear downward trend. Notably, the index-\textit{mhr} relationship appears approximately linear within Cluster 2, while Cluster 1 exhibits a nonlinear and decreasing trajectory. These patterns indicate that individuals in Cluster 2 tend to achieve higher \textit{mhr} than those in Cluster 1. Moreover, within each cluster, individuals diagnosed with heart disease tend to have lower \textit{mhr} than their non-diseased counterparts. Figure \ref{fig:CLE_pdf_D} further illustrates substantial heterogeneity in distributional shapes across subpopulations, providing visual evidence of the model’s ability to capture distributional differences attributable to latent cluster structure and disease status.

The structural dimension of the SDR-based cluster membership model was identified as one, with Cluster 2 serving as the reference group. The coefficients for the CS direction are estimated below.
\begin{table}[H]
  \centering
  \footnotesize
    \begin{tabular}{ccccccc}
    \toprule
     \textit{hds}  & \textit{age}   & \textit{gender} & \textit{rbp}   & \textit{sc}   & \textit{fbs}   & \textit{re} \\
    \midrule
    -2.39 (0.379)& -0.82 (0.079)& -0.29 (0.129)& 1 & 1.03 (0.110) & 1.47 (0.200)  & 0.03 (0.076)  \\
          \bottomrule
    \end{tabular}%
\end{table}%
\noindent The results indicate that $(\textit{rbp},\textit{sc},\textit{fbs})$ and $(\textit{hds},\textit{age},\textit{gender})$ exert significant effects in opposite directions on the probability of assignment to Cluster 1. These findings suggest that the proposed covariate-dependent cluster membership model successfully captures latent heterogeneity associated with the observed covariates.
Figure \ref{fig:CLE_prob} shows that the logistic model fails to account for the complex nonlinear structure underlying the data.

\begin{figure}[htbp]
\centering
 \includegraphics[width=12cm,height=6cm]{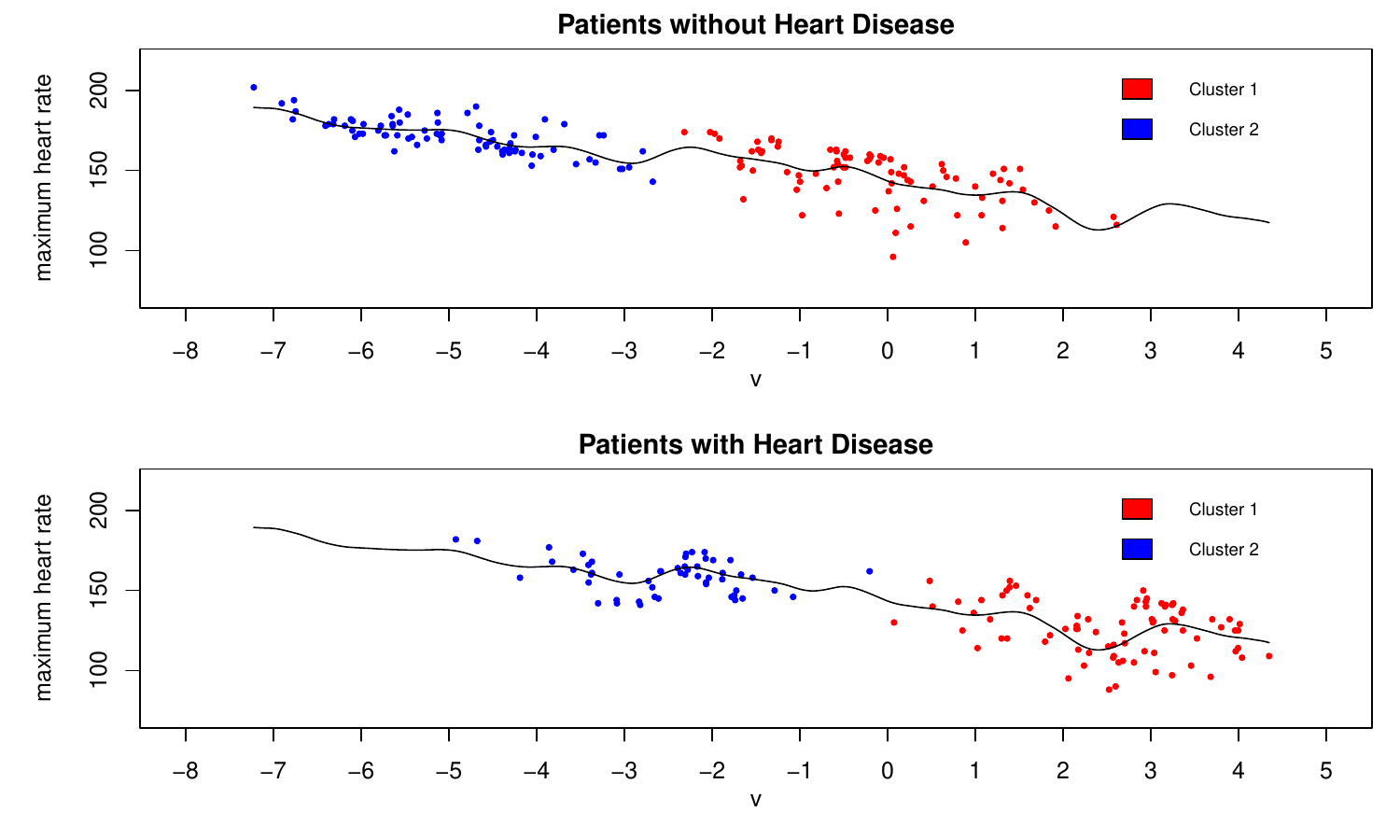}
 \caption{{\footnotesize Scatter plots of maximum heart rates of patients with and without heart disease from the Cleveland Heart Disease study overlaid with the estimated mean function from the clusterwise index mean model incorporating cluster-specific intercepts and disease effects.}} \label{fig:CLE_mean_D} 
\end{figure}

\end{subsection}

\end{section}

\begin{section}{Discussion}

\label{sec:discussion}

Estimation in unsupervised mixture regression models presents considerable challenges within semiparametric distribution and cluster membership frameworks, particularly in contrast to their parametric counterparts. Addressing these challenges necessitates the imposition of substantial structural conditions on both the cluster indices governing the distributional model and the functional form of the cluster membership mechanism. To this end, we employ the SID model as an alternative representation of the CID model to construct clustering and parameter estimators through the SP estimation method. A heuristic initialization strategy is proposed to enhance computational efficiency by providing effective values. A principal contribution of this work is developing an optimal clustering rule that sharpens clustering and parameter estimation. To characterize the covariate structure underlying cluster membership more effectively, we incorporate a SDR technique and propose a novel estimator for the central subspace. Concluding the methodological framework, we propose two semiparametric information criteria for selecting the number of clusters, emphasizing their empirical performance and practical relevance.

The proposed semiparametric CID model offers considerable flexibility and is a powerful tool for exploratory data analysis, enabling the investigation of complex underlying distributional forms. Our pseudo estimation method, developed in parallel with parametric counterparts, constitutes a crucial advancement and underpins constructing a model-checking rule. Assessing the independence of each covariate from the cluster variable enables us to determine whether the cluster membership model is covariate-independent or covariate-dependent. However, this aspect merits further theoretical and empirical study. When the covariate-dependent structure is supported, the proposed SDR method provides an effective estimator of the CS, facilitating a model selection strategy. Practical data scenarios, such as those encountered in medical image recognition and web page classification, often involve labeled and unlabeled observations. A central challenge remains the effective integration of information from labeled and unlabeled data to enhance the clustering estimator's accuracy and improve the parameter estimators' precision.

Consider the clusterwise linear model $Y=\sum^{k}_{\ell=1}I(C=\ell)\gamma^{\top}_{\ell}Z+\varepsilon$, where the error term $\varepsilon$ 
is independent of $(X, C)$ and 
follows an unknown density $g_{0}(\cdot)$ with mean zero and variance $\sigma^{2}$. The proposed estimation procedure can refine the clustering and parameter estimators previously established in the literature. The optimal clustering rule is constructed by integrating the clusterwise index density model $g_{0}(y-\gamma^{\top}_{\ell}z)$ with the SDR-based cluster membership model. For the clusterwise index model $Y=\sum^{k}_{\ell=1}I(C=\ell)m(\gamma^{\top}_{\ell}Z)+\varepsilon$, which specifies an unknown smooth response function 
 $m(\cdot)$, the separation penalty estimation is adapted by replacing the pseudo sum of integrated squares with the pseudo sum of squares. The resulting density estimator is subsequently used to construct the optimal clustering rule. By extending the generalized regression framework of \cite{han1987non} to accommodate latent cluster structures, clustering data points via penalized rank correlation remains computationally challenging.
Our methodology enables direct parameter estimation without imposing monotonicity constraints on the bivariate function involving the cluster indices and the error term.
     
When the response variable denotes the failure time of an event, the proposed CID framework naturally extends to a general semiparametric clusterwise index hazards model. Due to loss to follow-up or events preceding study entry, censoring or truncation poses considerable challenges for statistical analysis. The methodological developments presented herein lay the groundwork for advancing inference procedures under such complexities. Longitudinal data present distinct challenges, particularly in biomedicine, epidemiology, marketing, and the social sciences, where time-varying covariates and irregular or intermittent measurement schedules complicate modeling and inference. To address this, we extend the clusterwise index mean model to accommodate time-varying covariate effects. A central methodological difficulty in this extension is properly handling within-subject correlation to achieve efficient estimation.

\end{section}
\end{spacing}

\bibliographystyle{biom}
\bibliography{reference}{}

\appendix

\section*{}

\renewcommand{\thesection}{A}
\renewcommand{\theequation}{\thesection.\arabic{equation}}
\renewcommand{\thetable}{S\arabic{table}}
\renewcommand{\thefigure}{S\arabic{figure}}
\renewcommand{\thelem}{\thesection.\arabic{aplemma}}

\setcounter{table}{0}
\setcounter{figure}{0}

\subsection{Notation}
\label{sec:A.Notations}

Let $W_{*}$ denote the subvector of $W$ obtained by excluding the components corresponding to the known entries of $\gamma$, and define ${X}_{d*}=(X_{d+1},\dots,X_{p})^{\top}$.
Denote by $\mathcal{W}_{*}$ the support of $W_{*}$ and $\mathcal{A}_{d} = \mathcal{A}_{d1} \times \cdots \times \mathcal{A}_{dd}$ the parameter space of $A_{d*}$.
The following notation is used throughout the proofs:
\begin{align*}
N_{s,h}(y,v;\gamma)  =&\frac{1}{n}\sum_{i=1}^n \big(I(Y_i \leq y)\big)^s K_{2,h}\big(\gamma^{\top}W_i - v\big),\\
N_{s,\ell,h_d}(u;A_{d})  =& \frac{1}{n}\sum_{i=1}^n \big(I(C_i = \ell)\big)^s \mathcal{K}_{r,h_d}\big(A_{d}^{\top}X_i - u\big), \\
M_{s,m}(y, w, v)  =& (-1)^m\E\big[\big(G\big(y, \gamma^{\text{o} \top}W\big)\big)^s (W_{*} - w_{*})^{\otimes m}| \gamma^{\top}W = v\big], \\
M_{s,m,\ell}(x, u)  = &(-1)^m\E\big[\big(\pi_{\ell}\big(A^{\text{o} \top}X\big)\big)^s ({X}_{d*} - {x}_{d*})^{\otimes m}| A_{d}^{\top}X = u\big],
\end{align*}
where $s=0, 1$, $\ell =1, \dots, k-1$, and $m=0, 1, 2$.

The proof of Lemma \ref{A:Lemma_1} requires the uniform convergence of the derivatives $ \partial_{\gamma_{*}}^m N_{s,h}(y, \gamma^{\top}w; \gamma)$ and $\partial_{\vvec{(A_{d*})}}^m N_{s,\ell,h_d}(A_{d}^{\top}x; A_{d})$, for $s=0, 1$, $\ell = 1, \dots, k-1$, and $m=0,1,2$. 
These derivatives play a central role in analyzing the smooth functionals
\begin{align*}
\widehat{G}_{h}(y,\gamma^{\top}w; \gamma) = & \frac{N_{1,h}(y,\gamma^{\top}w;\gamma)}{N_{0,h}(y,\gamma^{\top}w;\gamma)}, ~\widehat{\pi}_{\ell,h_{d}}(A_{d}^{\top}x; A_{d}) = \frac{N_{1,\ell,h_{d}}(A_{d}^{\top}x; A_{d})}{N_{0,\ell,h_{d}}(A_{d}^{\top}x; A_{d})},\\
\partial_{\gamma_{*}} \widehat{G}_{h}(y,\gamma^{\top}w; \gamma) = & \sum_{s = 0}^1 \big( - \widehat{G}_h\big(y,\gamma^{\top}w;\gamma \big)\big)^{1-s} \frac{\partial_{\gamma_{*}}N_{s,h}(y,\gamma^{\top}w;\gamma)}{N_{0,h}(y,\gamma^{\top}w;\gamma)},\\
\partial_{\vvec{(A_{d*})}} \widehat{\pi}_{\ell,h_{d}}(A_{d}^{\top}x; A_{d}) = & \sum_{s = 0}^1 \big( - \widehat{\pi}_{\ell,h_{d}}(A_{d}^{\top}x; A_{d})\big)^{1-s} \frac{\partial_{\vvec{(A_{d*})}}N_{s,\ell,h_{d}}(A_{d}^{\top}x; A_{d})}{N_{0,\ell,h_{d}}(A_{d}^{\top}x; A_{d})}, \\
\partial_{\gamma_{*}}^2 \widehat{G}_{h}(y,\gamma^{\top}w; \gamma) 
=& \sum_{s = 0}^1 \big( - \widehat{G}_h\big(y,\gamma^{\top}w;\gamma \big)\big)^{1-s} \frac{\partial_{\gamma_{*}}^2N_{s,h}(y,\gamma^{\top}w;\gamma)}{N_{0,h}(y,\gamma^{\top}w;\gamma)} -\frac{\partial_{\gamma_{*}}N_{0,h}(y,\gamma^{\top}w;\gamma)}{N_{0,h}(y,\gamma^{\top}w;\gamma)} \\
&\big(\partial_{\gamma_{*}} \widehat{G}_h \big(y, \gamma^{\top}w; \gamma \big) \big)^{\top} - \partial_{\gamma_{*}} \widehat{G}_h \big(y, \gamma^{\top}w; \gamma \big) \bigg( \frac{\partial_{\gamma_{*}}N_{0,h}(y,\gamma^{\top}w;\gamma)}{N_{0,h}(y,\gamma^{\top}w;\gamma)}\bigg)^{\top},\\
\partial_{\vvec{(A_{d*})}}^2 \widehat{\pi}_{\ell,h_{d}}(A_{d}^{\top}x; A_{d})=&
\sum_{s = 0}^1 \big( - \widehat{\pi}_{\ell,h_{d}}(A_{d}^{\top}x; A_{d}) \big)^{1-s} \frac{\partial_{\vvec{(A_{d*})}}^2N_{s,\ell,h_{d}}(A_{d}^{\top}x; A_{d})}{N_{0,\ell,h_{d}}(A_{d}^{\top}x; A_{d})} \\
&-\frac{\partial_{\vvec{(A_{d*})}}N_{0,\ell,h_{d}}(A_{d}^{\top}x; A_{d})}{N_{0,\ell,h_{d}}(A_{d}^{\top}x; A_{d})} 
\big(\partial_{\gamma_{*}} \widehat{\pi}_{\ell,h_{d}}(A_{d}^{\top}x; A_{d})  \big)^{\top} \\
&- \partial_{\vvec{(A_{d*})}} \widehat{\pi}_{\ell,h_{d}}(A_{d}^{\top}x; A_{d})  \bigg( \frac{\partial_{\vvec{(A_{d*})}}N_{0,\ell,h_{d}}(A_{d}^{\top}x; A_{d})}{N_ {0,\ell,h_{d}}(A_{d}^{\top}x; A_{d})}\bigg)^{\top}.
\end{align*}

\subsection{Assumptions}
\label{sec:A.Assumptions}

Define $D_{\gamma}=\min_{\ell_{1} \neq \ell_{2}} \| \gamma^{\text{o}}_{\ell_{1}} - \gamma^{\text{o}}_{\ell_{2}}\|$ and $\phi_{n}=\phi_0 \sqrt{\frac{\log n }{n}}$, where $\|\cdot\|$ denotes the Frobenius norm of a matrix and $\phi_{0}>0$ is a constant that increases as the variation in the response variable grows. Let $f_{\gamma}(v)$ denote the density function of $\gamma^{\top}W$.
The following assumptions are imposed to establish the oracle property of $(\widehat{\beta}^{\lambda},\widehat{\gamma}^{\lambda})$:
\begin{itemize}
\item[A1.] $h \in \big(h_{\ell}n^{-1/5},h_{u}n^{-1/8}\big)$ for some positive constants $h_{\ell}$ and $h_u$.
  \item[A2.] $\mathcal{X}$, $\mathcal{Y}$, and $\Gamma$ are compact.
  \item[A3.] $d^{4}_{v}f_{\gamma}(v)$ and $\partial^{m+2}_v E\big[(G(y,\gamma^{\text{o}\top}W))^{s}(W_{*}-w_{*})^{\otimes m}\big|\gamma^{\top}W = v\big]$, $s=0,1$, $m = 1, 2$, are Lipschitz continuous in $(v, \gamma_{*})$ with Lipschitz constants independent of $(y,x)$.
  \item[A4.] $\inf_{\{v,\gamma_{*}\}}f_{\gamma}(v) > 0$.
  \item[A5.] $V_{1}$ and $V_{2}$ are positive definite.
  \item[A6.] $B_{n}=\Big\{ \beta: \sup_{\{1\leq i\leq n\}} \| \beta_i - \beta^{\text{o}}_i \|_{1} \leq \phi_{n} \Big\}$ and $\Gamma_{n}=\Big\{\gamma: \sup_{\{\ell\in\mathcal{C}\}}\|\gamma_{\ell}-\gamma^{\text{o}}_{\ell}\|_{1}\leq\phi_{n}\Big\}$.
\item[A7.] $D_{\gamma}>4\phi_{n}$ and $\lambda \geq c_{0}$ for some constant $c_{0}>0$.
\end{itemize}

Let $f_{A_{d}}(v)$ denote the density function of $A^{\top}_{d}X$. For each fixed $d$, define 
\begin{align*}
A^{\text{o}}_{d}=&\arg\min_{A_{d}}\sum_{\ell =1}^{k-1} \E\big[\big(I(C = \ell) - \pi^{[0]}_{\ell}\big({X}, A^{\top}_{d}X\big) \big)^{2}\big],\\
S_{A_{d}}=&\sum^{k-1}_{\ell=1}\big(I(C=\ell)-\pi^{[0]}_{\ell}\big({X}, A^{\top}_{d}X\big)\big)\pi^{[1]}_{\ell}\big({X},A^{\top}_{d}X\big), 
\end{align*}
and
\begin{align*}
 V_{A_{d}}=&\sum^{k-1}_{\ell=1}E\big[(\pi^{[1]}_{\ell}\big({X}, A^{\top}_{d}X)\big)^{\otimes 2}-(I(C=\ell)-\pi^{[0]}_{\ell}\big({X}, A^{\top}_{d}X)\big)\pi^{[2]}_{\ell}\big({X}, A^{\top}_{d}X\big)\big]. 
 \end{align*}
 By construction, we have
$E[S_{A^{\text{o}}_{d}}]=0$ and $V_{A^{\text{o}}_{d}}
=\sum^{k-1}_{\ell=1}E\big[ \big(\pi^{[1]}_{\ell}({X}, A^{\text{o} \top}_{d}X)\big)^{\otimes 2}\big]$, $d = d^{\text{o}}, \dots, p$.
The following assumptions are made for the asymptotic properties of $\hat{d}$ and $\widehat{A}$:
\begin{itemize}
\item[B1.] $h_{d} \in \big(h_{\ell}n^{-1/\max\{2d+2,d+4\}},h_{u}n^{-1/4r}\big)$ for some positive constants $h_{\ell}$ and $h_{u}$.
\item[B2.] $\mathcal{X}$ and $\mathcal{A}_d$ are compact.   
\item[B3.] $\partial^{r+2}_{u} f_{A_{d}}(u)$ and $\partial^{m+r}_{u} E\big[ \pi_{\ell}(A^{\text{o} \top}X) ({X}_{d*} - {x}_{d*})^{\otimes m} \big | A^{\top}_d X = u\big]$, $m = 1, 2$,  $\ell=1,\dots,k-1$, are Lipschitz continuous in $(u, A_{d*})$ with the Lipschitz constants independent of $(\ell, x)$.
\item[B4.] $\inf_{\{ u, A_{d*}\}} f_{A_{d}}(u) > 0$.
\item[B5.] $\E[(S_{A^{\text{o}}})^{\otimes 2}]$ and $V_{A^{\text{o}}}$ are positive definite.
\end{itemize}

The compact support conditions in {A2} and {B2}, frequently encountered in application, can be substituted by appropriate moment conditions. The uniform convergence of the estimated functions to their targets relies on the smoothness assumptions in {A3} and {B3}. The estimators’ asymptotic behavior is characterized under assumptions {A4}--{A5} and {B4}--{B5}. Assumptions {A6} and {A7} are further imposed to ensure the oracle property of the SP estimator.

\subsection{Technical Lemma}
\label{sec:A.lemma}

\begin{aplemma}
Under assumptions {A2}--{A4},
\begin{align*}
\sup_{\{y, w, \gamma_{*}\}} \bigg\| \partial_{\gamma_{*}}^m\widehat{G}_h(y, \gamma^{\top}w; \gamma) - G^{[m]}(y, w, \gamma^{\top}w) - \frac{1}{n}\sum_{i=1}^n \xi_i^{[m]}(y,w,\gamma) \bigg\| = o_p \bigg({\frac{\log n}{nh^{2m+1}}}\bigg) + O \big(h^4 \big)
\end{align*}
and
\begin{align*}
\sup_{\{y,w,\gamma_{*}\}} \Big\|\frac{1}{n} \sum_{i=1}^n \xi_i^{[m]}(y,w,\gamma)\Big\| = & o_p \Bigg(\sqrt{\frac{\log n}{nh^{2m+1}}}\Bigg) + O \big(h^2 \big),
\end{align*}
and under assumptions {B2}--{B4},
\begin{align*}
\sup_{\{\ell, x, A_{d*}\}} \bigg\| \partial_{\vvec{(A_{d*})}}^m\widehat{\pi}_{\ell,h_d}( A_{d}^{\top}x; A_{d}) - \pi_{\ell}^{[m]}( x,A_{d}^{\top}x) - \frac{1}{n}\sum_{i=1}^n \xi_{i,\ell}^{[m]}(x, A_{d}) \bigg\| = o_p \bigg({\frac{\log n}{nh_{d}^{2m+d}}}\bigg) + O \big(h_{d}^{2r} \big)
\end{align*}
and
\begin{align*}
\sup_{\{\ell, x, A_{d*} \}} \Big\|\frac{1}{n} \sum_{i=1}^n \xi_{i,\ell}^{[m]}(x,A_{d})\Big\| =& o_p \Bigg(\sqrt{\frac{\log n}{nh_{d}^{2m+d}}}\Bigg) + O \big(h_{d}^{r} \big),
\end{align*}
where $G^{[m]}(y, w, v)$, $\pi_{\ell}^{[m]}(x, u)$, $\xi_i^{[m]}(y,w,\gamma)$, and $\xi_{i,\ell}^{[m]}(x,A_{d})$ are defined in the proof for $i = 1, \dots, n$, $m = 0, 1, 2$, and $\ell = 1, \dots, k-1$.
\label{A:Lemma_1}
\end{aplemma}
\begin{proof}
Given the classes 
\begin{align*}
&\{ I(Y \leq y): y \in \mathcal{Y} \}, ~  \{ ({X}_{d*}-{x}_{d*})^{\otimes m}: {x}_{d*} \in {\mathcal{X}}_{d*}\},~\{ (W_{*}-w_{*})^{\otimes m}: w_{*} \in \mathcal{W}_{*}\},\\
&\{ c_1 K^{(m)}_2(\gamma^{\top}W + c_2): c_1, c_2 \in \mathbb{R}, \gamma_{*} \in \Gamma\}, \text{ and }
\{ c_1 K^{(m)}_{r}(\alpha_{j}^{\top}X + c_2): c_1, c_2 \in \mathbb{R}, \alpha_{j} \in \mathcal{A}_{j}\},
\end{align*}
with VC-indices 2, 1, 1, at most $kp-2$, and at most $p-d$, respectively, for $m = 0, 1, 2$, $j = 1, \dots, d$, Lemma 2.12 in \cite{pakes1989simulation} ensures that each class is Euclidean. Moreover, Lemma 2.14 therein implies that the following classes are also Euclidean:
\begin{align*}
&\mathcal{F}_{1m} = \big\{ I(Y\leq y) c_1 K^{(m)}_2(\gamma^{\top}W + c_2) (W_{*}-w_{*})^{\otimes m}: y \in \mathcal{Y}, w_{*} \in \mathcal{W}_{*}, \gamma_{*} \in \Gamma, c_1, c_2, \in \mathbb{R}\big\}\\
\text{and}&\\
&\mathcal{F}_{2m} = \Bigg\{ I(C = \ell) \prod_{\{\sum_{j=1}^d m_j = m\}} (c_{1j} K^{(m_j)}_r(\alpha_{j}^{\top}X + c_{2j})) ({X}_{d*}-{x}_{d*})^{\otimes m}: \ell \in \mathcal{C}, {x}_{d*} \in {\mathcal{X}}_{d*}, A_{d*}\\
&\hspace{0.6in} \in \mathcal{A}_{d}, c_{1j}, c_{2j}, \in \mathbb{R}\Bigg\}.
\end{align*}
Under Assumptions {A2}, {A3}, {B2}, and {B3}, it holds that
\begin{eqnarray}
\sup_{\{y, w, \gamma_{*}\}} \bigg\|\E \bigg[\bigg( \frac{1}{h^{m+1}} (I(Y \leq y))^s K_2^{(m)} \bigg( \frac{\gamma^{\top}(W-w)}{h} \bigg) \bigg)^2 \big(  \vvec(W_{*} - w_{*}) \big)^{\otimes 2m} \bigg] \bigg\| 
 = O \big( h^{-(2m+1)} \big)~\label{eq:lemma1_1}
\end{eqnarray}
and
\begin{eqnarray}
\sup_{\{\ell,x, A_{d*}\}} \bigg\|\E \bigg[\bigg( \frac{1}{h_{d}^{m+d}} (I(C = \ell))^s \mathcal{K}_r^{(m)} \bigg( \frac{A_{d}^{\top}(X-x)}{h_d} \bigg) \bigg)^2  \big( \vvec({X}_{d*} - {x}_{d*}) \big)^{\otimes 2m} \bigg] \bigg\| 
= O \big( h_{d}^{-(2m+d)} \big)~\label{eq:lemma2_1}
\end{eqnarray}
for $s = 0, 1$ and $m = 0, 1, 2$. Combining (\ref{eq:lemma1_1}) and (\ref{eq:lemma2_1}) with Theorem II.37 in \cite{pollard1984convergence}, we obtain
\begin{align}
&\sup_{\{y, w, \gamma_{*}\}} \Big\| \partial_{\gamma_{*}}^m N_{s,h}(y, \gamma^{\top}w; \gamma) - \E\big[\partial_{\gamma_{*}}^m \big( I(Y \leq y)^s K_{2,h}(\gamma^{\top}W - \gamma^{\top}w) \big)\big] \Big\| = o_p \bigg( \sqrt{\frac{\log n}{nh^{2m+1}}} \bigg)\label{eq:lemma1_2}
\end{align}
and
\begin{align}
&\sup_{\{ \ell, x, A_{d*}\}} \Big\| \partial_{\vvec{(A_{d*})}}^m N_{s,\ell,h_{d}}( A_{d}^{\top}x; A_{d}) - \E\big[\partial_{\vvec{(A_{d*})}}^m \big( I(C = \ell)^s \mathcal{K}_{r,h_d}(A_{d}^{\top}X - A_{d}^{\top}x) \big)\big] \Big\| \nonumber\\
&= o_p \bigg( \sqrt{\frac{\log n}{nh_{d}^{2m+d}}} \bigg), s = 0, 1, m = 0, 1, 2.\hspace{1.5in} \label{eq:lemma2_2}
\end{align}
Using the identities $\E [I(Y \leq y)|W = w] = G(y, \gamma^{\text{o} \top}w)$ and $\E [I(C = \ell)|X = x] = \pi_{\ell}(A^{\text{o} \top}x)$, the expectations in \eqref{eq:lemma1_2} and \eqref{eq:lemma2_2} can be rewritten as
\begin{align}
&\E\big[\partial_{\gamma_{*}}^m \big( \big(I(Y \leq y)\big)^s K_{2,h}(\gamma^{\top}W - \gamma^{\top}w) \big)\big]\nonumber \\
&= \frac{1}{h^{m+1}}\E\bigg[ \big(I(Y\leq y)\big)^s K^{(m)}_2 \bigg( \frac{\gamma^{\top}(W-w)}{h} \bigg) (W_{*}-w_{*})^{\otimes m}\bigg] \nonumber\\
&=\frac{1}{h^{m+1}} \E \bigg[K^{(m)}_2 \bigg( \frac{\gamma^{\top}(W-w)}{h} \bigg) M_{s,m} \big(y,w,\gamma^{\top}W \big) \bigg] \nonumber\\
&= \frac{1}{h^{m}} \int K^{(m)}_2 (\nu) M_{s,m}\big(y,w,\gamma^{\top}w + \nu h\big) f_{\gamma_{*}}\big(\gamma^{\top}w + \nu h\big)d \nu \nonumber\\
&= \partial_{v}^m \big(M_{s,m}\big(y,w,\gamma^{\top}w \big)f_{\gamma_{*}}\big(\gamma^{\top}w\big)\big) + r_{s,m}\big( y,w,\gamma \big) \label{eq:lemma1_3}
\end{align}
and
\begin{align}
&\E\Big[\partial_{\vvec{(A_{d*})}}^m \big( \big(I(C = \ell)\big)^s \mathcal{K}_{r,h_{d}}(A_{d}^{\top}X - A_{d}^{\top}x) \big)\Big]\nonumber 
\end{align}
\begin{align}
&= \frac{1}{h_{d}^{m+d}}\E\Bigg[ \big(I(C = \ell)\big)^s \prod_{\{\sum_{j=1}^d m_j = m\}} K^{(m_j)}_r \bigg( \frac{\alpha_{j}^{\top}(X-x)}{h_d} \bigg) ({X}_{d*}-{x}_{d*})^{\otimes m}\Bigg] \nonumber\\
&=\frac{1}{h_{d}^{m+d}} \E \Bigg[\prod_{\{\sum_{j=1}^d m_j = m\}} K^{(m_j)}_r \bigg( \frac{\alpha_{j}^{\top}(X-x)}{h_d} \bigg) M_{s,m,\ell} \big(x,A_{d}^{\top}X \big) \Bigg] \hspace{0.7in}\nonumber\\
&= \partial_{u}^m \big(M_{s,m,\ell}\big(x,A_{d}^{\top}x \big)f_{A_{d*}}\big(A_{d}^{\top}x\big)\big) + r_{s,m,\ell}\big(x,A_{d} \big), \label{eq:lemma2_3}
\end{align}
where
$\sup_{\{y, w, \gamma_{*}\}} \big\|r_{s,m}\big(y,w,\gamma \big) \big\| = O(h^2)$ and $\sup_{\{\ell, x, A_{d*}\}} \big\|r_{s,m,\ell}\big(x,A_{d} \big) \big\| = O(h_{d}^r)$,
$s = 0, 1$, $m = 0, 1, 2$.
Substituting (\ref{eq:lemma1_3}) and (\ref{eq:lemma2_3}) into (\ref{eq:lemma1_2}) and (\ref{eq:lemma2_2}), respectively, yields \begin{align}
&\sup_{\{y, w, \gamma_{*}\}} \Big\| \partial_{\gamma_{*}}^m N_{s,h}(y, \gamma^{\top}w; \gamma) - \partial_{v}^m \big(M_{s,m}\big(y,w,\gamma^{\top}w \big)f_{\gamma_{*}}\big(\gamma^{\top}w\big) \big)\Big\| \nonumber \\
&= \sup_{\{y, w, \gamma_{*}\}} \bigg\| \frac{1}{n} \sum_{i=1}^n \xi_{i,s,m}(y,w,\gamma) \bigg\|= o_p \Bigg( \sqrt{\frac{\log n}{nh^{2m+1}}} \Bigg) + O(h^2) \label{eq:lemma1_4}
\end{align}
and
\begin{align}
&\sup_{\{\ell,x, A_{d*}\}} \Big\| \partial_{\vvec{(A_{d*})}}^m N_{s,\ell,h_{d}}(A_{d}^{\top}x; A_{d}) - \partial_{u}^m \big(M_{s,m,\ell}\big(x,A_{d}^{\top}x \big)f_{A_{d*}}\big(A_{d}^{\top}x\big) \big)\Big\| \nonumber \\
& = \sup_{\{\ell, x, A_{d*}\}} \bigg\| \frac{1}{n} \sum_{i=1}^n \xi_{i,s,m,\ell}(x,A_{d})\bigg\| = o_p \Bigg( \sqrt{\frac{\log n}{nh_{d}^{2m+d}}} \Bigg) + O(h_{d}^r),\label{eq:lemma2_4}
\end{align} 
where
\begin{align*}
&\xi_{i,s,m}(y,w,\gamma)  =  \partial_{\gamma_{*}}^m \big( \big( I(Y_i \leq y) \big)^s K_{2,h} \big(\gamma^{\top}W_i -\gamma^{\top}w  \big) \big) - \partial_{v}^m \big(M_{s,m}\big(y,w,\gamma^{\top}w \big)f_{\gamma_{*}}\big(\gamma^{\top}w\big) \big)\\
\text{and}&\\
&\xi_{i,s,m,\ell}(x,A_{d})  =  \partial_{\vvec{(A_{d*})}}^m \big( \big( I(C_i = \ell) \big)^s \mathcal{K}_{r,h_{d}} \big(A_{d}^{\top}X_i -A_{d}^{\top}x  \big) \big) - \partial_{u}^m \big(M_{s,m,\ell}\big(x,A_{d}^{\top}x \big)f_{A_{d*}}\big(A_{d}^{\top}x\big) \big)
\end{align*}
for $i = 1, \dots, n$, $s = 0, 1$, $m = 0, 1, 2$, and $\ell = 1, \dots, k-1$.

By Taylor expansion and the properties in (\ref{eq:lemma1_4}) and (\ref{eq:lemma2_4}), we obtain
\begin{align}
&\frac{\partial_{\gamma_{*}}^m N_{s,h}(y, \gamma^{\top}w; \gamma)}{N_{0,h}\big( y, \gamma^{\top}w; \gamma\big)} = \frac{\partial_{v}^m \big(M_{s,m}\big(y,w,\gamma^{\top}w \big)f_{\gamma_{*}}\big(\gamma^{\top}w\big) \big)}{f_{\gamma_{*}}(\gamma^{\top}w)} + \frac{1}{n} \sum_{i=1}^n \xi_{i,s,m}^{*}(y,w,\gamma) + r^{*}_{s,m}(y,w,\gamma)  \label{eq:lemma1_5}
\end{align}
and
\begin{align}
&\frac{\partial_{\vvec{(A_{d*})}}^m N_{s,\ell,h_{d}}\big(A_{d}^{\top}x; A_{d}\big)}{N_{0,\ell,h_{d}}\big(A_{d}^{\top}x; A_{d}\big)} =  \frac{\partial_{u}^m \big(M_{s,m,\ell}\big(A_{d}^{\top}x; A_{d} \big)f_{A_{d*}}\big(A_{d}^{\top}x\big) \big)}{f_{A_{d*}}(A_{d}^{\top}x)} + \frac{1}{n} \sum_{i=1}^n \xi_{i,s,m,\ell}^{*}(x, A_{d})\nonumber\\
&\hspace{1.85in}+ r^{*}_{s,m,\ell}(x,A_{d}),  \label{eq:lemma2_5}
\end{align}
where
\begin{align*}
\xi_{i,s,m}^{*}(y,w,\gamma) = \frac{\xi_{i,s,m}(y,w,\gamma)}{f_{\gamma_{*}}(\gamma^{\top}w)} - \frac{\partial_{v}^m \big(M_{s,m}\big(y,w,\gamma^{\top}w \big)f_{\gamma_{*}}\big(\gamma^{\top}w\big) \big) \xi_{i,0,0}(y,w,\gamma)}{f_{\gamma_{*}}^2(\gamma^{\top}w)}
\end{align*}
and
\begin{align*}
\xi_{i,s,m,\ell}^{*}(x,A_{d}) = \frac{\xi_{i,s,m,\ell}(x,A_{d})}{f_{A_{d*}}(A_{d}^{\top}x)} - \frac{\partial_{u}^m \big(M_{s,m,\ell}\big(A_{d}^{\top}x; A_{d} \big)f_{A_{d*}}\big(A_{d}^{\top}x\big) \big) \xi_{i,0,0,\ell}(x,A_{d})}{f_{A_{d*}}^2(A_{d}^{\top}x)}
\end{align*}
for $i = 1, \dots, n$, $s = 0,1$, $m = 0, 1, 2$, and $\ell = 1, \dots, k-1$, with
\begin{align*}
&\sup_{\{y, w, \gamma_{*}\}} \bigg\| \frac{1}{n} \sum_{i=1}^n \xi_{i,s,m}^{*}(y,w,\gamma) \bigg\| = o_p \Bigg( \sqrt{\frac{\log n}{nh^{2m+1}}} \Bigg) + O(h^2),\\
&\sup_{\{\ell,x, A_{d*}\}} \bigg\| \frac{1}{n} \sum_{i=1}^n \xi_{i,s,m,\ell}^{*}(x,A_{d}) \bigg\| = o_p \Bigg( \sqrt{\frac{\log n}{nh_{d}^{2m+d}}} \Bigg) + O(h_{d}^r),\\
&\sup_{\{y, w, \gamma_{*}\}} \|r^{*}_{s,m}(y,w,\gamma)\| = o_p \bigg( {\frac{\log n}{nh^{2m+1}}} \bigg) + O(h^4),\\
&\sup_{\{\ell,x, A_{d*}\}} \|r^{*}_{s,m,\ell}(x,A_{d})\| = o_p \bigg( {\frac{\log n}{nh_{d}^{2m+d}}} \bigg) + O(h_{d}^{2r}).
\end{align*}
Combining (\ref{eq:lemma1_5}) and (\ref{eq:lemma2_5}) with the expressions of $\partial_{\gamma_{*}}^m \widehat{G}_h(y, \gamma^{\top}w; \gamma)$
and $\partial_{\vvec{(A_{d*})}}^m \widehat{\pi}_{\ell,h_{d}}(A_{d}^{\top}x; A_{d})$ in Appendix \ref{sec:A.Notations}, we establish the following results for $m=0,1,2$:
\begin{align}
\widehat{G}_h(y, \gamma^{\top}w; \gamma)  =& \frac{M_{1,0}\big(y,w,\gamma^{\top}w \big)f_{\gamma_{*}}\big(\gamma^{\top}w\big)}{f_{\gamma_{*}}(\gamma^{\top}w)} + \frac{1}{n} \sum_{i=1}^n \xi_{i,1,0}^{*}(y,w,\gamma) + r^{*}_{1,0}(y,w,\gamma) \nonumber \\
\stackrel{\triangle}{=} &G^{[0]}(y, w, \gamma^{\top}w) + \frac{1}{n}\sum_{i=1}^n \xi^{[0]}_i(y,w,\gamma) + r^{[0]}(y,w,\gamma), \label{eq:lemma1_6}\\
\partial_{\gamma_{*}}\widehat{G}_h(y, \gamma^{\top}w; \gamma) =& \sum_{s=0}^1 \bigg( - \bigg( G^{[0]}(y, w, \gamma^{\top}w) + \frac{1}{n}\sum_{i=1}^n \xi^{[0]}_i(y,w,\gamma) + r^{[0]}(y,w,\gamma) \bigg) \bigg)^{1-s} \nonumber\\
& \bigg( \frac{\partial_v \big( M_{s,1}(y,w,\gamma^{\top}w) f_{\gamma_{*}}(\gamma^{\top}w)\big)}{f_{\gamma_{*}}(\gamma^{\top}w)} + \frac{1}{n} \sum_{i=1}^n \xi_{i,s,1}^*(y,w,\gamma) + r_{s,1}^*(y,w,\gamma) \bigg) \nonumber\\
\stackrel{\triangle}{=}& G^{[1]}(y, w, \gamma^{\top}w) + \frac{1}{n}\sum_{i=1}^n \xi^{[1]}_i(y,w,\gamma) + r^{[1]}(y,w,\gamma), \label{eq:lemma1_7}\\
\partial_{\gamma_{*}}^2\widehat{G}_h(y, \gamma^{\top}w; \gamma) =& \sum_{s=0}^1 \bigg( - \bigg( G^{[0]}(y, w, \gamma^{\top}w) + \frac{1}{n}\sum_{i=1}^n \xi^{[0]}_i(y,w,\gamma) + r^{[0]}(y,w,\gamma) \bigg) \bigg)^{1-s} \nonumber\\
& \bigg( \frac{\partial_v \big( M_{s,2}(y,w,\gamma^{\top}w) f_{\gamma_{*}}(\gamma^{\top}w)\big)}{f_{\gamma_{*}}(\gamma^{\top}w)} + \frac{1}{n} \sum_{i=1}^n  \xi_{i,s,2}^*(y,w,\gamma) + r_{s,2}^*(y,w,\gamma) \bigg) \nonumber\\
&-\bigg( \frac{\partial_v \big( M_{0,1}(y,w,\gamma^{\top}w) f_{\gamma_{*}}(\gamma^{\top}w)\big)}{f_{\gamma_{*}}(\gamma^{\top}w)} + \frac{1}{n}\sum_{i=1}^n \xi_{i,0,1}^*(y,w,\gamma) + r_{0,1}^*(y,w,\gamma)  \bigg) \nonumber\\
& \bigg(G^{[1]}(y, w, \gamma^{\top}w) + \frac{1}{n}\sum_{i=1}^n \xi^{[1]}_i(y,w,\gamma) + r^{[1]}(y,w,\gamma)\bigg)^{\top} \nonumber\\
&- \bigg(G^{[1]}(y, w, \gamma^{\top}w) + \frac{1}{n}\sum_{i=1}^n \xi^{[1]}_i(y,w,\gamma) + r^{[1]}(y,w,\gamma)\bigg) \nonumber
\end{align}
\begin{align}
&\bigg( \frac{\partial_v \big( M_{0,1}(y,w,\gamma^{\top}w) f_{\gamma_{*}}(\gamma^{\top}w)\big)}{f_{\gamma_{*}}(\gamma^{\top}w)} + \frac{1}{n}\sum_{i=1}^n \xi_{i,0,1}^*(y,w,\gamma) + r_{0,1}^*(y,w,\gamma)  \bigg)^{\top} \nonumber \\
\stackrel{\triangle}{=}& G^{[2]}(y, w, \gamma^{\top}w) + \frac{1}{n}\sum_{i=1}^n \xi^{[2]}_i(y,w,\gamma) + r^{[2]}(y,w,\gamma), \label{eq:lemma1_8}\\
\widehat{\pi}_{\ell,h_{d}}( A_{d}^{\top}x; A_{d})  =& \frac{M_{1,0,\ell}\big(x, A_{d}^{\top}x \big)f_{A_{d*}}\big(A_{d}^{\top}x\big)}{f_{A_{d*}}(A_{d}^{\top}x)} + \frac{1}{n} \sum_{i=1}^n \xi_{i,1,0,\ell}^{*}(x,A_{d}) + r^{*}_{1,0,\ell}(x,A_{d}) \nonumber \\
\stackrel{\triangle}{=} &\pi^{[0]}_{\ell}(x, A_{d}^{\top}x) + \frac{1}{n}\sum_{i=1}^n \xi^{[0]}_{i,\ell}(x,A_{d}) + r^{[0]}_{\ell}(x,A_{d}), \label{eq:lemma2_6}\\
\partial_{\vvec{(A_{d*})}}\widehat{\pi}_{\ell,h_{d}}( A_{d}^{\top}x; A_{d}) =& \sum_{s=0}^1 \bigg( - \bigg( \pi^{[0]}_{\ell}(x, A_{d}^{\top}x) + \frac{1}{n}\sum_{i=1}^n \xi^{[0]}_{i,\ell}(x,A_{d}) + r^{[0]}_{\ell}(x,A_{d}) \bigg) \bigg)^{1-s} \nonumber\\
& \bigg( \frac{\partial_u \big( M_{s,1,\ell}(x, A_{d}^{\top}x) f_{A_{d*}}(A_{d}^{\top}x)\big)}{f_{A_{d*}}(A_{d}^{\top}x)} + \frac{1}{n} \sum_{i=1}^n \xi_{i,s,1,\ell}^*(x,A_{d}) + r_{s,1,\ell}^*(x,A_{d}) \bigg) \nonumber\\
\stackrel{\triangle}{=} &\pi_{\ell}^{[1]}(x, A_{d}^{\top}x) + \frac{1}{n}\sum_{i=1}^n \xi^{[1]}_{i,\ell}(x,A_{d}) + r^{[1]}_{\ell}(x,A_{d}), \label{eq:lemma2_7}
\end{align}
and
\begin{align}
\partial_{\vvec{(A_{d*})}}^2\widehat{\pi}_{\ell,h_{d}}( A_{d}^{\top}x; A_{d}) = &\sum_{s=0}^1 \bigg( - \bigg( \pi^{[0]}_{\ell}(x, A_{d}^{\top}x) + \frac{1}{n}\sum_{i=1}^n \xi^{[0]}_{i,\ell}(x,A_{d}) + r^{[0]}_{\ell}(x,A_{d}) \bigg) \bigg)^{1-s} \nonumber\\
& \bigg( \frac{\partial_u \big( M_{s,2,\ell}(x, A_{d}^{\top}x) f_{A_{d*}}(A_{d}^{\top}x)\big)}{f_{A_{d*}}(A_{d}^{\top}x)} + \frac{1}{n} \sum_{i=1}^n  \xi_{i,s,2,\ell}^*(x,A_{d}) + r_{s,2,\ell}^*(x,A_{d}) \bigg) \nonumber\\
\hspace{1.2in}&-\bigg( \frac{\partial_u \big( M_{0,1,\ell}(x, A_{d}^{\top}x) f_{A_{d*}}(A_{d}^{\top}x)\big)}{f_{A_{d*}}(A_{d}^{\top}x)} + \frac{1}{n}\sum_{i=1}^n \xi_{i,0,1,\ell}^*(x,A_{d}) + r_{0,1,\ell}^*(x,A_{d})  \bigg) \nonumber\\
&  \bigg(\pi_{\ell}^{[1]}(x, A_{d}^{\top}x) + \frac{1}{n}\sum_{i=1}^n \xi^{[1]}_{i,\ell}(x,A_{d}) + r^{[1]}_{\ell}(x,A_{d})\bigg)^{\top} \nonumber\\
&- \bigg(\pi_{\ell}^{[1]}(x, A_{d}^{\top}x) + \frac{1}{n}\sum_{i=1}^n \xi^{[1]}_{i,\ell}(x,A_{d}) + r^{[1]}_{\ell}(x,A_{d})\bigg) \nonumber\\
&\bigg( \frac{\partial_v \big( M_{0,1,\ell}(x, A_{d}^{\top}x) f_{A_{d*}}(A_{d}^{\top}x)\big)}{f_{A_{d*}}(A_{d}^{\top}x)} + \frac{1}{n}\sum_{i=1}^n \xi_{i,0,1,\ell}^*(x,A_{d}) + r_{0,1,\ell}^*(x,A_{d})  \bigg)^{\top} \nonumber \\
\stackrel{\triangle}{=} &\pi^{[2]}_{\ell}(x, A_{d}^{\top}x) + \frac{1}{n}\sum_{i=1}^n \xi^{[2]}_{i,\ell}(x,A_{d}) + r^{[2]}_{\ell}(x,A_{d}), \label{eq:lemma2_8}
\end{align}
where
\begin{align*}
G^{[0]}(y,w,v) =& M_{1,0}(y,w,v), ~G^{[1]}(y,w,v) = \sum_{s = 0}^1 \big( - G^{[0]}(y,w,v)\big)^{1-s} \frac{\partial_v(M_{s,1}(y,w,v)f_{\gamma_{*}}(v))}{f_{\gamma_{*}}(v)}, \\
G^{[2]}(y,w,v) =&\sum_{s = 0}^1 \big( - G^{[0]}(y,w,v)\big)^{1-s} \frac{\partial_v(M_{s,2}(y,w,v)f_{\gamma_{*}}(v))}{f_{\gamma_{*}}(v)} 
- \frac{\partial_v(M_{0,1}(y,w,v)f_{\gamma_{*}}(v))}{f_{\gamma_{*}}(v)}\\
&\big(G^{[1]}(y,w,v)\big)^{\top} - { G^{[1]}(y,w,v)} \bigg( \frac{\partial_v(M_{0,1}(y,w,v)f_{\gamma_{*}}(v))}{f_{\gamma_{*}}(v)} \bigg)^{\top},
\end{align*}
\begin{align*}
\pi_{\ell}^{[0]}(x,u) = &M_{1,0,\ell}(x,u), ~\pi_{\ell}^{[1]}(x,u) = \sum_{s = 0}^1 \big( - \pi_{\ell}^{[0]}(x,u)\big)^{1-s} \frac{\partial_u(M_{s,1,\ell}(x,u)f_{A_{d*}}(u))}{f_{A_{d*}}(u)},
\end{align*}
and
\begin{align*}
\pi_{\ell}^{[2]}(x,u) =&\sum_{s = 0}^1 \big( - \pi_{\ell}^{[0]}(x,u)\big)^{1-s} \frac{\partial_u(M_{s,2,\ell}(x,u)f_{A_{d*}}(u))}{f_{A_{d*}}(u)} 
- \frac{\partial_u(M_{0,1,\ell}(x,u)f_{A_{d*}}(u))}{f_{A_{d*}}(u)}\\
&\big(\pi_{\ell}^{[1]}(x,u)\big)^{\top} - { \pi_{\ell}^{[1]}(x,u)} \bigg( \frac{\partial_u(M_{0,1,\ell}(x,u)f_{A_{d*}}(u))}{f_{A_{d*}}(u)} \bigg)^{\top},\ell=0,\dots,k-1,
\end{align*}
with
\begin{align*}
\sup_{\{y, w, \gamma_{*}\}} \Big\| r^{[m]}(y,w,\gamma) \Big\| = o_p \bigg( {\frac{\log n}{nh^{2m+1}}} \bigg) + O(h^4)
\text{ and }\sup_{\{\ell , A_{d*}\}} \Big\| r^{[m]}_{\ell}(x,A_{d}) \Big\| = o_p \bigg( {\frac{\log n}{nh_{d}^{2m+d}}} \bigg) + O(h_{d}^{2r}).
\end{align*}
The proof is completed.

\end{proof}

\begin{aplemma} 
\label{Thm2.1}
Under assumptions {A1} -- {A5}, 
\begin{align*}
\widehat{\gamma}_{*}\stackrel{p}{\longrightarrow} \gamma^{\text{o}}_{*} \text{ and }\sqrt{n}\big(\widehat{\gamma}_{*} - \gamma^{\text{o}}_{*}\big)\stackrel{d}{\longrightarrow} N_{kp-2}(0,V_{2}^{-1}V_{1}V_{2}^{-1}) \text{ as }n \longrightarrow \infty.
\end{align*}
\end{aplemma}

\begin{proof}

A direct decomposition of $\textsc{psis}(\gamma,h)$ yields
\begin{align}
\frac{2}{n} \textsc{psis}(\gamma,h) = & \frac{1}{n} \sum_{i=1}^n \int \varepsilon^2_y\big(Y_i, \gamma^{\top}W_i\big) d \widehat{F}(y) 
 + \frac{1}{n} \sum_{i=1}^n \int \Big(G^{[0]}\big(y, W_i,\gamma^{\top}W_i\big) - \widehat{G}_h^{-i}\big(y, \gamma^{\top}W_i; \gamma\big) \Big)^2 d\widehat{F}(y)\nonumber \\
& + \frac{2}{n} \sum_{i=1}^n \int \varepsilon_y\big(Y_i, \gamma^{\top}W_i\big) \Big(G^{[0]}\big(y, W_i, \gamma^{\top}W_i\big) - \widehat{G}_h^{-i}\big(y, \gamma^{\top}W_i; \gamma\big) \Big) d \widehat{F}(y) \nonumber \\
\stackrel{\triangle}{=} & I_{1n}(\gamma) + I_{2n}(\gamma) + I_{3n}(\gamma),  \label{eq:thm1_1}
\end{align}
where $\varepsilon_y\big(Y, \gamma^{\top}W\big) = I(Y \leq y) - G^{[0]}\big(y, W, \gamma^{\top}W\big)$. Since $\varepsilon_y\big(y, \gamma^{\top}w\big)$ is of bounded variation in $(y,w)$ uniformly over $\gamma_{*} \in \Gamma_n$,
$\{ \varepsilon^2_{Y_j}\big(Y_i, \gamma^{\top}W_i\big): \gamma_{*} \in \Gamma_n \}$ is a Euclidean class by Lemma 22 in \cite{nolan1987u}. Moreover, noting that $I_{1n}(\gamma) - \int \E\big[\varepsilon^2_y\big(Y,$ $\gamma^{\top}W\big)\big] dF(y)$ coincides with the zero-mean $U$-process $\sum_{i = 1}^n \sum_{j = 1}^n \varepsilon^2_{Y_j}\big(Y_i, \gamma^{\top}W_i\big) / n^2- \int \E\big[\varepsilon^2_y\big(Y,\gamma^{\top}W\big)\big] dF(y)$, Corollary 7 in \cite{sherman1994maximal} implies that 
\begin{align}
\sup_{\gamma_{*}} \Big| I_{1n}(\gamma) - \int \E\big[\varepsilon^2_y\big(Y, \gamma^{\top}W\big)\big] dF(y) \Big| =  O_p\Big(\frac{1}{\sqrt{n}} \Big). \label{eq:thm1_2}
\end{align}
By assumption {A1}, the uniform boundedness of $\varepsilon_y\big(Y, \gamma^{\top}W\big)$, and the first assertion in Lemma \ref{A:Lemma_1}, we obtain 
\begin{align}
\sup_{\gamma_{*}} |I_{2n}(\gamma)| = o_p(1) \text{ and } \sup_{\gamma_{*}} |I_{3n}(\gamma)| = o_p(1). \label{eq:thm1_3}
\end{align}
Combining (\ref{eq:thm1_1})--(\ref{eq:thm1_3}) and invoking the triangle inequality yield
\begin{align}
\sup_{\gamma_{*}} \bigg| \frac{2}{n} \textsc{psis}(\gamma,h)  - \int \E\Big[\varepsilon^2_y\big(Y, \gamma^{\top}W\big)\Big]dF(y) \bigg| = o_p (1). \label{eq:thm1_4}
\end{align}
Furthermore, under the condition that
\begin{align*}
\int\E\big[\varepsilon^2_y\big(Y, \gamma^{\top}W\big)\big]dF(y) > \int\E\big[\varepsilon^2_y\big(Y, \gamma^{\text{o} \top}W\big)\big]dF(y) \text{ for } \gamma \neq \gamma^{\text{o}},
\end{align*}
it follows that $\gamma^{\text{o}}$ uniquely minimizes $\int\E[\varepsilon^2_y(Y, \gamma^{\top}W)]dF(y)$.
Thus, (\ref{eq:thm1_4}) implies that $2 \textsc{psis}(\widehat{\gamma},h) /n$ $>2 \textsc{psis}({\gamma^{\text{o}}},h) / n$ with probability converging to one.
Combined with $\textsc{psis}(\gamma^{\text{o}},h) \geq \textsc{psis}(\widehat{\gamma},h)$, the continuity of $\textsc{psis}({\gamma},h)$ in $\gamma$ ensures the consistency of $\widehat{\gamma}_{*}$ to $\gamma^{\text{o}}_{*}$.

Define $\textsc{ps}(\gamma)=\partial_{\gamma_{*}} 2 \textsc{psis}({\gamma}, h)/n$ and $\textsc{pi}(\gamma)=\partial_{\gamma_{*}}^2 2 \textsc{psis}({\gamma}, h)/n$.
For an arbitrary constant vector $\nu\in \mathbb{R}^{kp - 2}$, a first-order Taylor expansion of $\nu^{\top}\textsc{ps}(\widehat{\gamma})$ around $\gamma^{\text{o}}_{*}$ yields
\begin{align}
\nu^{\top} (\sqrt{n} \textsc{ps}(\gamma^{\text{o}})) + \nu^{\top} (\textsc{pi}(\bar{\gamma})) \sqrt{n}(\widehat{\gamma}_{*} - \gamma^{\text{o}}_{*}) = 0, \label{eq:thm1_5}
\end{align}
where $\bar{\gamma}$ lies on the line segment between $\widehat{\gamma}$ and $\gamma^{\text{o}}$.
The pseudo score vector $\textsc{ps} (\gamma^\text{o})$ in (\ref{eq:thm1_5}) admits the decomposition 
\begin{align}
\textsc{ps} (\gamma^\text{o}) =& \sum_{\ell_1=0}^1 \sum_{\ell_2=0}^1 \frac{2}{n} \sum_{i=1}^n \int (-1)^{\ell_1}\big( \varepsilon_y\big(Y_i, \gamma^{\text{o} \top}W_i\big) \big)^{1-\ell_1} \Big(\widehat{G}^{-i}_h \big(y, \gamma^{\text{o} \top} W_i; \gamma^{\text{o}}\big)  -  G\big(y, \gamma^{\text{o} \top} W_i \big)\Big)^{\ell_1} \hspace{0.3in} \nonumber \\
& \Big( G^{[1]} \big(y, W_i, \gamma^{\text{o} \top} W_i \big) \Big)^{1-\ell_2} \Big( \partial_{\gamma^{\text{o}}_{*}} \widehat{G}^{-i}_h \big(y, \gamma^{\text{o} \top} W_i; \gamma^{\text{o}}\big) - G^{[1]}\big(y, W_i, \gamma^{\text{o} \top} W_i \big) \Big)^{\ell_2} d \widehat{F}(y) \nonumber \\
\stackrel{\triangle}{=}& \sum_{\ell_1=0}^1 \sum_{\ell_2=0}^1 \textsc{ps}_{\ell_1 \ell_2}(\gamma^{\text{o}}). \label{eq:thm1_6}
\end{align}
Under assumption {A1}, together with the uniform boundedness of $\varepsilon_y\big(Y, \gamma^{\text{o} \top}W\big)$ and $G^{[1]}\big(y, \gamma^{\text{o} \top} w\big)$, and by invoking the first assertion in Lemma \ref{A:Lemma_1}, it follows that
\begin{align}
\sqrt{n} \|\textsc{ps}_{11}(\gamma^{\text{o}})\| & = \sqrt{n}\bigg\| \frac{2}{n} \sum_{i=1}^n \int \prod_{m=0}^1\Big( \partial_{\gamma^{\text{o}}_{*}}^m \widehat{G}^{-i}_h \big(y, \gamma^{\text{o} \top} W_i; \gamma^{\text{o}}\big) - G^{[m]}\big(y, W_i, \gamma^{\text{o} \top} W_i \big) \Big) d \widehat{F}(y) \bigg\|\nonumber \\
& \leq 2\sqrt{n} \prod_{m=0}^1 \sup_{\{y,w\}} \Big\| \partial_{\gamma^{\text{o}}_{*}}^m \widehat{G}^{-i}_h \big(y, \gamma^{\text{o} \top} w; \gamma^{\text{o}}\big) - G^{[m]}\big(y, w, \gamma^{\text{o} \top} w \big)  \Big\|  \nonumber \\
&= 2\sqrt{n} \Bigg(o_p\Bigg(\sqrt{\frac{\log n}{nh}}\Bigg) + O\big(h^2\big)  \Bigg) \Bigg( o_p\Bigg(\sqrt{\frac{\log n}{nh^3}}\Bigg) + O\big(h^2\big) \Bigg)= o_p(1), \label{eq:thm1_7} \\
\sqrt{n} \textsc{ps}_{01}(\gamma^{\text{o}}) = & \sqrt{n} \int \frac{2}{n} \sum_{i = 1}^n \varepsilon_y\big(Y_i, \gamma^{\text{o} \top}W_i\big) \Bigg( \frac{1}{n-1} \sum_{\{j: j \neq i\}} \xi^{[1]}_j(y, W_i, \gamma^{\text{o}}) + r^{[1]}(y, W_i, \gamma^{\text{o}}) \Bigg) d\widehat{F}(y) \nonumber \\
= & \sqrt{n} \int \frac{2}{n(n-1)} \sum_{i \neq j} \varepsilon_y\big(Y_i, \gamma^{\text{o} \top}W_i\big) \bar{\xi}^{[1]}_j(y, W_i, \gamma^{\text{o}}) d\widehat{F}(y) \nonumber\\
&+ \sqrt{n} \int \frac{2}{n(n-1)} \sum_{i \neq j} \varepsilon_y\big(Y_i, \gamma^{\text{o} \top}W_i\big) \E\Big[{\xi}^{[1]}_j(y, W_i, \gamma^{\text{o}})|W_i\Big] d\widehat{F}(y) + o_p(1),\label{eq:thm1_7.5} 
\end{align}
and
\begin{align}
\sqrt{n} \textsc{ps}_{10}(\gamma^{\text{o}}) = & \sqrt{n} \int -\frac{2}{n} \sum_{i = 1}^n G^{[1]}\big(y, W_i, \gamma^{\text{o} \top} W_i\big) \Bigg( \frac{1}{n-1} \sum_{\{j: j \neq i\}} \xi^{[0]}_j(y, W_i, \gamma^{\text{o}}) + r^{[0]}(y, W_i, \gamma^{\text{o}}) \Bigg) d\widehat{F}(y) \nonumber \\
= & \sqrt{n} \int -\frac{2}{n(n-1)} \sum_{i \neq j} G^{[1]}\big(y, W_i, \gamma^{\text{o} \top} W_i\big) \bar{\xi}^{[0]}_j(y, W_i, \gamma^{\text{o}}) d\widehat{F}(y) \nonumber 
\end{align}
\begin{align}
&- \sqrt{n} \int \frac{2}{n(n-1)} \sum_{i \neq j} G^{[1]}\big(y, W_i, \gamma^{\text{o} \top} W_i\big) \E\Big[{\xi}^{[0]}_j(y, W_i, \gamma^{\text{o}})|W_i\Big] d\widehat{F}(y)  + o_p(1), \label{eq:thm1_8}
\end{align}
where $\bar{\xi}^{[m]}_j(y, W_i, \gamma^{\text{o}}) = {\xi}^{[m]}_j(y, W_i, \gamma^{\text{o}}) - \E[{\xi}^{[m]}_j(y, W_i, \gamma^{\text{o}}) | W_i]$, $i, j = 1, \dots, n$, $m = 0, 1$.

A direct calculation yields the following expressions:
\begin{align*}
&\xi^{[0]}_{j}(y,w,\gamma^{\text{o}}) = \frac{\big( I(Y_j \leq y) - G\big( y, \gamma^{\text{o}\top} w\big) \big) K_{2,h}(\gamma^{\text{o}\top} (W_{j} - w)) }{f_{\gamma^{\text{o}}_{*}(\gamma^{\text{o}\top}W)}}, j = 1, \dots, n,\\
&G^{[1]}\big(y, w, \gamma^{\text{o} \top} w \big) = \big( \partial_v G \big(y, \gamma^{\text{o} \top} w \big)\big) \big( H(\gamma^{\text{o} \top} w) - w_{*} \big) \text{ with } H(v) = \E[W_{*}|\gamma^{\text{o} \top}W = v], \\
&\E\Big[ \varepsilon_y\big(Y_i, \gamma^{\text{o} \top}W_i\big) \bar{\xi}^{[1]}_j(y, W_i, \gamma^{\text{o}}) | (W_i, Y_i) \Big] = 0, \E\Big[ G^{[1]}\big(y, W_i, \gamma^{\text{o} \top} W_i \big) \bar{\xi}^{[0]}_j(y, W_i, \gamma^{\text{o}}) | (W_i, Y_i) \Big] = 0,\\
&\E\Big[ \varepsilon_y\big(Y_i, \gamma^{\text{o} \top}W_i\big) \bar{\xi}^{[1]}_j(y, W_i, \gamma^{\text{o}}) | (W_j, Y_j) \Big] = \E\Big[ \E\big[ \varepsilon_y\big(Y_i, \gamma^{\text{o} \top}W_i\big) | W_i\big] \bar{\xi}^{[1]}_j(y, W_i, \gamma^{\text{o}}) | (W_j, Y_j) \Big] = 0,
\end{align*}
\begin{align*}
\text{and}&\\
&\E \Big[ G^{[1]}\big(y, W_i, \gamma^{\text{o} \top} W_i \big) \bar{\xi}^{[0]}_j(y, W_i, \gamma^{\text{o}}) | (W_j, Y_j) \Big] \hspace{5in}\\
&=\E \Big[ \big( \partial_v G \big(y, \gamma^{\text{o} \top} W_i \big)\big) \big( H(\gamma^{\text{o} \top} W_i) - \E[W_{*i} | \gamma^{\text{o} \top} W_i] \big) \bar{\xi}^{[0]}_j(y, W_i, \gamma^{\text{o}}) | (W_j, Y_j) \Big] = 0.
\end{align*}
These identities imply that the second-order $U$-processes in  (\ref{eq:thm1_7.5}) and (\ref{eq:thm1_8}) are degenerate.
Together with the Euclidean classes
\begin{align*}
\big\{ \varepsilon_y(Y_i, \gamma^{\text{o} \top}W_i) \bar{\xi}^{[1]}_j(y, W_i, \gamma^{\text{o}}): y \in \mathcal{Y} \big\} 
\text{ and } \big\{ G^{[1]}(y, W_i, \gamma^{\text{o} \top} W_i)\bar{\xi}^{[0]}_j(y, W_i, \gamma^{\text{o}}): y \in \mathcal{Y} \big\},
\end{align*}
which are ensured by Lemma 22 in \cite{nolan1987u}, $\sup_{y}\E[(\varepsilon_y(Y_i, \gamma^{\text{o} \top}W_i)\bar{\xi}^{[1]}_j(y, W_i, \gamma^{\text{o}}) )^2]$ $= O(h^{-3})$, and $\sup_{y}\E[(G^{[1]}(y, W_i, \gamma^{\text{o} \top} W_i) \bar{\xi}^{[0]}_j(y, W_i, \gamma^{\text{o}}) )^2] = O(h^{-1})$, it follows from Corollary 4 in \cite{sherman1994maximal} that
\begin{align}
&\sup_y \Bigg\| \frac{2}{n(n-1)}\sum_{i \neq j} \varepsilon_y\big(Y_i, \gamma^{\text{o} \top}W_i\big) \bar{\xi}^{[1]}_j(y, W_i, \gamma^{\text{o}}) \Bigg\| = O_p\bigg(\frac{1}{n\sqrt{h^{3}}}\bigg) \label{eq:thm1_9.0}
\end{align}
and
\begin{align}
&\sup_y \Bigg\| \frac{2}{n(n-1)}\sum_{i \neq j} G^{[1]}\big(y, W_i, \gamma^{\text{o} \top} W_i \big) \bar{\xi}^{[0]}_j(y, W_i, \gamma^{\text{o}}) \Bigg\| = O_p\bigg(\frac{1}{n\sqrt{h}}\bigg).  \label{eq:thm1_9}
\end{align}
Moreover, Theorem 2.5 in \cite{Kosorok2008} implies that
\begin{align}
 \sup_{y} \bigg\| \frac{1}{n} \sum_{i=1}^n  \varepsilon_y\big(Y_i, \gamma^{\text{o} \top}W_i\big) \E\Big[{\xi}^{[1]}_j(y, W_i, \gamma^{\text{o}}) | W_i\Big] \bigg\| = O\bigg( \frac{h^2}{\sqrt{n}} \bigg) \label{eq:thm1_9.5.0}
\end{align}
and
\begin{align}
 \sup_{y} \bigg\| \frac{1}{n} \sum_{i=1}^n  G^{[1]}\big(y, W_i, \gamma^{\text{o} \top} W_i \big) \E\Big[{\xi}^{[0]}_j(y, W_i, \gamma^{\text{o}}) | W_i\Big] \bigg\| = O\bigg( \frac{h^2}{\sqrt{n}} \bigg). \label{eq:thm1_9.5}
\end{align}
Under assumption {A1}, substituting  (\ref{eq:thm1_9.0}) and (\ref{eq:thm1_9.5.0}) into (\ref{eq:thm1_7.5}), and substituting (\ref{eq:thm1_9}) and (\ref{eq:thm1_9.5}) into (\ref{eq:thm1_8}), we conclude that
\begin{align}
\sqrt{n} \textsc{ps}_{01}(\gamma^{\text{o}}) = o_p(1) \text{ and } \sqrt{n} \textsc{ps}_{10}(\gamma^{\text{o}}) = o_p(1). \label{eq:thm1_10}
\end{align}
Since $\textsc{ps}_{00}(\gamma^{\text{o}})$ is a degenerate $U$-statistic and $\E \big[ \big( \varepsilon_{Y_2}\big(Y_1, \gamma^{\text{o} \top}W_1\big) G^{[1]}\big(Y_2, W_1, \gamma^{\text{o} \top}W_1 \big) \big)^{\otimes 2}\big]$ $= V_1$, the limiting distribution follows from Theorem 8.1 in \cite{hoeffding1948probability}:
\begin{align}
\sqrt{n} \nu^{\top} \textsc{ps}_{00}(\gamma^{\text{o}}) \stackrel{d}{\longrightarrow} N_{kp-2}(0, \nu^{\top}V_1 \nu). \label{eq:thm1_11}
\end{align}
Substituting (\ref{eq:thm1_7}), (\ref{eq:thm1_10}), and (\ref{eq:thm1_11}) into (\ref{eq:thm1_6}) yields
\begin{align}
\sqrt{n} \textsc{ps}(\gamma^{\text{o}})\stackrel{d}{\longrightarrow} N_{kp-2}(0, V_1). \label{eq:thm1_12}
\end{align}

The first assertion in Lemma \ref{A:Lemma_1}, together with the uniform boundedness of $G^{[m]}(y, \gamma^{\top}w)$ for $m = 0, 1, 2$, 
enables the derivation of the pseudo information matrix in (\ref{eq:thm1_5}) as
\begin{align}
\textsc{pi}(\gamma) = & \frac{2}{n} \int  \sum_{i=1}^n \Big( G^{[1]}\big(y, W_i, \gamma^{\top}W_i\big) \Big)^{\otimes 2} d\widehat{F}(y) + o_p\Bigg( \sqrt{\frac{\log n}{nh^3}} \Bigg) + O(h^2)  \nonumber\\
& - \frac{2}{n} \int \sum_{i=1}^n \varepsilon_y\big(Y_i, \gamma^{\top}W_i\big) G^{[2]} \big(y, W_i, \gamma^{\top}W_i\big)d\widehat{F}(y) + o_p\Bigg( \sqrt{\frac{\log n}{nh^5}} \Bigg) + O(h^2). \label{eq:thm1_13}
\end{align}
Under assumption {A1}, (\ref{eq:thm1_13}) and the consistency of $U$-statistics in \cite{hoeffding1961strong} imply
\begin{align}
\textsc{pi}(\gamma)  \stackrel{p}{\longrightarrow} E\bigg[\int \bigg(\big(G^{[1]}\big(y, W, \gamma^{\top}W\big)\big)^{\otimes 2} - \varepsilon_y\big(Y, \gamma^{\top}W\big) G^{[2]} \big(y, W, \gamma^{\top}W\big) \bigg) d F(y) \bigg].  \label{eq:thm1_14}
\end{align}
The continuity of $\textsc{pi}(\gamma)$ in $\gamma$, the consistency of $\widehat{\gamma}_{*}$ to $\gamma^{\text{o}}_*$, and the fact that the second term in (\ref{eq:thm1_14}) has zero expectation when $\gamma = \gamma^{\text{o}}$, together yield
\begin{align}
\textsc{pi}(\bar{\gamma}) \stackrel{p}{\longrightarrow} V_2. \label{eq:thm1_15}
\end{align}
Substituting (\ref{eq:thm1_12}) and (\ref{eq:thm1_15}) into (\ref{eq:thm1_5}), applying Slutsky's theorem, and invoking assumption {A5}, we 
establish the asymptotic normality of $\widehat{\gamma}_{*}$.
\end{proof}

\subsection{Proof of Theorem \ref{Thm3.1}}
\label{sec:A.1}
 
Define 
$\beta^{*}=(\beta^{*\top}_{1},\dots,\beta^{*\top}_{n})^{\top}$ and $\gamma^{*}=(\gamma^{*\top}_{1},\dots,\gamma^{*\top}_{k})^{\top}$,
where $\beta^{*}_{i}=\sum^{k}_{\ell=1}\gamma^{*}_{\ell}I(i\in\mathcal{G}_{\ell}^{\text{o}})$ and $\gamma^{*}_{\ell}=\sum_{\{i\in\mathcal{G}_{\ell}^{\text{o}}\}}\beta_{i}/|\mathcal{G}_{\ell}^{\text{o}}|$ for $i=1,\dots, n$ and $\ell\in\mathcal{C}$.
The oracle property of the SP estimator $\widehat{\beta}^{\lambda}$ holds if 
\begin{align*}
\textsc{psis}_{\textsc{sp}}\big(\beta,\gamma; \lambda\big) > \textsc{psis}_{\textsc{sp}}\big(\widehat{\beta}^{or},\widehat{\gamma}; \lambda\big)
\end{align*}
for all $\beta \in B_{n}$ and $\gamma\in\Gamma_{n}$ such that $\beta \neq \widehat{\beta}^{or}$ or $\gamma\neq \widehat{\gamma}$, with probability converging to one. To establish this, it suffices to show that with probability converging to one, 
\begin{align*}
\textsc{psis}_{\textsc{sp}}\big(\beta^{*},\gamma^{*}; \lambda\big)  >\textsc{psis}_{\textsc{sp}}\big(\widehat{\beta}^{or},\widehat{\gamma}; \lambda\big)\text{ and }\textsc{psis}_{\textsc{sp}}\big(\beta,\gamma; \lambda\big) \geq \textsc{psis}_{\textsc{sp}}\big(\beta^{*}
,\gamma^{*}; \lambda\big).
\end{align*} 
for all $\beta \in B_{n}$ and $\gamma\in\Gamma_{n}$ with $\beta \neq \widehat{\beta}^{or}$ or $\gamma\neq \widehat{\gamma}$.

By definition of the pair $(\beta^{*},\gamma^{*})$, the penalty term $\sum^{n}_{i=1}\min_{\ell}\big\|\beta^{*}_{(1)i}-\gamma^{*}_{(1)\ell}\big\|_{1}=0$. Thus, we have the identity
\begin{align}
\textsc{psis}_{\textsc{sp}}\big(\beta^{*},\gamma^{*}; \lambda\big)=
\textsc{psis}\big(\gamma^{*},h;\mathcal{G}^{\text{o}} \big). 
\label{eq:A.2.1}
\end{align}
Invoking the asymptotic normality of $\widehat{\gamma}_*$, $\gamma^{*} \in \Gamma_{n}$, and the convergence property in (\ref{eq:thm1_14}), a second-order Taylor expansion of $\textsc{psis} \big(\gamma^{*},h; \mathcal{G}^{\text{o}}\big)$ around $\widehat{\gamma}_*$ yields
\begin{align}
\frac{1}{n} \textsc{psis} \big(\gamma^{*},h; \mathcal{G}^{\text{o}}\big) - \frac{1}{n} \textsc{psis} \big(\widehat{\gamma}, h; \mathcal{G}^{\text{o}}\big) = \frac{1}{2} \big(\gamma^{*}_{*} - \widehat{\gamma}_*\big)^{\top} V_{2} \big(\gamma^{*}_{*} - \widehat{\gamma}_*\big) (1 + o_p(1)). \label{eq:A.2.2}
\end{align}
Combining (\ref{eq:A.2.1}) and (\ref{eq:A.2.2}) with assumption {A5}, we obtain
\begin{align}
\frac{1}{n} \textsc{psis}_{\textsc{sp}}\big (\beta^{*},\gamma^{*}; \lambda\big) > \frac{1}{n}\textsc{psis}_{\textsc{sp}} \big(\widehat{\beta}^{or},\widehat{\gamma}; \lambda\big) (1 + o_p(1)) \text{ for } \beta^{*} \in B_{n} \text{ with } \beta^{*} \neq \widehat{\beta}^{or}.
\label{eq:A.2.3}
\end{align}
A first-order Taylor expansion of the first term in $\textsc{psis}_{\textsc{sp}} \big(\beta,\gamma; \lambda\big)$ around $\beta = \beta^*$ gives
\begin{align}
\textsc{psis}_{\textsc{sp}}\big(\beta,\gamma;\lambda\big) - \textsc{psis}_{\textsc{sp}}\big(\beta^{*},\gamma^{*};\lambda\big)
&= \sum_{i = 1}^n \big(\textsc{ps}_i \big(\bar{\beta}\big)\big)^{\top} \big(\beta_i - \beta^*_i\big) + \lambda \sum^{n}_{i=1}\min_{\ell}\big\|\beta_{(1)i}-\gamma_{(1)\ell}\big\|_{1}\nonumber\\
&\stackrel{\triangle}{=}I_{1}+I_{2},  \label{eq:A.2.4}
\end{align}
where
\begin{align*}
\textsc{ps}_i(\beta)=- \int \big(I(Y_i \leq y) - \widehat{G}^{-i}_h \big(y, \beta^{\top}_i Z_i; \beta\big)\big) \big( \partial_{\beta_i} \widehat{G}^{-i}_h (y, \beta^{\top}_i Z_i; \beta) \big) d\widehat{F}(y)
\end{align*} 
and $\bar{\beta}_i = \alpha \beta_i + (1-\alpha) \beta^*_i$ for some $\alpha \in (0,1)$, $i = 1, \dots, n$. 

By the identity 
\begin{align*}
\sup_{i} \big\| \beta^*_{(1)i} - \beta^{\text{o}}_{(1)i} \big\|^2 = \max_{\ell} \big\| \gamma^{*}_{(1)\ell} - \gamma^{\text{o}}_{(1)\ell} \big\|^2
\end{align*}
and the triangle inequality, we have
\begin{align}
\sup_{i} \big\| \bar{\beta}_{(1)i} - \beta^{\text{o}}_{(1)i} \big\|  \leq \alpha \sup_{i}\big\| \beta_{(1)i} - \beta^{\text{o}}_{(1)i} \big\| + (1-\alpha) \sup_{i} \big\| \beta^*_{(1)i} - \beta^{\text{o}}_{(1)i} \big\| \leq \phi_n. \label{eq:A.2.5} 
\end{align} 
A direct calculation yields
\begin{align}
I_1 = \sum_{\ell = 1}^k \sum_{\{i \in \mathcal{G}_{\ell}\}} \big( \textsc{ps}_i \big(\bar{\beta}\big) \big)^{\top} \sum_{\{j \in \mathcal{G}_{\ell}\}} \frac{\big(\beta_i - \beta_j\big)}{|\mathcal{G}_{\ell}|}  = - \sum_{\ell = 1}^k \sum_{\{i,j \in \mathcal{G}_{\ell}, i < j\}} \frac{\big(\textsc{ps}_i\big(\bar{\beta}\big) - \textsc{ps}_j\big(\bar{\beta}\big)\big)^{\top}\big(\beta_j - \beta_i\big)}{|\mathcal{G}_{\ell}|}. \label{eq:A.2.6}
\end{align}
Following analogous arguments to those in the proof of (\ref{eq:lemma1_4}), and under assumptions {A2} and {A3}, we establish the uniform convergence 
\begin{align}
&\sup_{\{y,w\}} \bigg\| \frac{1}{nh^{m+1}} \sum_{i = 1}^n \big( I(Y_i \leq y) \big)^s K^{(m)}_2 \bigg( \frac{\gamma^{\text{o} \top} W_i - \gamma^{\text{o} \top} w}{h} \bigg)z^{\otimes m} - \partial_v^m \big( M^*_{s,m}\big(y, z, \gamma^{\text{o} \top}w \big) f_{\gamma^{\text{o}}_{*}}\big( \gamma^{\text{o} \top}w \big) \big) \bigg\| \nonumber \\
& = \sup_{\{y,w\}} \| \xi_{s,m}(y,w,\gamma^{\text{o}}) \| (1+o_p(1)) =  o_p \Bigg( \sqrt{\frac{\log n}{n h^{2m+1}}} \Bigg) + O \big( h^2 \big), \label{eq:A.2.7}
\end{align}
where 
$M^*_{s,m}(y, z, v) = \E \big[ \big( G\big(y, \gamma^{\text{o} \top} W \big) \big)^s | \gamma^{\top}W = v \big] (-z)^{\otimes m}$, $s, m = 0, 1.$
Combining (\ref{eq:A.2.5}), (\ref{eq:A.2.7}), and the uniform boundedness of $I(Y \leq y)$ and $K^{(m)}_2(u)$, we obtain
\begin{align}
&\frac{1}{nh^{m+1}} \sum_{\{j:j \neq i\}} ( I(Y_j \leq y) )^s K^{(m)}\bigg(\frac{\bar{\beta}^{\top}_j Z_j - \bar{\beta}^{\top}_i Z_i}{h}\bigg) Z_i^{m}  \nonumber\\
&= \frac{1}{nh^{m+1}} \sum_{\{j:j \neq i\}} ( I(Y_j \leq y) )^s  K^{(m)}_2\bigg(\frac{\beta^{\text{o}}_j Z_j - \beta^{\text{o}}_i Z_i}{h}\bigg)Z_i^{m}  \nonumber 
\end{align}
\begin{align}
& \hspace{0.2in}+ \frac{1}{nh^{m+1}} \sum_{\{j:j \neq i\}} ( I(Y_j \leq y) )^s \frac{1}{h} K^{(m+1)}_2 \bigg(\frac{\widetilde{\beta}^{\top}_j Z_j - \widetilde{\beta}^{\top}_i Z_i}{h}\bigg) \big[ Z^{\top}_j \big(\bar{\beta}_j - \beta^{\text{o}}_j\big) + Z^{\top}_i \big(\bar{\beta}_i - \beta^{\text{o}}_i\big) \big]Z_i^{m} \nonumber\\
&=  \frac{1}{n h^{m+1}} \sum_{\{j:j \neq i\}} ( I(Y_j \leq y) )^s K^{(m)}_2\bigg(\frac{\gamma^{\text{o} \top} W_j - \gamma^{\text{o} \top} W_i}{h}\bigg)Z_i^{m} + \frac{b_{s,m}(y)\phi_n}{h^{m+1}} \nonumber \\
&=  \partial^m_v \big( M^*_{s,m} \big(y, Z_i, \gamma^{\text{o} \top} W_i \big) f_{\gamma^{\text{o}}_*}\big( \gamma^{\text{o} \top} W_i \big) \big) + \xi_{s,m}\big(y, W_i, \gamma^{\text{o}}\big) + \frac{b_{s,m}(y)\phi_n}{h^{m+1}}, \label{eq:A.2.8}
\end{align}
where $\widetilde{\beta}_i$ lies on the line segment between $\bar{\beta}_i$ and $\beta^{\text{o}}_i$, $i=1,\dots, n$, and $\sup_y |b_{s,m}(y)| = O_p(1)$, $s, m = 0, 1$. 
Using (\ref{eq:A.2.8}) and Assumption {A1}, arguments analogous to those in the proofs of \eqref{eq:lemma1_5}--\eqref{eq:lemma1_7} yield the uniform convergence
\begin{align}
& \sup_{\{i,y\}} \Big| \widehat{G}^{-i}_h \big(y, \bar{\beta}^{\top}_iZ_i; \bar{\beta} \big) - G\big(y, \gamma^{\text{o} \top} W_i\big) \Big| = o_p(1) \label{eq:A.2.9.0}
\end{align}
and
\begin{align}
&\sup_{\{i,y\}} \Big\|\partial_{\beta} \widehat{G}^{-i}_h \big(y, \bar{\beta}^{\top}_iZ_i; \bar{\beta} \big) - G^{[1]}\big(y, Z_i, \gamma^{\text{o} \top} W_i\big) \Big\| = o_p(1), \label{eq:A.2.9}
\end{align} 
where $G^{[1]}(y, z, v)$ is defined as $G^{[1]}(y, w, v)$ with $M_{s,1}(y,w,v)$ replaced by $M^*_{s,1}(y,z,v)$, $s = 0, 1$.
From (\ref{eq:A.2.9.0}) and (\ref{eq:A.2.9}), it follows that for each $i = 1 ,\dots, n$,
\begin{align*}
\big\| \textsc{ps}_i\big(\bar{\beta}\big) \big\|  \leq \bigg[\int \Big|\big(I(Y_i \leq y) - G\big(y,\beta^{\text{o}\top}_i Z_i\big)\big)\Big| \Big\|  G^{[1]}\big(y,Z_i,\beta^{\text{o}\top}_i Z_i\big)\Big\| d \widehat{F}(y) \bigg]+ o_p(1) + O_p\bigg( \frac{ \phi_n}{h^2} \bigg).
\end{align*}
Consequently,
\begin{align*}
\bigg|\sum_{\ell = 1}^k \sum_{\{ i,j \in \mathcal{G}_{\ell},i<j\}} \frac{\big(\textsc{ps}_i\big(\bar{\beta}\big) - \textsc{ps}_j\big(\bar{\beta})\big)^{\top}\big(\beta_j - \beta_i\big)}{|\mathcal{G}_{\ell}|} \bigg| \leq \frac{c_{0}}{n} \bigg( \sum_{\ell = 1}^k \sum_{\{ i,j \in \mathcal{G}_{\ell},i<j\}} \big\| \beta_{(1)i} - \beta_{(1)j} \big\|\bigg) (1 +o_p(1))
\end{align*} 
and, hence,
\begin{align}
I_1 \geq   -\frac{c_{0}}{n}\bigg(\sum_{\ell = 1}^k \sum_{\{ i,j \in \mathcal{G}_{\ell},i<j\}} \big\| \beta_{(1)i} - \beta_{(1)j} \big\|\bigg) (1+o_p(1))\label{eq:A.2.11}
\end{align}
for some constant $c_{0}>0$. 

Given $\beta\in B_{n}$, $\gamma \in \Gamma_{n}$, $i\in\mathcal{G}_{\ell_{1}}$, and $\ell_{1}\neq\ell_{2}$, the triangle inequality implies
\begin{align*}
&\|\beta_{(1)i}-\gamma_{(1)\ell_{1}}\|_{1}\leq  \|\beta_{(1)i}-\beta^{\text{o}}_{(1)i}\|_{1}+\| \gamma_{(1)\ell_{1}} - \gamma^{\text{o}}_{(1)\ell_{1}}\|_{1}\leq 2\phi_{n} \hspace{2.35in}\\
\text{and}&\\
&\|\beta_{(1)i}-\gamma_{(1)\ell_{2}}\|_{1} = \|\beta_{(1)i}-\beta^{\text{o}}_{(1)i}+ \gamma^{\text{o}}_{(1)\ell_{1}}- \gamma^{\text{o}}_{(1)\ell_{2}}+ \gamma^{\text{o}}_{(1)\ell_{2}}- \gamma_{(1)\ell_{2}}\|_{1}\nonumber\\
&\geq \| \gamma^{\text{o}}_{(1)\ell_{1}}- \gamma^{\text{o}}_{(1)\ell_{2}}\|_{1}- \|\beta_{(1)i}-\beta^{\text{o}}_{(1)i}\|_{1}-\| \gamma^{\text{o}}_{(1)\ell_{2}}- \gamma_{(1)\ell_{2}}\|_{1}\geq D_{\gamma} -2\phi_{n}>2\phi_{n}.
\end{align*}
It follows that
\begin{align*}
\min_{\{\ell\in\mathcal{C}\}}\|\beta_{(1)i}-\gamma_{(1)\ell}\|_{1} =\|\beta_{(1)i}-\gamma_{(1)\ell_{1}}\|_{1}.
\end{align*}
Using this, a lower bound for $I_2$ is derived as
\begin{align}
I_{2}&=\lambda\sum^{k}_{\ell=1}\sum_{\{i\in\mathcal{G}_{\ell}\}} \|\beta_{(1)i}-\gamma_{(1)\ell}\|_{1}
=\lambda\sum^{k}_{\ell=1}\frac{1}{2 |\mathcal{G}_{\ell}|}\sum_{\{i,j\in\mathcal{G}_{\ell}\}} \big(\|\beta_{(1)i}-\gamma_{(1)\ell}\|_{1}+\|\beta_{(1)j}-\gamma_{(1)\ell}\|_{1}\big)\nonumber \\
&\geq  \lambda\sum^{k}_{\ell=1}\frac{1}{2 |\mathcal{G}_{\ell}|}\sum_{\{i,j\in\mathcal{G}_{\ell}\}} \|\beta_{(1)i}-\beta_{(1)j}\|_{1} \geq  \frac{\lambda}{n} \sum^{k}_{\ell=1}\sum_{\{i,j\in\mathcal{G}_{\ell},i<j\}} \|\beta_{(1)i}-\beta_{(1)j}\|_{1}\nonumber\\
&\geq \frac{\lambda}{n}\sum^{k}_{\ell=1}\sum_{\{i,j\in\mathcal{G}_{\ell},i<j\}} \|\beta_{(1)i}-\beta_{(1)j}\|.  \label{eq:A.2.13}
\end{align}

Substituting (\ref{eq:A.2.11}) and (\ref{eq:A.2.13}) into (\ref{eq:A.2.4}) yields
\begin{align}
\textsc{psis}_{\textsc{sp}}\big(\beta,\gamma;\lambda\big) - \textsc{psis}_{\textsc{sp}}\big(\beta^{*},\gamma^{*};\lambda\big)
\geq  \Big(\frac{\lambda}{n} -\frac{c_{0}}{n}\Big) \bigg(\sum_{\ell = 1}^k \sum_{\{ i,j \in \mathcal{G}_{\ell},i<j\}} \| \beta_{(1)i} - \beta_{(1)j} \|\bigg) (1+o_p(1)).
 \label{eq:A.2.14}
\end{align}
Therefore, with probability converging to one, 
\begin{align*}
\textsc{psis}_{\textsc{sp}}\big(\beta,\gamma;\lambda\big) \geq\textsc{psis}_{\textsc{sp}}\big(\beta^{*},\gamma^{*};\lambda\big).
\end{align*}
The oracle property stated in Theorem \ref{Thm3.1} follows directly from (\ref{eq:A.2.3}) and (\ref{eq:A.2.14}).

\subsection{Proof of Theorem \ref{Thm3.3_1}}
\label{sec:A.3.1}
 
Given that $P(\widehat{\mathcal{G}}=\mathcal{G}^{\text{o}})\longrightarrow 1$ as $n\longrightarrow\infty$, it follows that
\begin{align}
\sup_{\{i, \ell\}} \big| I\big(i \in \widehat{\mathcal{G}}_{\ell}\big)-I\big(i \in\mathcal{G}_{\ell}^{\text{o}}\big)\big|=o_{p}\big( n^{-1} \big).\label{eq:A.3.1}
\end{align}
Define $N_{s,\ell,h_d}(u;{\mathcal{G}}_{\ell},A_{d})  = \sum_{i=1}^n \big(I\big(i \in {\mathcal{G}}_{\ell}\big)\big)^s \mathcal{K}_{r,h_d}\big(A_{d}^{\top}X_i - u\big) / n$, $s = 0, 1,$ $\ell=1,\dots,k-1.$
By (\ref{eq:A.3.1}), the identity $I(C_{i}=\ell)=I(i\in\mathcal{G}^{\text{o}}_{\ell})$ for $i=1,\dots, n$ and $\ell=1,\dots, k-1$, the uniform boundedness of $\mathcal{K}^{(m)}_{r}(u)$ for $m=0,1,2$, and assumption {B1}, we obtain
\begin{align*}
&\sup_{\{\ell, x, A_{d*}\}} \bigg\| \partial_{\vvec{(A_{d*})}}^m N_{1,\ell,h_d}(u;\widehat{\mathcal{G}}_{\ell},A_{d}) - \partial_{\vvec{(A_{d*})}}^m N_{1,\ell,h_d}(u;A_{d}) \bigg\| \\
&=\sup_{\{\ell, x, A_{d*}\}} \bigg\| \frac{1}{nh^{m+d}_{d}}\sum^{n}_{i=1} \big( I\big(i \in \widehat{\mathcal{G}}_{\ell}\big)-I\big(i \in\mathcal{G}^{\text{o}}_{\ell}\big)\big) \\
&\hspace{0.2in}\prod_{\{\sum_{j=1}^d m_j = m\}} K^{(m_j)}_r \bigg( \frac{\alpha_{j}^{\top}(X_{i}-x)}{h_d} \bigg) \big(\vvec({X}_{d*}-{x}_{d*})\big)^{\otimes m}\bigg\| \\
&\leq \frac{1}{h^{m}_{d}}\sup_{\{i,\ell\}} \big| I\big(i \in \widehat{\mathcal{G}}_{\ell}\big)-I\big(i \in\mathcal{G}^{\text{o}}_{\ell}\big)\big| \\
&\hspace{0.2in}\sup_{\{x, A_{d*}\}} \bigg\| \frac{1}{nh^{d}_{d}}\sum^{n}_{i=1} \prod_{\{\sum_{j=1}^d m_j = m\}}K^{(m_j)}_r \bigg( \frac{\alpha_{j}^{\top}(X_{i}-x)}{h_d} \bigg) \big(\vvec({X}_{d*}-{x}_{d*})\big)^{\otimes m}\bigg\| \\
&=o_{p}\bigg(\frac{1}{nh^{m}_{d}} \bigg), m = 0, 1, 2.
\end{align*}
Since $N_{0,\ell,h_d}(u;\widehat{\mathcal{G}}_{\ell},A_{d})=N_{0,\ell,h_d}(u;A_{d})$, it follows that
\begin{eqnarray}
\sup_{\{\ell,x, A_{d*}\}} \bigg\| \partial_{\vvec{(A_{d*})}}^m\widehat{\pi}_{\ell,h_d}( A_{d}^{\top}x; \widehat{\mathcal{G}},A_{d}) -  \partial_{\vvec{(A_{d*})}}^m\widehat{\pi}_{\ell,h_d}( A_{d}^{\top}x; \mathcal{G}^{\text{o}},A_{d}) \bigg\| = o_p \bigg({\frac{1}{nh_{d}^{m}}}\bigg), m=0,1,2.~ \label{eq:A.3.2}
\end{eqnarray}
Combining (\ref{eq:A.3.2}) with the second assertion in Lemma \ref{A:Lemma_1} and under assumptions {B2}--{B4}, we conclude that
\begin{align}
& \sup_{\{\ell, x, A_{d*}\}} \bigg\| \partial_{\vvec{(A_{d*})}}^m\widehat{\pi}_{\ell,h_d}( A_{d}^{\top}x; \widehat{\mathcal{G}},A_{d}) - \pi_{\ell}^{[m]}( x,A_{d}^{\top}x) - \frac{1}{n}\sum_{i=1}^n \xi_{i,\ell}^{[m]}(x, A_{d}) \bigg\|\nonumber\\
& = o_p \bigg({\frac{\log n}{nh_{d}^{2m+d}}}\bigg) + O \big(h_{d}^{2r} \big), m = 0, 1, 2.\label{eq:A.3.3}
\end{align}

Let $\partial_{\vvec{(A_{d*})}}^m\widehat{\pi}_{\ell,h_d}( A_{d}^{\top}x; A_{d})=\partial_{\vvec{(A_{d*})}}^m\widehat{\pi}_{\ell,h_d}( A_{d}^{\top}x; \mathcal{G}^{\text{o}},A_{d})$ for $m=0,1,2$.
In view of the property in (\ref{eq:A.3.3}), the proof of Theorem \ref{Thm3.3_1} proceeds identically to the argument below, 
with $\widehat{\mathcal{G}}$ replaced by $\mathcal{G}^{\text{o}}$.

\subsubsection*{(Consistency of $\vvec\big(\widehat{A}_{d*}\big)$ to $\vvec(A^{\text{o}}_{d*})$)}
A direct decomposition of $\textsc{pss}(A_d,h_d)$ yields
\begin{align}
\frac{2}{n} \textsc{pss}(A_d,h_d) = & \frac{1}{n} \sum_{i=1}^n \sum_{\ell = 1}^{k-1} \varepsilon^2_{\ell}\big(C_i, A_{d}^{\top}X_i\big) 
 + \frac{1}{n} \sum_{i=1}^n \sum_{\ell = 1}^{k-1} \Big(\pi_{\ell}^{[0]}\big(X_i,A_{d}^{\top}X_i\big) - \widehat{\pi}_{\ell,h_d}^{-i}\big(A_{d}^{\top}X_i; A_{d}\big) \Big)^2 \nonumber \\
& + \frac{2}{n} \sum_{i=1}^n \sum_{\ell = 1}^{k-1} \varepsilon_{\ell}\big(C_i, A_{d}^{\top}X_i\big) \Big(\pi_{\ell}^{[0]}\big(X_i,A_{d}^{\top}X_i\big) - \widehat{\pi}_{\ell,h_d}^{-i}\big(A_{d}^{\top}X_i; A_{d}\big) \Big) \nonumber \\
\stackrel{\triangle}{=} & I_{1 n}(A_{d}) + I_{2 n}(A_{d}) + I_{3 n}(A_{d}),  \label{eq:thm2_1}
\end{align}
where $\varepsilon_{\ell}\big(C, A_{d}^{\top}X\big)  = I(C = \ell) - \pi_{\ell}^{[0]}\big(X_, A_{d}^{\top}X\big)$.
By Corollary 7 in \cite{sherman1994maximal}, the zero-mean $U$-process $I_{1 n}(A_{d}) - \sum_{\ell = 1}^{k-1} \E\big[ \varepsilon^2_{\ell}\big(C, A_{d}^{\top}X\big) \big]$ satisfies
\begin{align}
\sup_{A_{d*}} \bigg| I_{1 n}(A_d) - \sum_{\ell = 1}^{k-1} \E\big[ \varepsilon^2_{\ell}\big(C, A_{d}^{\top}X\big) \big]\bigg| = O_p\Big(\frac{1}{\sqrt{n}} \Big). \label{eq:thm2_2}
\end{align}
Under assumption {B1}, the uniform boundedness of $\varepsilon_{\ell}\big(C, A_{d}^{\top}X\big)$, and the second assertion in Lemma \ref{A:Lemma_1}, it follows that 
\begin{align}
\sup_{A_{d*}} |I_{2 n}(A_{d})| = o_p(1) \text{ and } \sup_{A_{d*}} |I_{3 n}(A_{d})| = o_p(1). \label{eq:thm2_3}
\end{align}
Substituting (\ref{eq:thm2_2}) and (\ref{eq:thm2_3}) into (\ref{eq:thm2_1}) and applying the triangle inequality yield
\begin{align}
\sup_{A_{d*}} \bigg| \frac{2}{n} \textsc{pss}(A_{d},h_{d})  - \sum_{\ell = 1}^{k-1} \E\Big[\varepsilon^2_{\ell}\big(C, A_{d}^{\top}X\big)\Big] \bigg| = o_p (1). \label{eq:thm2_4}
\end{align}
Combining (\ref{eq:thm2_4}) with the strict inequality
\begin{align*}
\sum_{\ell=1}^{k-1}\E\big[\varepsilon^2_{\ell}\big(C, A_{d}^{\top}X\big)\big] > \sum_{\ell=1}^{k-1}\E\big[\varepsilon^2_{\ell}\big(C, A^{\text{o} \top}_{d}X\big)\big] \text{ for } A_{d} \neq A^{\text{o}}_{d},
\end{align*}
implies that
\begin{align*}
 \frac{2}{n} \textsc{pss}(\widehat{A}_{d},h_{d})  \geq \frac{2}{n} \textsc{pss}({A^{\text{o}}_{d}},h_{d})
\end{align*}
with probability converging to one. Since $\widehat{A}_d$ minimizes $\textsc{pss}(A_d, h_d)$ by definition, the reverse inequality holds deterministically. By the continuity of $\textsc{pss}(A_d, h_d)$ in $A_d$, it follows that
\begin{align}
\vvec(\widehat{A}_{d*})\stackrel{p}{\longrightarrow}\vvec(A^{\text{o}}_{d*}) \text{ as } n\longrightarrow\infty. \label{eq:thm2_4.1}
\end{align}

\subsubsection*{(Asymptotic normality of $\vvec(\widehat{A}_{*})$ and $\sqrt{n}$-consistency of $\vvec(\widehat{A}_{d*})$ for $d\geq d^{\text{o}}$)}
Let $\textsc{ps}(A_{d})=\partial_{\vvec(A_{d*})} 2 \textsc{pss}(A_d, h_d)/n$ and $\textsc{pi}(A_{d})=\partial_{\vvec(A_{d*})}^2 2 \textsc{pss}(A_d, h_d)/n$.
For an arbitrary constant vector $\nu\in\mathbb{R}^{(p-d)d}$,  a Taylor expansion of $\nu^{\top}\textsc{ps}_{d}(\widehat{A}_{d})$ around $\vvec(A^{\text{o}}_{d*})$ yields
\begin{align}
\nu^{\top} (\sqrt{n} \textsc{ps}(A^{\text{o}}_{d})) + \nu^{\top} (\textsc{pi}(\bar{A}_{d})) \sqrt{n}(\vvec(\widehat{A}_{d*}) - \vvec(A^{\text{o}}_{d*})) = 0, \label{eq:thm2_5}
\end{align}
where $\vvec(\bar{A}_{d})$ lies on the line segment between $\vvec(\widehat{A}_d)$ and $\vvec(A^{\text{o}}_{d})$.
The pseudo score vector in (\ref{eq:thm2_5}) is decomposed as 
\begin{align}
\textsc{ps} (A^{\text{o}}_{d}) =& \sum_{\ell_1=0}^1 \sum_{\ell_2=0}^1 \frac{2}{n} \sum_{i=1}^n \sum_{\ell = 1}^{k-1} (-1)^{\ell_1}\big( \varepsilon_{\ell}\big(C, A_{d}^{\text{o}\top}X\big) \big)^{1-\ell_1} \Big(\widehat{\pi}^{-i}_{\ell, h_{d}} \big(A^{\text{o} \top}_{d} X_i; A^{\text{o}}_{d}\big)  -  \pi_{\ell}\big(A^{\text{o} \top}_{d} X_i \big)\Big)^{\ell_1} \nonumber \\
& \Big( \pi_{\ell}^{[1]} \big(X_i, A^{\text{o} \top}_{d} X_i \big) \Big)^{1-\ell_2} \Big( \partial_{\vvec(A^{\text{o}}_{d*})} \widehat{\pi}^{-i}_{\ell,h_d} \big(A^{\text{o} \top}_{d} X_i; A^{\text{o}}_{d}\big) - \pi^{[1]}_{\ell}\big(X_i, A^{\text{o} \top}_{d} X_i \big) \Big)^{\ell_2}  \nonumber \\
\stackrel{\triangle}{=}& \sum_{\ell_1=0}^1 \sum_{\ell_2=0}^1 \textsc{ps}_{\ell_1 \ell_2}(A^{\text{o}}_{d}). \label{eq:thm2_6}
\end{align}
By assumption {B1}, the uniform boundedness of $\varepsilon^2_{\ell}\big(C, A_{d}^{\text{o}\top}X\big)$ and $\pi^{[1]}_{\ell}\big(x, A^{\text{o} \top}_d x \big)$, and the second assertion in Lemma \ref{A:Lemma_1}, we obtain
\begin{align}
\sqrt{n} \|\textsc{ps}_{11}(A^{\text{o}}_d)\| & = \sqrt{n}\bigg\| \frac{2}{n} \sum_{i=1}^n \sum_{\ell = 1}^{k-1} \prod_{m=0}^1\Big( \partial_{\vvec(A^{\text{o}}_{d*})}^m \widehat{\pi}^{-i}_{\ell,h_d} \big(A^{\text{o} \top}_{d} X_i; A^{\text{o}}_{d}\big) - \pi^{[m]}_{\ell}\big(X_i, A^{\text{o} \top}_{d} X_i \big) \Big) \bigg\|\nonumber \\
& \leq 2\sqrt{n} (k-1) \prod_{m=0}^1 \sup_{x} \Big\| \partial_{\vvec(A^{\text{o}}_{d*})}^m \widehat{\pi}^{-i}_{\ell,h_d} \big(A^{\text{o} \top}_{d} X_i; A^{\text{o}}_{d}\big) - \pi^{[m]}_{\ell}\big(X_i, A^{\text{o} \top}_{d} X_i \big)  \Big\|  \nonumber \\
&= 2\sqrt{n}  \Bigg(o_p\Bigg(\sqrt{\frac{\log n}{nh_d^{d}}}\Bigg) + O\big(h_d^r\big)  \Bigg) \Bigg( o_p\Bigg(\sqrt{\frac{\log n}{nh_d^{2+d}}}\Bigg) + O\big(h_d^r\big) \Bigg)= o_p(1), \label{eq:thm2_7} \\
\sqrt{n} \textsc{ps}_{01}(A^{\text{o}}_{d}) = & \sqrt{n} \sum_{\ell=1}^{k-1} \frac{2}{n} \sum_{i = 1}^n \varepsilon_{\ell}\big(C,_i A_{d}^{\text{o} \top}X_i\big) \Bigg( \frac{1}{n-1} \sum_{\{j: j \neq i\}} \xi^{[1]}_{j,\ell}(X_i, A^{\text{o}}_{d}) + r^{[1]}_{\ell}(X_i, A^{\text{o}}_{d}) \Bigg)  \nonumber \\
= & \sqrt{n} \sum_{\ell=1}^{k-1} \frac{2}{n(n-1)} \sum_{i \neq j} \varepsilon_{\ell}\big(C_i, A_{d}^{\text{o} \top}X_i\big) \bar{\xi}^{[1]}_{j,\ell}(X_i, A^{\text{o}}_{d})  \nonumber\\
&+ \sqrt{n} \sum_{\ell=1}^{k-1} \frac{2}{n(n-1)} \sum_{i \neq j} \varepsilon_{\ell}\big(C_i, A_{d}^{\text{o} \top}X_i\big) \E\Big[{\xi}^{[1]}_{j,\ell}(X_i, A^{\text{o}}_{d})|X_i\Big] + o_p(1),\label{eq:thm2_7.5} 
\end{align}
and
\begin{align}
\sqrt{n} \textsc{ps}_{10}(A^{\text{o}}_{d}) = & \sqrt{n} \sum_{\ell=1}^{k-1} -\frac{2}{n} \sum_{i = 1}^n \pi^{[1]}_{\ell}\big(X_i, A^{\text{o} \top}_{d} X_i\big) \Bigg( \frac{1}{n-1} \sum_{\{j: j \neq i\}} \xi^{[0]}_{j,\ell}(X_i, A^{\text{o}}_{d}) + r^{[0]}_{\ell}(X_i, A^{\text{o}}_{d}) \Bigg)  \nonumber\\
= & \sqrt{n} \sum_{\ell=1}^{k-1} -\frac{2}{n(n-1)} \sum_{i \neq j} \pi^{[1]}_{\ell}\big(X_i, A^{\text{o} \top}_{d} X_i\big) \bar{\xi}^{[0]}_{j,\ell}(X_i, A^{\text{o}}_{d})  \nonumber \\
&- \sqrt{n} \sum_{\ell=1}^{k-1} \frac{2}{n(n-1)} \sum_{i \neq j} \pi^{[1]}_{\ell}\big(X_i, A^{\text{o} \top}_{d} X_i\big) \E\Big[{\xi}^{[0]}_{j,\ell}(X_i, A^{\text{o}}_{d})|X_i\Big] + o_p(1), \label{eq:thm2_8}
\end{align} 
where $\bar{\xi}^{[m]}_{j,\ell}(X_i, A^{\text{o}}_{d}) = {\xi}^{[m]}_{j,\ell}(X_i, A^{\text{o}}_{d}) - \E[{\xi}^{[m]}_{j,\ell}(X_i, A^{\text{o}}_{d}) | X_i]$, $m = 0, 1$, $i, j = 1, \dots, n$, $\ell = 1 ,\dots, k-1$. 
For $d \geq d^{\text{o}}$, straightforward calculations yield
\begin{align*}
&\xi^{[0]}_{j,\ell}(x,A^{\text{o}}_{d}) = \frac{\big( I(C_j = \ell) - \pi_{\ell}\big( A^{\text{o}\top}_{d} x\big) \big) \mathcal{K}_{r,h_d}(A^{\text{o}\top}_{d} (X_{j} - x)) }{f_{A^{\text{o}}_{d*}(A^{\text{o}\top}_{d}x)}}, j = 1, \dots, n, \ell = 1, \dots, k-1,\\
&\pi^{[1]}_{\ell}\big(x, A^{\text{o} \top}_{d} x \big) = \big( \partial_u \pi_{\ell} \big(A^{\text{o} \top}_{d} x \big)\big) \big( H(A^{\text{o} \top}_{d} x) - x_{d*} \big),\text{ where } H(u) = \E[X_{d*}|A^{\text{o} \top}_{d}X = u], \\
&\E\Big[ \varepsilon_{\ell}\big(C_i, A^{\text{o} \top}_{d}X_i\big) \bar{\xi}^{[1]}_{j,\ell}(X_i, A^{\text{o}}_{d}) | (X_i, C_i) \Big] = 0, \E\Big[ \pi^{[1]}_{\ell}\big(X_i, A^{\text{o} \top}_{d} X_i \big) \bar{\xi}^{[0]}_{j,\ell}(X_i, A^{\text{o}}_{d}) | (X_i, C_i) \Big] = 0,\\
&\E\Big[ \varepsilon_{\ell}\big(C_i, A^{\text{o} \top}_{d}X_i\big) \bar{\xi}^{[1]}_{j,\ell}(X_i, A^{\text{o}}_{d}) | (X_j, C_j) \Big] = \E\Big[ \E\big[ \varepsilon_{\ell}\big(C_i, A^{\text{o} \top}_{d}X_i\big) | X_i\big] \bar{\xi}^{[1]}_{j,\ell}(X_i, A^{\text{o}}_{d}) | (X_j, C_j) \Big] = 0, \\
\text{and}&\\
&\E \Big[ \pi^{[1]}_{\ell}\big(X_i, A^{\text{o} \top}_{d} X_i \big) \bar{\xi}^{[0]}_{j,\ell}(X_i, A^{\text{o}}_{d}) | (X_j, C_j) \Big] \\
&=\E \Big[ \big( \partial_u \pi_{\ell} \big(A^{\text{o} \top}_{d} X_i \big)\big) \big( H(A^{\text{o} \top}_{d} X_i) - \E[X_{d*i} | A^{\text{o} \top}_{d} X_i] \big) \bar{\xi}^{[0]}_{j,\ell}(X_i, A^{\text{o}}_{d}) | (X_j, C_j) \Big] = 0.
\end{align*}
These imply that the $U$-statistics of order 2 in (\ref{eq:thm2_7.5}) and (\ref{eq:thm2_8}) are degenerate.
Using the bounds $\E[(\varepsilon_{\ell}(C_i, A^{\text{o} \top}_{d}X_i) \bar{\xi}^{[1]}_{j,\ell}(X_i, A^{\text{o}}_{d}) )^2] = O(h^{-(2+d)})$ and $\E[(\pi^{[1]}_{\ell}(X_i, A^{\text{o} \top}_{d} X_i) \bar{\xi}^{[0]}_{j,\ell}(X_i, A^{\text{o}}_{d}) )^2] = O(h^{-d})$, Corollary 4 in \cite{sherman1994maximal} yields, for $d \geq d^{\text{o}}$,
\begin{align}
&\Bigg\| \frac{2}{n(n-1)}\sum_{i \neq j} \varepsilon_{\ell}\big(C_i, A^{\text{o} \top}_{d}X_i\big) \bar{\xi}^{[1]}_{j,\ell}(X_i, A^{\text{o}}_{d}) \Bigg\| = O_p\Bigg(\frac{1}{n\sqrt{h_{d}^{2+d}}}\Bigg)\nonumber
\end{align}
and
\begin{align}
&\Bigg\| \frac{2}{n(n-1)}\sum_{i \neq j} \pi_{\ell}^{[1]}\big(X_i, A^{\text{o} \top}_{d} X_i \big) \bar{\xi}^{[0]}_{j,\ell}(X_i, A^{\text{o}}_{d}) \Bigg\| = O_p\Bigg(\frac{1}{n\sqrt{h_{d}^{d}}}\Bigg), \ell = 1, \dots, k-1.  \label{eq:thm2_9}
\end{align}
Applying Theorem 2.5 in \cite{Kosorok2008} for $d \geq d^{\text{o}}$ further establishes
\begin{align}
&  \bigg\| \frac{1}{n} \sum_{i=1}^n  \varepsilon_{\ell}\big(C_i, A^{\text{o} \top}_{d}X_i\big) \E\Big[{\xi}^{[1]}_{j,\ell}(X_i, A^{\text{o}}_{d}) | X_i\Big] \bigg\| = O\bigg( \frac{h_d^r}{\sqrt{n}} \bigg)\nonumber
 \end{align}
and
\begin{align}
&  \bigg\| \frac{1}{n} \sum_{i=1}^n  \pi^{[1]}_{\ell}\big(X_i, A^{\text{o} \top}_{d} X_i \big) \E\Big[{\xi}^{[0]}_{j,\ell}(X_i, A^{\text{o}}_{d}) | X_i\Big] \bigg\| = O\bigg( \frac{h_d^r}{\sqrt{n}} \bigg), \ell = 1, \dots, k-1. \label{eq:thm2_9.5}
\end{align}
Under assumption {B1}, substituting (\ref{eq:thm2_9}) and (\ref{eq:thm2_9.5}) into (\ref{eq:thm2_7.5}) and (\ref{eq:thm2_8}) yields
\begin{align}
\sqrt{n} \textsc{ps}_{01}(A^{\text{o}}_{d}) = o_p(1) \text{ and } \sqrt{n} \textsc{ps}_{10}(A^{\text{o}}_{d}) = o_p(1) \text{ for } d \geq d^{\text{o}}. \label{eq:thm2_10}
\end{align}
By substituting (\ref{eq:thm2_7}), (\ref{eq:thm2_10}), and $\sqrt{n} \nu^{\top} \textsc{ps}_{00}(A^{\text{o}}_{d}) \stackrel{d}{\longrightarrow} N_{(p-d)d}\big(0, \nu^{\top} \E[S^{\otimes 2}_{A^{\text{o}}_{d}}] \nu\big)$, which follows from the central limit theorem, into (\ref{eq:thm2_6}), we obtain the asymptotic distribution:
\begin{align}
\sqrt{n} \textsc{ps}(A^{\text{o}}_{d})\stackrel{d}{\longrightarrow} N_{(p-d)d}\Big(0, \E\big[S^{\otimes 2}_{A^{\text{o}}_{d}}\big]\Big). \label{eq:thm2_12}
\end{align}
By the second assertion in Lemma \ref{A:Lemma_1} and the uniform boundedness of $\pi^{[m]}_{\ell}(A^{\top}_{d}x)$, $m = 0, 1, 2$, $\ell = 1 ,\dots, k-1$, the pseudo information matrix in (\ref{eq:thm2_5}) admits the representation
\begin{align}
\textsc{pi}(A_d) = & \frac{2}{n} \sum_{\ell=1}^{k-1}  \sum_{i=1}^n \Big( \pi^{[1]}_{\ell}\big(X_i, A_{d}^{\top}X_i\big) \Big)^{\otimes 2}  + o_p\Bigg( \sqrt{\frac{\log n}{nh_d^{2+d}}} \Bigg) + O(h_d^r)  \nonumber\\
& - \frac{2}{n} \sum_{\ell=1}^{k-1}  \sum_{i=1}^n \varepsilon_{\ell} \big(C_i, A_{d}^{\top}X_i\big) \pi^{[2]}_{\ell} \big(X_i, A_{d}^{\top}X_i\big) + o_p\Bigg( \sqrt{\frac{\log n}{nh_d^{4+d}}} \Bigg) + O(h_d^r). \label{eq:thm2_13}
\end{align}
Applying the law of large numbers to (\ref{eq:thm2_13}) under assumption {B1} gives 
\begin{align}
\textsc{pi}_d(A_d)  \stackrel{p}{\longrightarrow} E\bigg[\sum_{\ell=1}^{k-1} \bigg(\big(\pi^{[1]}_{\ell}\big(X, A_{d}^{\top}X\big)\big)^{\otimes 2} - \varepsilon_{\ell}\big(C, A_{d}^{\top}X\big) \pi^{[2]}_{\ell} \big(X, A_{d}^{\top}X\big) \bigg) \bigg].  \label{eq:thm2_14}
\end{align}
The continuity of $\textsc{pi}(A_{d})$ in $A_{d}$, the consistency of $\vvec(\widehat{A}_{d*})$ to $\vvec(A^{\text{o}}_{d*})$, and the fact that the second term in (\ref{eq:thm2_14}) has zero mean, for $A_{d} = A^{\text{o}}_{d}$ and $d \geq d^{\text{o}}$, together imply
\begin{align}
\textsc{pi}(\bar{A}_{d}) \stackrel{p}{\longrightarrow} V_{A^{\text{o}}_{d}}. \label{eq:thm2_15}
\end{align}
Substituting (\ref{eq:thm2_12}) and (\ref{eq:thm2_15}) into (\ref{eq:thm2_5}) and invoking Slutsky's theorem yields 
\begin{align}
\vvec\big(\widehat{A}_{d*}\big) =\vvec(A^{\text{o}}_{d})+O_{p}\big(n^{-\frac{1}{2}}\big) \text{ for }d \geq d^{\text{o}}. \label{eq:thm2_15_d0}
\end{align}
Under assumption {B5}, we establish the limiting distribution
\begin{align}
\sqrt{n} (\vvec\big(\widehat{A}^{\text{o}}_{*}\big) - \vvec(A^{\text{o}})) \stackrel{d}{\longrightarrow} N_{(p-d^{\text{o}})d^{\text{o}}}\Big(0, V_{A^{\text{o}}}^{-1}\E[S^{\otimes 2}_{A^{\text{o}}}]V_{A^{\text{o}}}^{-1}\Big) \text{ as } n \longrightarrow \infty. \label{eq:thm2_15_d}
\end{align}

\subsubsection*{(Consistency of $\hat{d}$ to $d^{\text{o}}$)}
By $\sup_{\{x, A_{d}\}} |\widehat{\pi}_{\ell, \tilde{h}_{d}}\big( A^{\top}_{d}x; A_{d}\big)-\pi^{[0]}_{\ell}(x,A^{\top}_{d}x) | = O_{p}( n^{-{r}/{(2r+d)}})$ and
\begin{align*}
\vvec(\widehat{A}_{d*}) -  \vvec(A^{\text{o}}_{d*}) =\left\{
 \begin{array}{ll}
 o_{p}(1) & d < d^{\text{o}},\\\\
 O_{p}\big( n^{-\frac{1}{2}}\big) & d \geq d^{\text{o}},
  \end{array}\right.
\end{align*}
we obtain
\begin{align}
&\sup_{x}\Big|\widehat{\pi}_{\ell,\tilde{h}_{d}}\big(\widehat{A}^{\top}_{d}x; \widehat{A}_{d}\big)-\pi_{\ell}(A^{\text{o} \top}_{d}x) \Big |\nonumber \\
&\leq \sup_{\{x, A_{d*}\}}\Big |\widehat{\pi}_{\ell, \tilde{h}_{d}}\big( A^{\top}_{d}x; A_{d}\big)-\pi^{[0]}_{\ell}(x,A^{\top}_{d}x) \Big | + \sup_{\{u,x\}}\Big |\partial_{u} \pi_{\ell}^{[0]}(u) \Big | \Big\| \vvec(\widehat{A}_{d*}) -\vvec(A^{\text{o}}) \Big\| \| x \| \nonumber \\
&= \left\{
\begin{array}{ll}
O_{p}\big( n^{-\frac{r}{2r+d}} \big) + O_{p}\big( \| \vvec(\widehat{A}_{d*}) -\vvec(A^{\text{o}}_{d*}) \| \big)
& d < d^{\text{o}}, \\\\
O_{p}\big( n^{-\frac{1}{2}} \big)
& d \geq d^{\text{o}}.
\end{array}\right. \label{eq:thm2_16}
\end{align}
Combining this with the decomposition of $\textsc{pss}\big(\widehat{A}_{d},\tilde{h}_{d}\big)$ yields
\begin{align}
&\frac{2}{n} \textsc{pss}\big(\widehat{A}_{d},\tilde{h}_{d}\big)\nonumber\\
& = \frac{1}{n}\sum^{n}_{i=1} \sum_{\ell = 1}^{k-1} \varepsilon^2_{\ell}\big( X_i, A^{\text{o} \top}_{d} X_i \big) - \frac{2}{n}\sum^{n}_{i=1} \sum_{\ell = 1}^{k-1} \varepsilon_{\ell}\big( X_i, A^{\text{o} \top}_{d} X_i \big) \Big(\widehat{\pi}^{-i}_{\ell, \tilde{h}_{d}}\big(\widehat{A}^{\top}_{d}X_{i};\widehat{A}_{d}\big) -\pi_{\ell}(A^{\text{o}\top}_{d} X_{i}\big)\Big) \nonumber\\
& \hspace{0.2in}+ \frac{1}{n}\sum^{n}_{i=1} \sum_{\ell = 1}^{k-1} \Big( \widehat{\pi}^{-i}_{\ell, \tilde{h}_{d}}\big(\widehat{A}^{\top}_{d}X_{i};\widehat{A}_{d}\big) -\pi_{\ell}(A^{\text{o}\top}_{d} X_{i}\big) \Big)^{2}\nonumber
\end{align}
\begin{align}
& \stackrel{\triangle}{=} \left\{
\begin{array}{ll}
I_{1}({d})+I_{2}({d})+O_{p}\big( n^{-\frac{2r}{2r+d}} \big) + O_{p}\big( \| \vvec(\widehat{A}_{d*}) -\vvec(A^{\text{o}}_{d*}) \|^2 \big)
& d < d^{\text{o}},\\\\
I_{1}({d})+I_{2}({d})+O_{p}\big( n^{-\frac{2r}{2r+d}} \big)
& d \geq d^{\text{o}}.
\end{array}\right. \label{eq:thm2_17}
\end{align}
Applying the central limit theorem to $I_1(d)$ for $d < d^{\text{o}}$ and using the fact that $I_{1}({d}) - I_{1}(d^{\text{o}}) = 0$ for $d \geq d^{\text{o}}$, we have
\begin{align}
I_{1}({d}) - I_{1}(d^{\text{o}}) = \left\{
\begin{array}{ll}
b^{2}(d) + O_p\big( n^{-\frac{1}{2}} \big)
& d < d^{\text{o}},\\\\
0
& d > d^{\text{o}},
\end{array}\right. \label{eq:thm2_18}
\end{align}
where $b^{2}(d) = \E[ \sum_{\ell = 1}^{k-1} ( \pi_{\ell}^{[0]}(A^{\text{o}\top}_{d} X\big) - \pi_{\ell}(A^{\text{o}\top}X))^2]$.
The second term in (\ref{eq:thm2_17}) admits the expansion
\begin{align}
I_{2}(d) = &\sum_{\ell = 1}^{k-1} \bigg(   \frac{1}{n} \sum_{i=1}^n  \varepsilon_{\ell}\big( X_i, A^{\text{o} \top}_{d} X_i \big) \pi^{[1]}_{\ell} \big( X_i, A^{\text{o}\top}_{d}X_i \big) \bigg)^{\top} \big(\vvec(\widehat{A}_{d*})-\vvec(A_{d*}^{\text{o}})\big)(1+o_p(1)) \nonumber\\
&+ \sum_{\ell = 1}^{k-1} \bigg( \frac{1}{n(n-1)} \sum_{i \neq j} \varepsilon_{\ell}\big( X_i, A^{\text{o} \top}_{d} X_i \big) \xi^{[1]}_{j,\ell} \big(X_i, A^{\text{o}}_{d}\big)  \bigg)^{\top} \big(\vvec(\widehat{A}_{d*})-\vvec(A_{d*}^{\text{o}})\big)(1+o_p(1)) \nonumber \\
& + \sum_{\ell = 1}^{k-1} \frac{1}{n(n-1)} \sum_{i \neq j} \varepsilon_{\ell}\big( X_i, A^{\text{o} \top}_{d} X_i \big) \xi^{[0]}_{j,\ell} \big(X_i, A^{\text{o}}_{d}\big) + O_p\Big( n^{-\frac{2r}{2r+d}} \Big). \label{eq:thm2_19}
\end{align}
By applying the $\sqrt{n}$-consistency of $\vvec(\widehat{A}_{d*})$ to $\vvec(A_{d*}^{\text{o}})$ in (\ref{eq:thm2_15_d0}), the central limit theorem to $\sum_{i=1}^n  \varepsilon_{\ell}\big( X_i, A^{\text{o} \top}_{d} X_i \big) \pi^{[1]}_{\ell} \big( X_i, A^{\text{o}\top}_{d}X_i \big) /n$, and the uniform convergence rates in (\ref{eq:thm2_9}) and (\ref{eq:thm2_9.5}) for (\ref{eq:thm2_7.5}) and (\ref{eq:thm2_8}), we obtain
\begin{align}
I_{2}(d)= O_{p}\big(n^{-\frac{2r}{2r+d}}\big) \text{ for } d \geq d^{\text{o}}. \label{eq:thm2_20}
\end{align}
In contrast, for $d < d^{\text{o}}$, the first two terms in $I_2(d)$ are $O_{p}\big( \|\vvec(\widehat{A}_{d*})-\vvec(A_{d*}^{\text{o}}) \|\big)$.
Using this result and applying Theorem 2.5 in \cite{Kosorok2008} to the third term imply 
\begin{align}
I_{2}(d)= O_{p}\big(n^{-\frac{r}{2r+d}}\big) + O_{p}\big( \|\vvec(\widehat{A}_{d*})-\vvec(A_{d*}^{\text{o}}) \|\big) \text{ for } d < d^{\text{o}}. \label{eq:thm2_21}
\end{align}
Substituting (\ref{eq:thm2_18}), (\ref{eq:thm2_20}), and (\ref{eq:thm2_21}) into $2\big( \textsc{pss}\big(\widehat{A}_{d}, \tilde{h}_{d}\big) - \textsc{pss}\big(\widehat{A}_{d^{\text{o}}}, \tilde{h}_{d^{\text{o}}}\big) \big)/n$, where the terms in the difference are specified in (\ref{eq:thm2_17}), yields
\begin{align}
\textsc{pss}(d)-\textsc{pss}(d^{\text{o}}) = \left\{
\begin{array}{ll}
\frac{\log n }{2 n^{\frac{2r}{2r+d}}} (1+o_{p}(1))
& d > d^{\text{o}}, \\\\
 b^{2}(d) (1+o_{p}(1))
&d < d^{\text{o}}.
\end{array}\right. \label{eq:thm2_22}
\end{align}
The consistency of $\hat{d}$ to $d^{\text{o}}$ follows directly from (\ref{eq:thm2_22}). 

\subsubsection*{(Asymptotic normality of $P_{\widehat{A}}$)}
The consistency of $\widehat{d}$ to $d^{\text{o}}$ and the boundedness of $\vvec(P_{\widehat{A}} - P_{A^{\text{o}}})$ imply
\begin{align}
P\big(\big\| \sqrt{n} \vvec \big(P_{\widehat{A}} - P_{A^{\text{o}}}\big) I\big(\widehat{d} \neq d^{\text{o}}\big) \big\| > \varepsilon\big) \leq P\big( \widehat{d} \neq d^{\text{o}} \big) \longrightarrow 0 \text{ as } n \longrightarrow \infty.   \label{eq:thm2_23}
\end{align}
Together with (\ref{eq:thm2_15_d}) and the decomposition 
\begin{align*}
\sqrt{n} \vvec(P_{\widehat{A}} - P_{A^{\text{o}}}) = \sqrt{n} \vvec(P_{\widehat{A}_{d^{\text{o}}}} - P_{A^{\text{o}}}) I(\widehat{d} = d^{\text{o}}) +  \sqrt{n} \vvec(P_{\widehat{A}} - P_{A^{\text{o}}}) I(\widehat{d} \neq d^{\text{o}}),
\end{align*} 
Slutsky's theorem yields
\begin{align}
 \sqrt{n} \vvec \big(P_{\widehat{A}} - P_{A^{\text{o}}}\big) = \sqrt{n} \vvec \big(P_{\widehat{A}_{d^{\text{o}}}} - P_{A^{\text{o}}}\big) (1+o_p(1)) \stackrel{d}{\longrightarrow} N_{(p-d^{\text{o}})d^{\text{o}}}(0, \Sigma). \label{eq:thm2_24}
\end{align}
The proof is completed.

\subsection{Proof of Theorem \ref{Thm3.4}} 
\label{sec:A.4}
\subsubsection*{(Consistency of $\hat{k}^{*}_{1}$ to $k^{*}$)}
Define
\begin{align*}
N_{s,h} \big(y,v; \mathcal{G}, \gamma \big)=\frac{1}{n}\sum_{i = 1}^n \big( I(Y_i \leq y) \big)^s \sum_{\ell = 1}^{k} I\big( i \in \mathcal{G}_{\ell}\big) K_{2,h}\big( \gamma^{\top}_{\ell} Z_i - v \big), s=0,1.
\end{align*}
Since $I(Y\leq y)$ and $K_2(u)$ are uniformly bounded, it follows that
\begin{align}
&\sup_{\{y,v,\gamma\}}\big| N_{s,\tilde{h}}\big(y,v; \widetilde{\mathcal{G}}, \gamma\big) - N_{s,\tilde{h}}\big(y,v; \mathcal{G}^{\text{o}}, \gamma\big) \big| \nonumber\\
&=\sup_{\{y,v,\gamma\}}\bigg | \frac{1}{n} \sum_{i = 1}^n \big( I(Y_i \leq y) \big)^s \sum_{\ell = 1}^{k} \big( I\big(i\in\widetilde{\mathcal{G}}_{\ell}\big)-I\big(i\in \mathcal{G}^{\text{o}}_{\ell} \big) \big) K_{2,\tilde{h}} \big( \gamma^{\top}_{\ell}Z_i - v \big) \bigg|  \nonumber \\
&\leq \sup_{\{ i, \ell\}} \Big| I\big(i\in\widetilde{\mathcal{G}}_{\ell}\big)-I\big(i\in\mathcal{G}^{\text{o}}_{\ell}\big)\Big|\cdot
\sup_{\{y,v,\gamma\}}\bigg | \frac{1}{n}\sum^{n}_{i=1}(I(Y_{i}\leq y))^{s}\sum^{k}_{\ell=1}K_{\tilde{h}}\big(\gamma^{\top}_{\ell}Z_{i}-v\big) \bigg | \nonumber \\
&=o_{p}\big( n^{-1} \big), s=0,1.  \label{eq:A.4.0}
\end{align}
Applying arguments analogous to those used in the proof of Lemma \ref{A:Lemma_1}, we obtain
\begin{align}
\sup_{\{v,\gamma\}}\big| N_{0,\tilde{h}}\big(y,v; \mathcal{G}^{\text{o}}, \gamma\big)  - f_{\gamma_{*}}(v) \big| = o_p\bigg( \frac{\sqrt{\log n}}{n^{\frac{2}{5}}} \bigg). \label{eq:A.4.1}
\end{align}
By (\ref{eq:A.4.0}) and (\ref{eq:A.4.1}), we have
\begin{align}
&\sup_{\{y,v,\gamma_{*}\}} \Big|\widehat{G}_{\tilde{h}}\big(y, v; \widetilde{\mathcal{G}},\gamma\big)-\widehat{G}_{\tilde{h}}\big(y, v; \mathcal{G}^{\text{o}},\gamma\big) \Big | \hspace{3.4in} \nonumber\\
& = \sup_{\{y,v,\gamma_{*}\}}  \frac{ \big| \sum_{s=0}^1 (-1)^s \big( N_{s,\tilde{h}}\big(y,v; \widetilde{\mathcal{G}}, \gamma\big) - N_{s,\tilde{h}}\big(y,v; \mathcal{G}^{\text{o}}, \gamma\big) \big) N_{1-s,\tilde{h}}\big(y,v; \mathcal{G}^{\text{o}}, \gamma\big) \big|}{N_{0,\tilde{h}}\big(y,v; \widetilde{\mathcal{G}},\gamma\big) N_{0,\tilde{h}}\big(y,v;\mathcal{G}^{\text{o}}, \gamma\big)}  \nonumber \\
& \leq  \frac{ \sum_{s=0}^1  \sup_{\{y,v,\gamma\}}\big|\big( N_{s,\tilde{h}}\big(y,v; \widetilde{\mathcal{G}},\gamma\big) - N_{s,\tilde{h}}\big(y,v; \mathcal{G}^{\text{o}}, \gamma\big) \big)\big|  \sup_{\{y,v,\gamma\}} \big| N_{1-s,\tilde{h}}\big(y,v; \mathcal{G}^{\text{o}}, \gamma\big)\big|}{ \Big( \inf_{\{v,\gamma\}} f_{\gamma_{*}}(v) + o_p\Big( \frac{\sqrt{\log n}}{n^{\frac{2}{5}}} \Big) - O_p\big( n^{-1}\big)\Big)  \Big( \inf_{\{v,\gamma\}} f_{\gamma_{*}}(v) + o_p\Big( \frac{\sqrt{\log n}}{n^{\frac{2}{5}}} \Big)\Big)  }  \nonumber \\
&=o_{p}\big( n^{-1} \big).\label{eq:A.4.2}
\end{align}
The uniform convergence of $\widehat{G}_{\tilde{h}}\big(y, \gamma^{\top}w;\mathcal{G}^{\text{o}},{\gamma}\big)$ to $G^{[0]}(y,w,\gamma^{\top}w)$ follows by tthe same reasoning as in Lemma \ref{A:Lemma_1}.
Moreover, Lemma \ref{Thm2.1} implies
\begin{align*}
\widetilde{\gamma} - \gamma^{\text{o}}=\left\{
 \begin{array}{ll}
 O_{p}\big( n^{-\frac{1}{2}} \big) & \mathcal{G}^{\text{o}} \in \mathcal{C}_{\mathcal{G}},\\\\
  o_{p}(1) & \text{ otherwise.}
  \end{array}\right.
\end{align*}
Thus,
\begin{align}
&\sup_{\{y,w\}}\Big |\widehat{G}_{\tilde{h}}\big(y, \widetilde{\gamma}^{\top}w; \mathcal{G}^{\text{o}},\widetilde{\gamma}\big)-G(y,\gamma^{\text{o}\top}w) \Big |\nonumber \\
&\leq \sup_{\{y,w, \gamma_{*}\}}\Big |\widehat{G}_{\tilde{h}}\big(y, \gamma^{\top}w; \mathcal{G}^{\text{o}},{\gamma}\big)-G^{[0]}(y,w,\gamma^{\top}w) \Big | + \sup_{\{y,v,w\}}\Big |\partial_{v} G^{[0]}(y,v) \Big | \Big\| \widetilde{\gamma} - \gamma^{\text{o}} \Big\| \| w \| \nonumber \\
&= \left\{
\begin{array}{ll}
o_p\Big( \frac{\sqrt{\log n}}{n^{\frac{2}{5}}} \Big) + O_{p}\big( n^{-\frac{1}{2}} \big)
&\mathcal{G}^{\text{o}} \in \mathcal{C}_{\mathcal{G}}, \\\\
o_p\Big( \frac{\sqrt{\log n}}{n^{\frac{2}{5}}} \Big) + O_{p}\big( \| \widetilde{\gamma} - \gamma^{\text{o}} \| \big)
& \text{otherwise.}
\end{array}\right. \label{eq:A.4.3}
\end{align}
Combining (\ref{eq:A.4.2}) and (\ref{eq:A.4.3}) and applying the triangle inequality, we obtain
\begin{align}
&\sup_{\{y,w\}}\Big |\widehat{G}_{\tilde{h}}\big(y, \widetilde{\gamma}^{\top}w; \widetilde{\mathcal{G}},\widetilde{\gamma}\big)-G(y,\gamma^{\text{o} \top}w) \Big |= \left\{
\begin{array}{ll}
o_p\Big( \frac{\sqrt{\log n}}{n^{\frac{2}{5}}} \Big)
& \mathcal{G}^{\text{o}} \in \mathcal{C}_{\mathcal{G}}, \\\\
o_p\Big( \frac{\sqrt{\log n}}{n^{\frac{2}{5}}} \Big) + O_{p}\big( \| \widetilde{\gamma} - \gamma^{\text{o}} \| \big)
& \text{otherwise.}
\end{array}\right.\label{eq:A.4.4}
\end{align}
Using the decomposition of $\textsc{psis}\big(\widetilde{\gamma},\tilde{h};\widetilde{\mathcal{G}}\big)$, together with (\ref{eq:A.4.4}), we have
\begin{align}
&\frac{2}{n} \textsc{psis}\big(\widetilde{\gamma},\tilde{h};\widetilde{\mathcal{G}}\big)\nonumber\\
& = \int\frac{1}{n}\sum^{n}_{i=1} \big( \varepsilon_{y}\big( Y_i, \gamma^{\text{o} \top} W_i \big) \big)^{2}d\widehat{F}(y)-2\int \frac{1}{n}\sum^{n}_{i=1} \varepsilon_{y}\big( Y_i, \gamma^{\text{o} \top} W_i \big) \Big(\widehat{G}^{-i}_{\tilde{h}}\big(y, \widetilde{\gamma}^{\top}\widetilde{W}_{i};\widetilde{\mathcal{G}},\widetilde{\gamma}\big) \nonumber\\
& \hspace{0.2in}-G^{[0]}(y,\gamma^{\text{o}\top}W_{i}\big)\Big)d\widehat{F}(y)+ \int\frac{1}{n}\sum^{n}_{i=1} \Big( \widehat{G}^{-i}_{\tilde{h}}\big(y, \widetilde{\gamma}^{\top}\widetilde{W}_{i};\widetilde{\mathcal{G}},\widetilde{\gamma}\big) -G^{[0]}(y,\gamma^{\text{o}\top}W_{i}\big) \Big)^{2}d\widehat{F}(y) \nonumber\\
& \stackrel{\triangle}{=} \left\{
\begin{array}{ll}
I_{1}(k)+I_{2}(k)+ o_p\Big( \frac{{\log n}}{n^{\frac{4}{5}}} \Big)
&\mathcal{G}^{\text{o}} \in \mathcal{C}_{\mathcal{G}}, \\\\
I_{1}(k)+I_{2}(k)+ o_p\Big( \frac{{\log n}}{n^{\frac{4}{5}}} \Big) + O_{p}\big( \| \widetilde{\gamma} - \gamma^{\text{o}} \|^2 \big)
& \text{otherwise,}
\end{array}\right. \label{eq:A.4.5}
\end{align}
where $W_{i}=(I(i\in\mathcal{G}^{\text{o}}_{1}), \dots, I(i\in\mathcal{G}^{\text{o}}_{k}))^{\top} \otimes Z_{i}$, $\widetilde{W}_{i}=(I(i\in\widetilde{\mathcal{G}}_{1}), \dots, I(i\in\widetilde{\mathcal{G}}_{k}))^{\top} \otimes Z_{i}$, and $ \varepsilon_{y}\big( Y_i, \gamma^{\text{o} \top} W_i \big) = I(Y_{i}\leq y)-G^{[0]}\big(y,\gamma^{\text{o}\top}W_{i}\big)$, $i=1,\dots, n$. 

By noting that $I_1(k) - I_1(k^{*})=0$ whenever $\mathcal{G}^{\text{o}} \in \mathcal{C}_{\mathcal{G}}$, and invoking Theorem 8.1 in \cite{hoeffding1948probability} alongside the consistency of $U$-statistics in \cite{hoeffding1961strong}, we obtain 
\begin{align}
I_{1}(k) - I_{1}(k^{*}) = \left\{
\begin{array}{ll}
0
&\mathcal{G}^{\text{o}} \in \mathcal{C}_{\mathcal{G}},\\ \\
b^{2}_{1}(k) (1+o_p(1))
& \text{otherwise,}
\end{array}\right. \label{eq:A.4.6}
\end{align}
where $b^{2}_{1}(k)= \int\E\big[ \big(G^{[0]}\big(y,\gamma^{\text{o}\top}W\big)-G\big(y,\gamma^{*\text{o}\top}W^{*}\big)\big)^{2}\big]dF(y)$, with
$W^{*}=(I(C^{*}=1), \dots, I(C^{*}=k^{*}))^{\top} \otimes Z$. Observe that the process $\{\varepsilon_{y}\big( Y_i,$ $\gamma^{\text{o} \top} W_i \big):y\in\mathcal{Y}\}$ degenerates to the zero process $\{\varepsilon_y(Y_{i}, \gamma^{* \text{o} \top}W^{*}_{i}):y\in\mathcal{Y}\}$ for all $i=1,\dots, n$ when $\mathcal{G}^{\text{o}} \in \mathcal{C}_{\mathcal{G}}$.
From (\ref{eq:A.4.2}) and a standard Taylor expansion, we derive the decomposition
\begin{align}
\widehat{G}^{-i}_{\tilde{h}}\big(y, \widetilde{\gamma}^{\top}\widetilde{W}_{i};\widetilde{\mathcal{G}},\widetilde{\gamma}\big)-G^{[0]}(y,\gamma^{\text{o}\top}W_{i}\big)
=&\big(\partial_{\gamma_{*}}\widehat{G}^{-i}_{\tilde{h}}\big(y, \bar{\gamma}^{\top}W_{i};\mathcal{G}^{\text{o}},\bar{\gamma}\big)\big)^{\top}\big(\widetilde{\gamma}_{*}-\gamma^{\text{o}}_{*}\big)+\widehat{G}^{-i}_{\tilde{h}}\big(y,\gamma^{\text{o}\top}W_{i};\mathcal{G}^{\text{o}},\gamma^{\text{o}}\big) \nonumber\\
&-G^{[0]}(y,\gamma^{\text{o}\top}W_{i}\big) + r_{0i}(y),\label{eq:A.4.7.0}
\end{align}
where $\bar{\gamma}$ lies on the line segment between $\widetilde{\gamma}$ and $\gamma^{\text{o}}$, and
\begin{align*}
r_{0i}(y)=\widehat{G}^{-i}_{\tilde{h}}\big(y, \widetilde{\gamma}^{\top}\widetilde{W}_{i};\widetilde{\mathcal{G}},\widetilde{\gamma}\big)-\widehat{G}^{-i}_{\tilde{h}}\big(y, \widetilde{\gamma}^{\top}W_{i};\mathcal{G}^{\text{o}},\widetilde{\gamma}\big),\text{ with }
\sup_{\{i,y\}}|r_{0i}(y)|=o_{p}(n^{-1}).
\end{align*}
Using the techniques in the proofs of Lemma \ref{A:Lemma_1} and the asymptotic normality of $\widehat{\gamma}_{*}$ in Lemma \ref{Thm2.1}, the expression in (\ref{eq:A.4.7.0}) simplifies to
\begin{align*}
&\widehat{G}^{-i}_{\tilde{h}}\big(y, \widetilde{\gamma}^{\top}\widetilde{W}_{i};\widetilde{\mathcal{G}},\widetilde{\gamma}\big)-G^{[0]}(y,\gamma^{\text{o}\top}W_{i}\big)\\
&= \bigg( G^{[1]}\big( y, W_i, \gamma^{\text{o} \top}W_i \big) + \frac{1}{n-1} \sum_{\{j:j \neq i\}} \xi^{[1]}_j \big(y, W_i, \gamma^{\text{o}}\big) + r_{1i}(y) \bigg)^{\top} \big(\widetilde{\gamma}_{*}-\gamma^{\text{o}}_{*}\big)  \\
& \hspace{0.2 in} + \bigg( \frac{1}{n-1} \sum_{\{j:j \neq i\}} \xi^{[0]}_j \big(y, W_i, \gamma^{\text{o}}\big) + r_{2i}(y) \bigg) + r_{0i}(y),  
\end{align*}
where $\sup_{\{i,y\}}\|r_{1i}(y)\| =o_{p}(n^{-2/5} \sqrt{\log n})$ and $\sup_{\{i,y\}}|r_{2i}(y)|=o_{p}(n^{-4/5}{\log n})$. Accordingly, the second term in (\ref{eq:A.4.5}) becomes
\begin{align*}
I_2(k) = &\bigg( \int  \frac{1}{n} \sum_{i=1}^n \varepsilon_{y}\big( Y_i, \gamma^{*\text{o} \top} W^{*}_i \big) G^{[1]}\big( y, W_i, \gamma^{\text{o} \top}W_i \big)  d\widehat{F}(y) \bigg)^{\top} \big(\widetilde{\gamma}_{*}-\gamma^{\text{o}}_{*}\big)\\
&+ \bigg( \int\frac{1}{n(n-1)} \sum_{\{j:j \neq i\}} \varepsilon_{y}\big( Y_i, \gamma^{*\text{o} \top} W^{*}_i \big) \xi^{[1]}_j \big(y, W_i, \gamma^{\text{o}}\big) d\widehat{F}(y) \bigg)^{\top} \big(\widetilde{\gamma}_{*}-\gamma^{\text{o}}_{*}\big) \nonumber \\
& + \int \frac{1}{n(n-1)} \sum_{\{j:j \neq i\}} \varepsilon_{y}\big( Y_i, \gamma^{*\text{o} \top} W^{*}_i \big) \xi^{[0]}_j \big(y, W_i, \gamma^{\text{o}}\big) d\widehat{F}(y)+ o_p\bigg( \frac{{\log n}}{n^{\frac{4}{5}}} \bigg). 
\end{align*}
Applying Theorem 2.5 in \citet{Kosorok2008} to the first term, and using the techniques in (\ref{eq:thm1_9}) and (\ref{eq:thm1_9.5}) for the remaining two, we deduce the respective orders $O_p(n^{-1})$, $O_p(n^{-12/10})$, and $O_p(n^{-4/5})$ under $\mathcal{G}^{\text{o}} \in \mathcal{C}_{\mathcal{G}}$. Thus,
\begin{align}
I_{2}(k)= o_p\bigg( \frac{{\log n}}{n^{\frac{4}{5}}} \bigg) \text{ if } \mathcal{G}^{\text{o}} \in \mathcal{C}_{\mathcal{G}}.
\label{eq:A.4.9}
\end{align}
When $\mathcal{G}^{\text{o}} \notin \mathcal{C}_{\mathcal{G}}$, the expectations of the integrands in $I_2(k)$ are strictly nonzero. Together with (\ref{eq:A.4.4}), this yields
\begin{align}
I_{2}(k)= o_p\bigg( \frac{\sqrt{\log n}}{n^{\frac{2}{5}}} \bigg) + O_{p}\big( \| \widetilde{\gamma} - \gamma^{\text{o}} \|^2 \big). \label{eq:A.3.10}
\end{align}
Substituting (\ref{eq:A.4.6}), (\ref{eq:A.4.9}), and (\ref{eq:A.3.10}) into (\ref{eq:A.4.5}) gives
\begin{align}
\textsc{spic}_{1}(k)-\textsc{spic}_{1}(k^{*}) = \left\{
\begin{array}{ll}
\max\Big\{ b^{2}_{1}(k),(k-k^{*})\frac{\log n }{2 n^{\frac{4}{5}}}\Big\}(1+o_{p}(1))
& k > k^{*}, \\\\
 b^{2}_{1}(k) (1+o_{p}(1))
&\ell <k.
\end{array}\right. \label{eq:A.3.11}
\end{align}
The consistency of $\widehat{k}^{*}_{1}$ to $k^{*}$ then follows directly from (\ref{eq:A.3.11}). 

\subsubsection*{(Consistency of $\hat{k}^{*}_{2}$ to $k^{*}$)}

It follows from (\ref{eq:A.4.4}) that
\begin{align}
\sup_{\{y,w\}} \Big| I\Big(\widehat{G}_{\tilde{h}}\big(y, \widetilde{\gamma}^{\top}w; \widetilde{\mathcal{G}},\widetilde{\gamma}\big) = s \Big)-I \Big(G^{[0]}(y,\gamma^{\text{o} \top}w) = s\Big) \Big|=o_{p}(n^{-1}), s = 0, 1.
\label{eq:A.5.1}
\end{align}
Combining (\ref{eq:A.3.1}), (\ref{eq:A.4.4}), and (\ref{eq:A.5.1}) with a Taylor expansion yields
\begin{align}
\textsc{spic}_2(k) =& - \frac{1}{n^2} \sum_{i = 1}^n\sum_{j = 1}^n \bigg[ I(Y_i \leq Y_j) \log G^{[0]}_{0}\big(Y_j,\gamma^{\text{o} \top}W_i \big)+ I(Y_i > Y_j) \log \Big(1-  G^{[0]}_{1}\big(Y_j,\gamma^{\text{o} \top} W_i \big) \Big) \bigg]  \nonumber \\
&-  \int \frac{1}{n} \sum_{i = 1}^n   \Bigg( \frac{I(Y_i \leq y)}{G^{[0]}_{0}\big(y,\gamma^{\text{o} \top} W_i \big) } 
-\frac{I(Y_i > y)}{1-  G^{[0]}_{1}\big(y,\gamma^{\text{o} \top} W_i \big) } \Bigg)\Big(\widehat{G}^{-i}_{\tilde{h}}\big(y, \widetilde{\gamma}^{\top}\widetilde{W}_{i}; \widetilde{\mathcal{G}},\widetilde{\gamma}\big)\nonumber\\
&- G^{[0]}(y,\gamma^{\text{o}}W_{i}) \Big) d \widehat{F}(y)
+ O_{p}\big(n^{-\frac{4}{5}}\big) +  \frac{k \log n}{n^{\frac{4}{5}}}\nonumber\\
 \stackrel{\triangle}{=} &~ I\!I_{1}(k)+ I\!I_{2}(k) +  \frac{k \log n}{n^{\frac{4}{5}}},
\label{eq:A.5.3}
\end{align}
where 
\begin{align*}
&G^{[0]}_{s}\big(y,\gamma^{\text{o} \top}W_i \big) =G^{[0]}\big(y,\gamma^{\text{o} \top}W_i \big) +(-1)^{s} I\big( G^{[0]}\big(y,\gamma^{\text{o} \top}W_i \big) = s \big), s=0,1, i=1,\dots,n.
\end{align*}

Clearly, $I\!I_{1}(k)-I\!I_{1}(k^{*})=0$ whenever $\mathcal{G}^{\text{o}} \in \mathcal{C}_{\mathcal{G}}$. 
When $\mathcal{G}^{\text{o}} \notin \mathcal{C}_{\mathcal{G}}$, the uniform convergence of the estimated functions, together with Theorem 8.1 in \cite{hoeffding1948probability}, the consistency of $U$-statistics in \cite{hoeffding1961strong}, and Jensen’s inequality, implies that
\begin{align*}
& I\!I_{1}(k) - I\!I_{1}(k^{*})=b^{2}_{2}(k) (1+o_p(1))\nonumber \\
&>-\int \log \E\bigg[  \E\bigg[I(Y_i \leq y) \frac{G^{[0]}_{0}\big(y,\gamma^{\text{o}\top}W_i\big) }{G_{0}\big(y,\gamma^{*\text{o}\top}W^{*}_i\big) }+ I(Y_i >y)\frac{1-G^{[0]}_{1}\big(y,\gamma^{\text{o}\top}W_i\big)}{1-G_{1}\big(y,\gamma^{*\text{o}\top}W^{*}_i\big)}\bigg|  W_i \bigg]\bigg]d{F}(y) (1+o_p(1))  \\
&=- \int \log\E \Big[ I\big( G\big(y,\gamma^{*\text{o}\top}W^{*}_i\big) = 0 \big) \big( 1-G^{[0]}\big(y,\gamma^{\text{o}\top}W_i\big)  \big)  + I\big( G\big(y,\gamma^{*\text{o}\top}W^{*}_i\big) = 1 \big) \big( G^{[0]}_{0}\big(y,\gamma^{\text{o}\top}W_i\big)  \big) \\
&\hspace{0.2in}+ I\big( G\big(y,\gamma^{*\text{o}\top}W^{*}_i\big) \in (0,1) \big) \Big] dF(y) (1+o_p(1))= 0,
\end{align*}
where 
\begin{align*}
b^{2}_{2}(k)= -\int   \E\bigg[I(Y_i \leq y) \log\frac{G^{[0]}_{0}\big(y,\gamma^{\text{o}\top}W_i\big) }{G_{0}\big(y,\gamma^{*\text{o}\top}W^{*}_i\big) }+ I(Y_i >y)\log \frac{1-G^{[0]}_{1}\big(y,\gamma^{\text{o}\top}W_i\big)}{1-G_{1}\big(y,\gamma^{*\text{o}\top}W^{*}_i\big)}\bigg]d{F}(y)
\end{align*}
is zero when $\mathcal{G}^{\text{o}} \in \mathcal{C}_{\mathcal{G}}$ and strictly positive otherwise.
Thus,
\begin{align}
I\!I_{1}(k) - I\!I_{1}(k^{*})
&= \left\{
\begin{array}{ll}
0
&\mathcal{G}^{\text{o}} \in \mathcal{C}_{\mathcal{G}}, \\\\
b^{2}_{2}(k)(1+o_p(1)) 
& \text{otherwise.}
\end{array}\right. \label{eq:A.5.5}
\end{align}
Similar to the derivations of (\ref{eq:A.4.9}) and (\ref{eq:A.3.10}), we obtain
\begin{align}
I\!I_{2}(k)= \left\{
\begin{array}{ll}
o_p\Big( \frac{{\log n}}{n^{\frac{4}{5}}} \Big)
&\mathcal{G}^{\text{o}} \in \mathcal{C}_{\mathcal{G}}, \\\\
o_p\Big( \frac{\sqrt{\log n}}{n^{\frac{2}{5}}} \Big)
&\text{otherwise}.
\end{array}\right. \label{eq:A.5.6} 
\end{align}
Combining (\ref{eq:A.5.5}), (\ref{eq:A.5.6}), and the decomposition of $\textsc{spic}_{2}(k)$ in (\ref{eq:A.5.3}) yields
\begin{align}
\textsc{spic}_{2}(k)-\textsc{spic}_{2}(k^{*}) = \left\{
\begin{array}{ll}
\max\Big\{ b^{2}_{2}(k),(k-k^{*})\frac{\log n }{ n^{\frac{4}{5}}}\Big\}(1+o_{p}(1))
& k > k^{*}, \\\\
 b^{2}_{2}(k) (1+o_{p}(1))
& k < k^{*}.
\end{array}\right. \label{eq:A.5.7}
\end{align}
The consistency of $\hat{k}^{*}_{2}$ to $k^{*}$ follows immediately from (\ref{eq:A.5.7}).

\clearpage

\subsection{Pseudocode and Flowcharts}
\label{spcode-flowcharts}


\RestyleAlgo{ruled}

\begin{algorithm}
\footnotesize  
\caption*{Pseudocode for the SP estimation procedure}
Initialize $\widehat{\beta}^{(0)}, \widehat{\eta}^{(0)}, \widehat{\gamma}^{(0)}$, and set $\widehat{\nu}^{(0)} \gets 0$, $\varepsilon_1 \gets \text{ predefined tolerance}$ \;
\Begin{
 \For{$m = 0, 1, 2, \cdots$}{
  Set ${\beta}^{(0)} \gets \widehat{\beta}^{(m)}$, $\varepsilon_2 \gets \text{ predefined tolerance}$ \;
  \For{$s = 0, 1, 2, \cdots$}{
  Compute ${\beta}^{(s+1)}_{(1)*}$ and ${\beta}^{(s+1)}_{(2)*}$ using (\ref{alg:SP_beta})\;
  \eIf{$\| {\beta}^{(s+1)} - {\beta}^{(s)} \| < \varepsilon_2$,}{
Set $\widehat{\beta}^{(m+1)} = {\beta}^{(s+1)}$ and exit the inner loop\;
   }{$s = s+1$ \;
   }{
  }
   }{${\eta}^{(0)} \gets \widehat{\eta}^{(m)}, {\gamma}^{(0)}_{(1)} \gets \widehat{\gamma}^{(m)}_{(1)} $ \;
    \For{$s = 0, 1, 2, \cdots$}{
    Set ${\delta}_{i\ell}^{(0)} \gets {\eta}_{i}^{(s)} - {\gamma}_{(1)\ell}^{(s)}, {u}^{(0)} \gets 0$, $\varepsilon_3 \gets \text{ predefined tolerance}$  \;
      \For{$s_1 = 0, 1, 2, \cdots$}{
      Compute $\bar{\eta}^{(s_1+1)}$ using (\ref{alg:SP_eta})\;
      Compute ${\delta}^{(s_1+1)}$ using (\ref{alg:SP_delta})\;
      Compute ${u}^{(s_1+1)}$ using (\ref{alg:SP_u})\;
        \eIf{$\sum_{i=1}^n \sum_{\ell=1}^k \| \bar{\eta}^{(s_1+1)}_i - {\gamma}^{(s)}_{(1)\ell} - {\delta}^{(s_1+1)}_{i\ell} \| < \varepsilon_3$}{
         Set ${\eta}^{(s+1)} = \bar{\eta}^{(s_1+1)}$ and exit the inner loop\;
        }{$s_1 = s_1+1$ \;
        }{
        }
      }
    Compute ${\gamma}^{(s+1)}_{(1)}$ using (\ref{alg:SP_gamma})\;
                \eIf{$\textsc{psisp}_{\textsc{sp}}\big(\widehat{\beta}^{(m+1)},{\eta}^{(s+1)},{\gamma}^{(s+1)}_{(1)},\widehat{\nu}^{(m)};\lambda \big) > \textsc{psisp}_{\textsc{sp}}\big(\widehat{\beta}^{(m+1)},{\eta}^{(s)},{\gamma}^{(s)}_{(1)},\widehat{\nu}^{(m)};\lambda \big)$}{
 Set $\widehat{\eta}^{(m+1)} = {\eta}^{(s)}$ and $\widehat{\gamma}^{(m+1)}_{(1)} = {\gamma}^{(s)}_{(1)}$, and exit the loop\;
        }{$s = s+1$ \;
        }{
        }
   }
   }{
  }
  Compute $\widehat{\nu}^{(m+1)}$ using (\ref{alg:SP_nu})\;
  \eIf{$\| \widehat{\beta}^{(m+1)}_{(1)}- \widehat{\eta}^{(m+1)} \| < \varepsilon_1$}{
Set $\widehat{\beta}^{\lambda} = \widehat{\beta}^{(m+1)}$ and $\widehat{\gamma}^{\lambda}_{(1)} = \widehat{\gamma}^{(m+1)}_{(1)}$; terminate\;
   }{$m = m+1$ \;
   }{
  }
 }

  }
\end{algorithm}

\begin{figure}[htbp]
\centering
\begin{subfigure}[b]{0.8\textwidth}
    \centering
    \includegraphics[width=\textwidth]{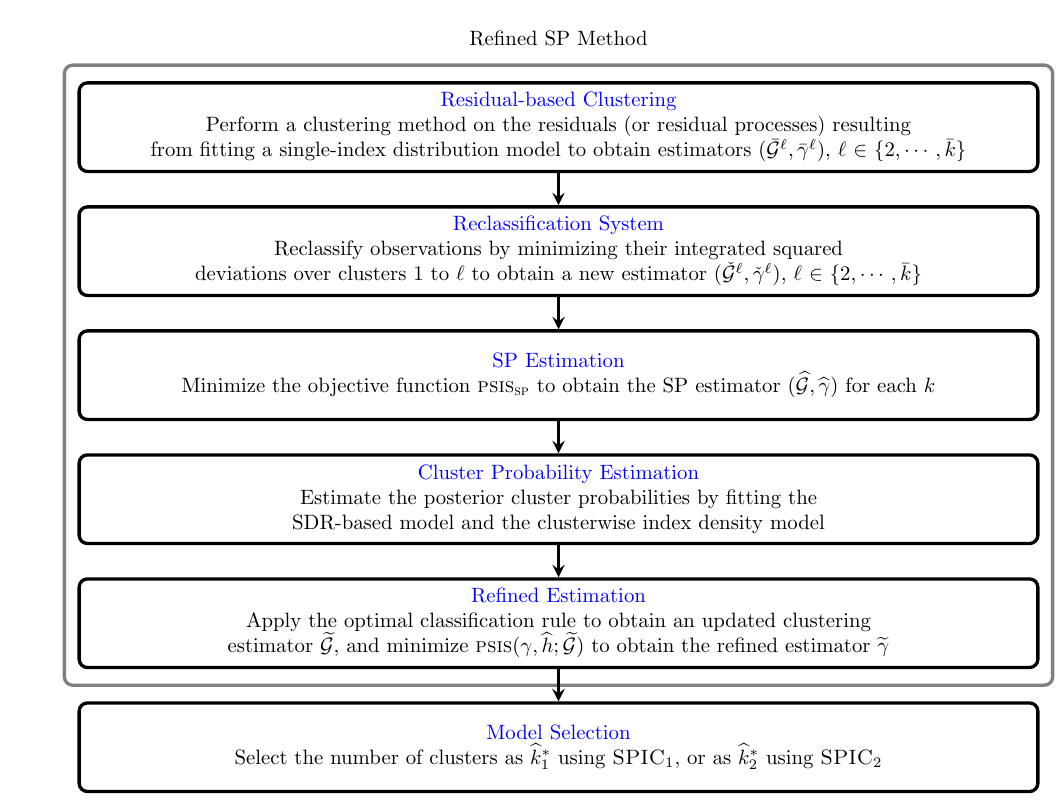}
\end{subfigure}

\vspace{0.5cm}

\begin{subfigure}[b]{1\textwidth}
    \centering
    \includegraphics[width=\textwidth]{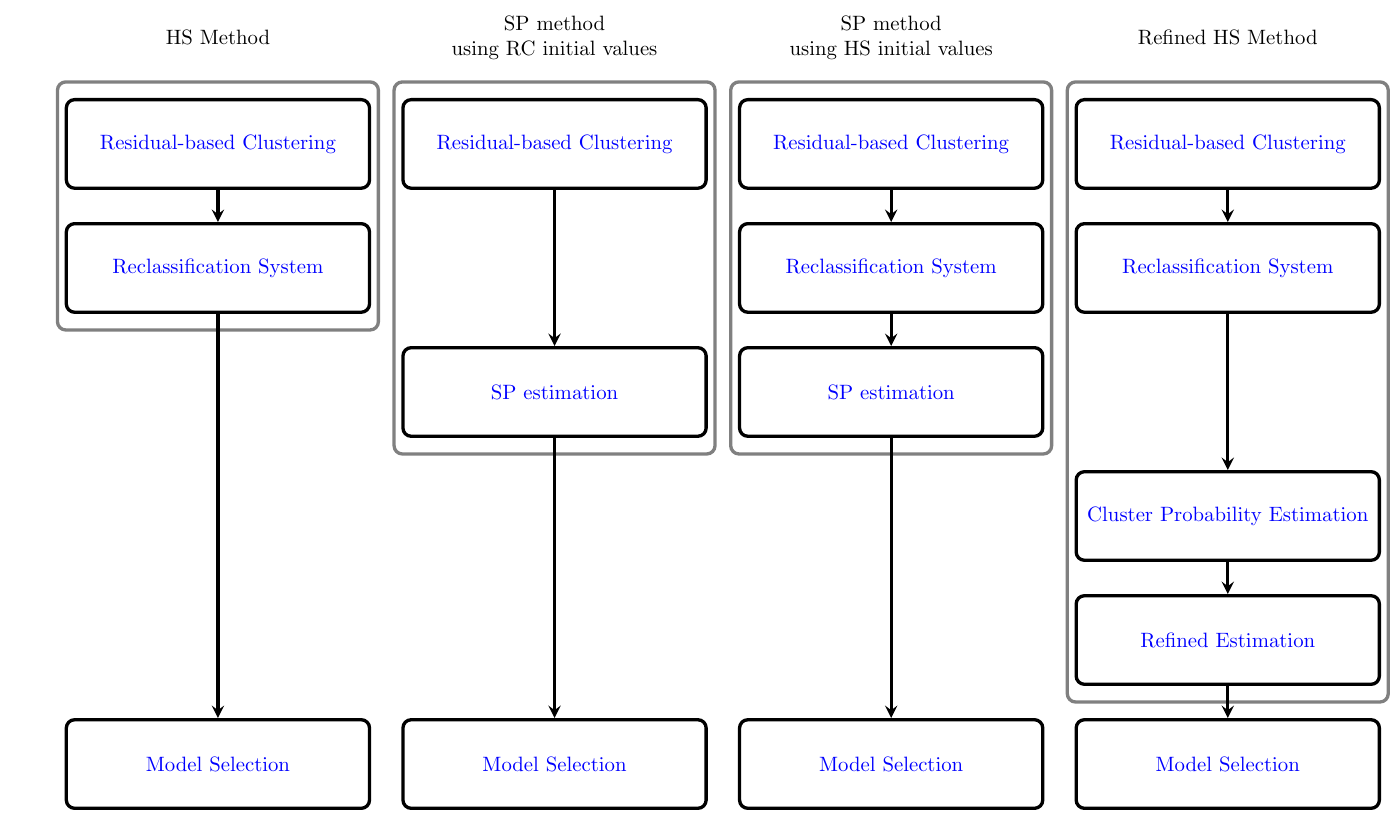}
\end{subfigure}

\caption*{The flowcharts illustrating the refined SP method and related methods derived from its estimation procedure.}
\end{figure}

\newpage
\subsection{Supplementary Tables and Figures}
\label{table-figure}

\begin{table}[H]
  \centering
  \caption{Means (standard deviations) of 500 RI values of the true partition and clustering estimates from various methods under different CID and cluster membership models.}
      {\small
      \begin{adjustbox}{max width=\linewidth}
    \begin{tabular}{ccccccccccc}
    \toprule
      \multicolumn{2}{c}{Method} & SP    & FP    & FP   & RHS & & SP    & FP    & FP   & RHS \\
      \cmidrule(rl){3-6} \cmidrule(rl){8-11}  
      \multicolumn{2}{c}{Initial value} & RC    & RC    & HS   & HS & & RC    & RC    & HS   & HS \\
      \cmidrule(rl){3-6} \cmidrule(rl){8-11}  
    CID   & $n$     & \multicolumn{4}{c}{Covariate-independence} & & \multicolumn{4}{c}{Covariate-dependence}\\
    \midrule          
 M1    & 100   & 0.924 (0.0716)& 0.923 (0.0732)& 0.941 (0.0673)& 0.955 (0.0693)&       & 0.888 (0.0975)& 0.883 (0.1020)& 0.954 (0.0878)& 0.954 (0.0982)\\
          & 200   & 0.974 (0.0246)& 0.973 (0.0281)& 0.975 (0.0252)& 0.985 (0.0206)&       & 0.946 (0.0409)& 0.945 (0.0593) & 0.986 (0.0444)& 0.987 (0.0443)\\
          & 300   & 0.982 (0.0169)& 0.982 (0.0164)& 0.983 (0.0156)& 0.990 (0.0123)&       & 0.967 (0.0485)& 0.967 (0.0549)& 0.992 (0.0322)& 0.993 (0.0240)\\
          & 400   & 0.987 (0.0129)& 0.987 (0.0127)& 0.988 (0.0118)& 0.992 (0.0080)&       & 0.974 (0.0589)& 0.973 (0.0504)& 0.993 (0.0359)& 0.994 (0.0289) \\
          \midrule  
M2    & 100   & 0.858 (0.0805) & 0.858 (0.0804) & 0.877 (0.0744) & 0.894 (0.0897) &       & 0.823 (0.0894) & 0.828 (0.0912) & 0.904 (0.0698) & 0.908 (0.0817) \\
      & 200   & 0.917 (0.0467) & 0.919 (0.0458) & 0.921 (0.0439) & 0.939 (0.0363) &       & 0.891 (0.0494) & 0.889 (0.0558) & 0.946 (0.0309) & 0.950 (0.0290) \\
      & 300   & 0.934 (0.0361) & 0.935 (0.0352) & 0.936 (0.0336) & 0.951 (0.0277) &       & 0.919 (0.0409) & 0.919 (0.0434) & 0.956 (0.0205) & 0.957 (0.0207) \\
      & 400   & 0.944 (0.0285) & 0.944 (0.0283) & 0.944 (0.0279) & 0.956 (0.0216) &       & 0.923 (0.0390) & 0.924 (0.0410) & 0.958 (0.0175) & 0.960 (0.0166) \\
\midrule  
M3    & 100   & 0.885 (0.0847) & 0.884 (0.0845) & 0.897 (0.0816) & 0.918 (0.0774) &       & 0.833 (0.0977) & 0.831 (0.1027) & 0.907 (0.0915) & 0.912 (0.1046) \\
      & 200   & 0.934 (0.0444) & 0.936 (0.0434) & 0.937 (0.0434) & 0.959 (0.0321) &       & 0.895 (0.0700) & 0.887 (0.0826) & 0.948 (0.0644) & 0.954 (0.0650) \\
      & 300   & 0.951 (0.0312) & 0.951 (0.0303) & 0.952 (0.0299) & 0.967 (0.0221) &       & 0.908 (0.0793) & 0.909 (0.0772) & 0.956 (0.0594) & 0.965 (0.0485) \\
      & 400   & 0.959 (0.0250) & 0.959 (0.0244) & 0.959 (0.0243) & 0.970 (0.0193) &       & 0.926 (0.0663) & 0.926 (0.0663) & 0.964 (0.0478) & 0.969 (0.0455) \\
\midrule  
M4    & 100   & 0.824 (0.0855) & 0.825 (0.0850) & 0.889 (0.0850) & 0.900 (0.0912) &       & 0.822 (0.0854) & 0.823 (0.0850) & 0.894 (0.0795) & 0.905 (0.0871) \\
      & 200   & 0.898 (0.0618) & 0.898 (0.0614) & 0.955 (0.0452) & 0.968 (0.0377) &       & 0.898 (0.0589) & 0.897 (0.0591) & 0.960 (0.0407) & 0.972 (0.0357) \\
      & 300   & 0.926 (0.0458) & 0.927 (0.0461) & 0.975 (0.0249) & 0.982 (0.0179) &       & 0.925 (0.0449) & 0.924 (0.0449) & 0.976 (0.0229) & 0.982 (0.0190) \\
      & 400   & 0.942 (0.0416) & 0.943 (0.0408) & 0.983 (0.0166) & 0.986 (0.0096) &       & 0.940 (0.0364) & 0.941 (0.0358) & 0.983 (0.0143) & 0.986 (0.0113) \\
\midrule  
M5    & 100   & 0.764 (0.1094) & 0.760 (0.1257) & 0.805 (0.1589) & 0.815 (0.1579) &       & 0.776 (0.1009) & 0.778 (0.1229) & 0.808 (0.1533) & 0.810 (0.1628) \\
      & 200   & 0.896 (0.0656) & 0.891 (0.0681) & 0.954 (0.0521) & 0.951 (0.0702) &       & 0.877 (0.0778) & 0.879 (0.0751) & 0.946 (0.0646) & 0.948 (0.0725) \\
      & 300   & 0.924 (0.0463) & 0.917 (0.0353) & 0.970 (0.0189) & 0.969 (0.0300) &       & 0.917 (0.0476) & 0.916 (0.0366) & 0.967 (0.0264) & 0.970 (0.0321) \\
      & 400   & 0.938 (0.0358) & 0.934 (0.0278) & 0.977 (0.0141) & 0.978 (0.0130) &       & 0.945 (0.0321) & 0.935 (0.0251) & 0.977 (0.0137) & 0.975 (0.0145) \\
          \bottomrule
    \end{tabular}%
    \end{adjustbox}
    }
  \label{tab:S_RI}%
\end{table}%

\begin{table}[H]
  \centering
  \caption{Means of 500 RSEs of subject-specific coefficient estimates across subjectwise representations corresponding to CID models {M1}--{M5} from various methods under different cluster membership models.}
    {\small
     \begin{adjustbox}{max width=0.6\linewidth}
         \begin{tabular}{ccccccccccc}
    \toprule
      \multicolumn{2}{c}{Method} & SP    & FP    & FP   & RHS & & SP    & FP    & FP   & RHS \\
      \cmidrule(rl){3-6} \cmidrule(rl){8-11}  
      \multicolumn{2}{c}{Initial value} & RC    & RC    & HS   & HS & & RC    & RC    & HS   & HS \\
      \cmidrule(rl){3-6} \cmidrule(rl){8-11}  
    CID   & $n$     & \multicolumn{4}{c}{Covariate-independence} & & \multicolumn{4}{c}{Covariate-dependence}\\
    \midrule  
    \multicolumn{1}{c}{M1} & 100   & 0.541 & 0.559 & 0.490 & 0.433 &       & 0.819 & 0.816 & 0.609 & 0.615 \\
          & 200   & 0.326 & 0.330 & 0.317 & 0.253 &       & 0.614 & 0.618 & 0.542 & 0.467 \\
          & 300   & 0.265 & 0.259 & 0.253 & 0.205 &       & 0.505 & 0.507 & 0.399 & 0.367 \\
          & 400   & 0.212 & 0.215 & 0.211 & 0.173 &       & 0.463 & 0.451 & 0.379 & 0.319 \\
          \midrule
    \multicolumn{1}{c}{M2} & 100   & 0.434 & 0.431 & 0.401 & 0.384 &       & 0.563 & 0.566 & 0.449 & 0.416 \\
          & 200   & 0.313 & 0.310 & 0.306 & 0.273 &       & 0.405 & 0.408 & 0.326 & 0.298 \\
          & 300   & 0.279 & 0.278 & 0.275 & 0.242 &       & 0.316 & 0.316 & 0.234 & 0.231 \\
          & 400   & 0.256 & 0.257 & 0.252 & 0.230 &       & 0.307 & 0.308 & 0.229 & 0.223 \\
          \midrule
    \multicolumn{1}{c}{M3} & 100   & 0.707 & 0.711 & 0.676 & 0.608 &       & 1.009 & 1.005 & 0.791 & 0.747 \\
          & 200   & 0.485 & 0.484 & 0.480 & 0.397 &       & 0.760 & 0.749 & 0.681 & 0.641 \\
          & 300   & 0.410 & 0.419 & 0.417 & 0.345 &       & 0.755 & 0.757 & 0.544 & 0.485 \\
          & 400   & 0.380 & 0.379 & 0.370 & 0.325 &       & 0.599 & 0.599 & 0.460 & 0.425 \\
          \midrule
    \multicolumn{1}{c}{M4} & 100   & 1.435 & 1.444 & 1.212 & 1.151 &       & 1.354 & 1.357 & 1.134 & 1.114 \\
          & 200   & 1.214 & 1.212 & 0.981 & 0.896 &       & 1.051 & 1.040 & 0.797 & 0.753 \\
          & 300   & 0.950 & 0.951 & 0.765 & 0.745 &       & 0.911 & 0.911 & 0.755 & 0.739 \\
          & 400   & 0.904 & 0.904 & 0.729 & 0.718 &       & 0.892 & 0.891 & 0.720 & 0.698 \\
          \midrule
    \multicolumn{1}{c}{M5} & 100   & 0.989 & 0.976 & 0.906 & 0.896 &       & 0.785 & 0.773 & 0.710 & 0.715 \\
          & 200   & 0.401 & 0.395 & 0.257 & 0.253 &       & 0.431 & 0.407 & 0.314 & 0.303 \\
          & 300   & 0.300 & 0.305 & 0.212 & 0.205 &       & 0.296 & 0.297 & 0.219 & 0.216 \\
          & 400   & 0.293 & 0.290 & 0.165 & 0.163 &       & 0.255 & 0.256 & 0.148 & 0.147 \\
          \bottomrule
    \end{tabular}%
      \end{adjustbox}
   }
  \label{tab:S_RMSE_b}%
\end{table}%

\begin{table}[htbp]
  \centering
  \caption{Means of 500 estimated number of clusters using distinct criteria across various methods under different CID and cluster membership models.}
 \setlength{\tabcolsep}{4pt} 
  \begin{subfigure}[t]{0.48\textwidth}
    \centering
    \caption*{$\textsc{SPIC}_1$}
    \begin{adjustbox}{max width=0.9\linewidth}
      \begin{tabular}{ccccccccccc}
        \toprule
      \multicolumn{2}{c}{Method} & SP    & FP    & FP   & RHS & & SP    & FP    & FP   & RHS \\
      \cmidrule(rl){3-6} \cmidrule(rl){8-11}  
      \multicolumn{2}{c}{Initial value} & RC    & RC    & HS   & HS & & RC    & RC    & HS   & HS \\
      \cmidrule(rl){3-6} \cmidrule(rl){8-11}  
       CID   & $n$     & \multicolumn{4}{c}{Covariate-independence} & & \multicolumn{4}{c}{Covariate-dependence}\\
      \midrule  
    M1    & 100   & 1.84  & 1.84  & 1.89  & 1.90 & & 1.90  & 1.88  & 1.97  & 1.95 \\
          & 200   & 2.00  & 2.00  & 2.00  & 2.00 & & 2.00  & 2.00  & 2.00  & 2.00 \\
          & 300   & 2.00  & 2.00  & 2.00  & 2.00 & & 2.00  & 2.00  & 2.00  & 2.00 \\
          & 400   & 2.00  & 2.00  & 2.00  & 2.00 & & 2.01  & 2.01  & 2.00  & 2.00 \\
          \midrule
    M2    & 100   & 1.43  & 1.42  & 1.49  & 1.56 & & 1.61  & 1.62  & 1.90  & 1.88 \\
          & 200   & 2.00  & 2.00  & 2.00  & 2.00 & & 2.00  & 2.00  & 2.00  & 2.00 \\
          & 300   & 2.00  & 2.00  & 2.00  & 2.00 & & 2.00  & 2.00  & 2.00  & 2.00 \\
          & 400   & 2.00  & 2.00  & 2.00  & 2.00 & & 2.00  & 2.00  & 2.00  & 2.00 \\
          \midrule
    M3    & 100   & 1.73  & 1.72  & 1.79  & 1.84 & & 1.83  & 1.83  & 1.95  & 1.94 \\
          & 200   & 2.00  & 2.00  & 2.00  & 2.00 & & 2.00  & 2.00  & 2.00  & 2.00 \\
          & 300   & 2.00  & 2.00  & 2.00  & 2.00 & & 2.00  & 2.00  & 2.00  & 2.00 \\
          & 400   & 2.00  & 2.00  & 2.00  & 2.01 & & 2.02  & 2.02  & 2.01  & 2.00 \\
          \midrule
    M4    & 100   & 1.87  & 1.87  & 1.98  & 1.96 & & 1.85  & 1.85  & 1.96  & 1.94 \\
          & 200   & 2.19  & 2.20  & 2.43  & 2.49 & & 2.25  & 2.24  & 2.34  & 2.34 \\
          & 300   & 2.86  & 2.86  & 2.98  & 2.99 & & 2.86  & 2.86  & 2.97  & 2.98 \\
          & 400   & 2.98  & 2.98  & 3.00  & 3.00 & & 2.98  & 2.98  & 3.00  & 3.00 \\
          \midrule
    M5    & 100   & 1.56  & 1.57  & 1.72  & 1.74 & & 1.62  & 1.61  & 1.73  & 1.74 \\
          & 200   & 2.00  & 1.99  & 2.00  & 2.00 & & 2.00  & 1.99  & 2.00  & 2.00 \\
          & 300   & 2.00  & 2.00  & 2.00  & 2.00 & & 2.00  & 2.00  & 2.00  & 2.00 \\
          & 400   & 2.00  & 2.00  & 2.00  & 2.00 & & 2.00  & 2.00  & 2.00  & 2.00 \\
        \bottomrule
      \end{tabular}
    \end{adjustbox}
  \end{subfigure}
  \hfill
  \begin{subfigure}[t]{0.48\textwidth}
    \centering
    \caption*{$\textsc{SPIC}_2$}
    \begin{adjustbox}{max width=0.9\linewidth}
           \begin{tabular}{ccccccccccc}
        \toprule
      \multicolumn{2}{c}{Method} & SP    & FP    & FP   & RHS & & SP    & FP    & FP   & RHS \\
      \cmidrule(rl){3-6} \cmidrule(rl){8-11}  
      \multicolumn{2}{c}{Initial value} & RC    & RC    & HS   & HS & & RC    & RC    & HS   & HS \\
      \cmidrule(rl){3-6} \cmidrule(rl){8-11}   
       CID   & $n$     & \multicolumn{4}{c}{Covariate-independence} & & \multicolumn{4}{c}{Covariate-dependence}\\
        \midrule
    \multicolumn{1}{c}{M1} & 100   & 1.95  & 1.96  & 1.97  & 1.97 &  & 1.97  & 1.96  & 2.00  & 1.99 \\
          & 200   & 2.01  & 2.01  & 2.01  & 2.00 & & 2.02  & 2.03  & 2.01  & 2.01 \\
          & 300   & 2.03  & 2.03  & 2.03  & 2.00 & & 2.03  & 2.02  & 2.01  & 2.00 \\
          & 400   & 2.02  & 2.02  & 2.02  & 2.01 & & 2.03  & 2.03  & 2.01  & 2.01 \\
          \midrule
    \multicolumn{1}{c}{M2} & 100   & 1.87  & 1.87  & 1.92  & 1.93 & & 1.86  & 1.91  & 1.98  & 1.98 \\
          & 200   & 2.00  & 2.00  & 2.00  & 2.00 & & 2.00  & 2.01  & 2.00  & 2.00 \\
          & 300   & 2.01  & 2.01  & 2.02  & 2.02 & & 2.00  & 2.01  & 2.00  & 2.00 \\
          & 400   & 2.06  & 2.06  & 2.06  & 2.05 & & 2.02  & 2.03  & 2.00  & 2.00 \\
          \midrule
    \multicolumn{1}{c}{M3} & 100   & 1.95  & 1.94  & 1.96  & 1.97&  & 1.98  & 1.98  & 2.00  & 2.01 \\
          & 200   & 2.03  & 2.05  & 2.05  & 2.03 & & 2.04  & 2.04  & 2.02  & 2.01 \\
          & 300   & 2.11  & 2.12  & 2.12  & 2.11 & & 2.07  & 2.07  & 2.03  & 2.02 \\
          & 400   & 2.36  & 2.37  & 2.38  & 2.36 & & 2.08  & 2.08  & 2.04  & 2.02 \\
          \midrule
    \multicolumn{1}{c}{M4} & 100   & 2.12  & 2.12  & 2.21  & 2.18 & & 2.12  & 2.11  & 2.22  & 2.21 \\
          & 200   & 2.74  & 2.74  & 2.95  & 2.97  && 2.75  & 2.73  & 2.93  & 2.95 \\
          & 300   & 2.97  & 2.97  & 3.00  & 3.00  && 2.98  & 2.98  & 3.00  & 3.00 \\
          & 400   & 3.01  & 3.01  & 3.00  & 3.00  && 3.01  & 3.01  & 3.00  & 3.00 \\
          \midrule
    \multicolumn{1}{c}{M5} & 100   & 1.91  & 1.95  & 1.97  & 1.91  && 1.80  & 1.84  & 1.86  & 1.90 \\
          & 200   & 2.03  & 2.03  & 2.02  & 2.01  && 2.02  & 2.02  & 2.01  & 2.01 \\
          & 300   & 2.02  & 2.00  & 2.00  & 2.01  && 2.01  & 2.01  & 2.00  & 2.01 \\
          & 400   & 2.02  & 2.02  & 2.01  & 2.00  && 2.00  & 2.00  & 2.00  & 2.00 \\
        \bottomrule
      \end{tabular}
    \end{adjustbox}
  \end{subfigure}
  \label{tab:S_est_k}
\end{table}

\begin{figure}[htbp]
\centering
 \includegraphics[width=16cm,height=9cm]{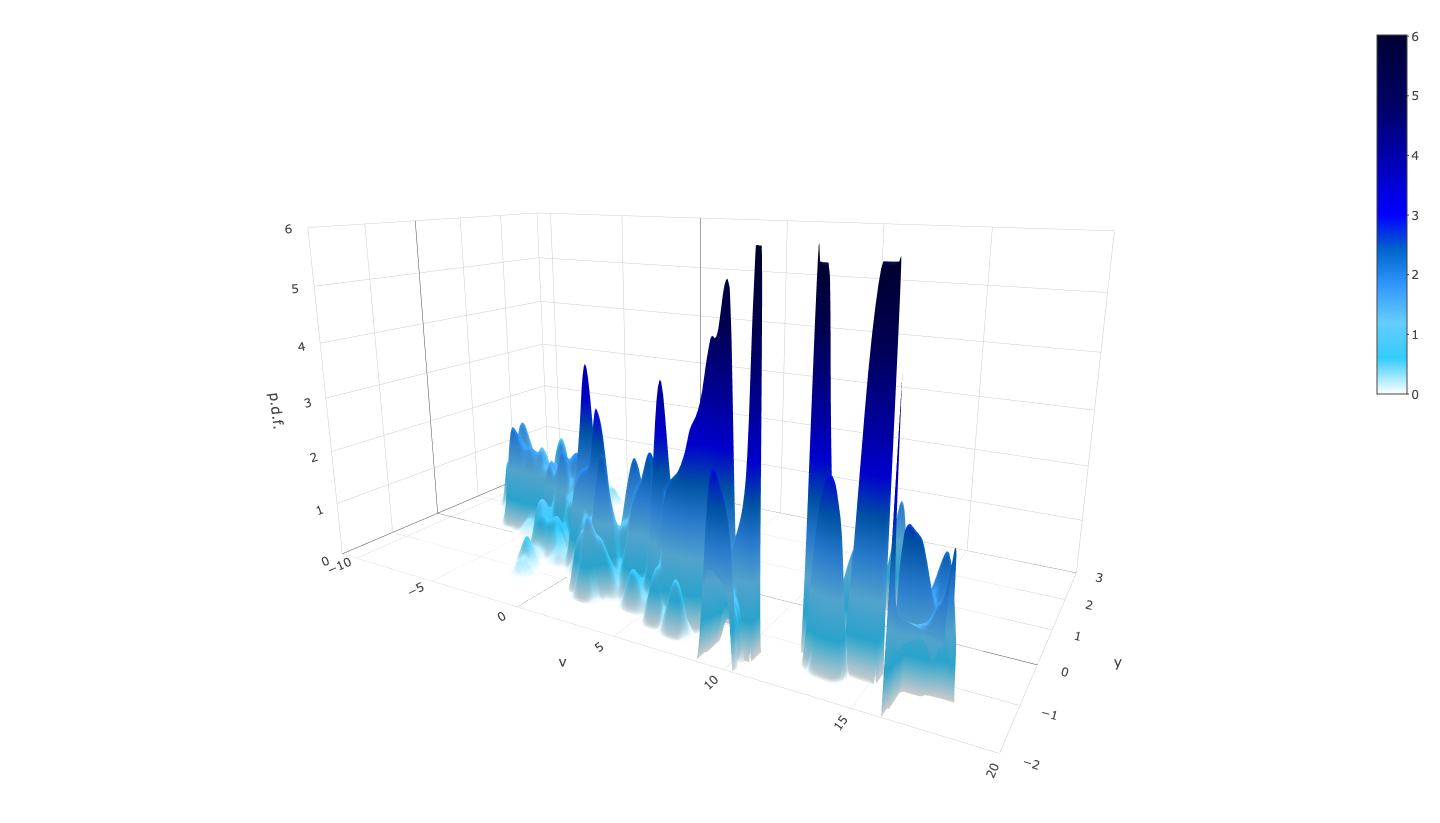}
 \caption{{\footnotesize Estimated conditional probability density function from the clusterwise index probability density model incorporating cluster-specific intercepts and riverbank effects for a real estate valuation study.}} \label{fig:Sindian_pdf_D} 
\end{figure}

\begin{figure}[htbp]
\centering
\begin{subfigure}[b]{\textwidth} 
    \centering
    \includegraphics[width=12cm,height=3.375cm]{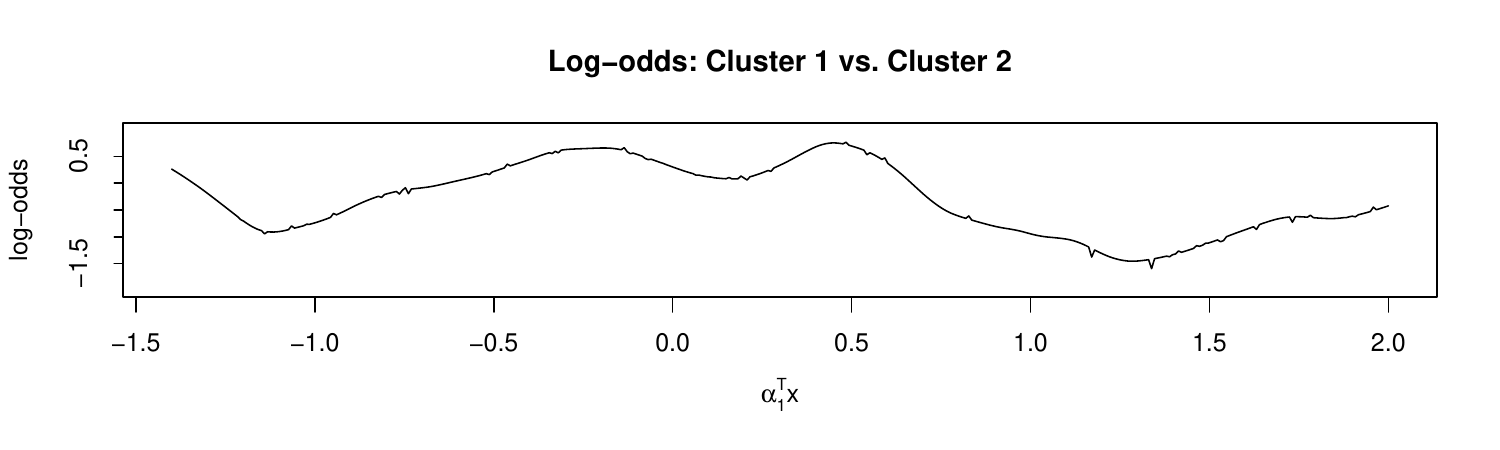} 
    \caption{{\footnotesize SDR-based Cluster Membership Model}}
\end{subfigure}


\begin{subfigure}[b]{\textwidth} 
    \centering
    \includegraphics[width=12cm,height=3.375cm]{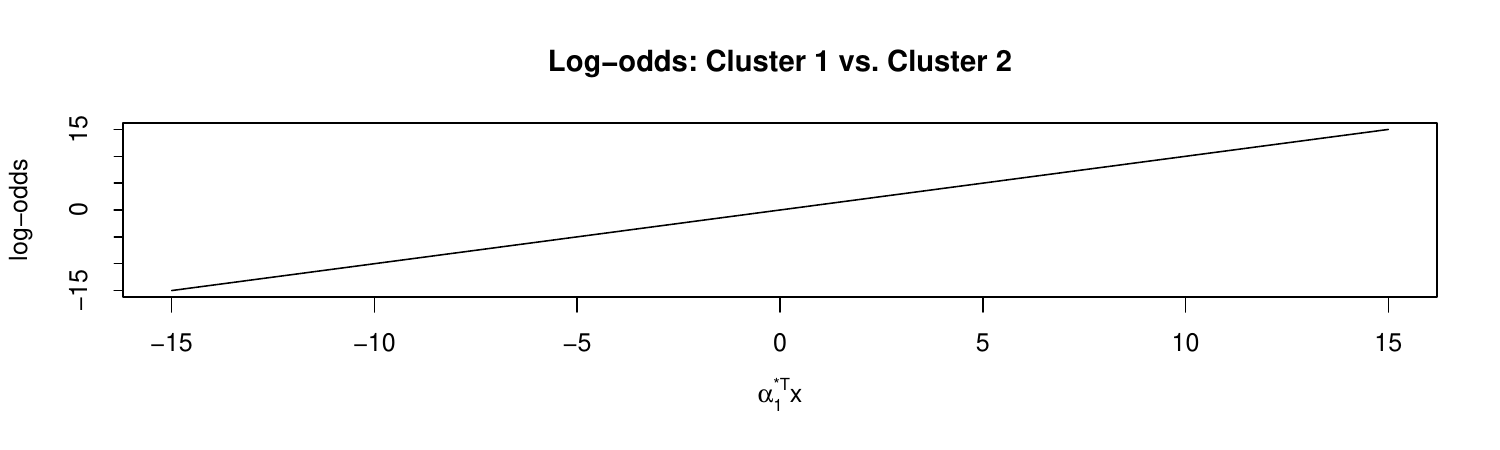}
    \caption{{\footnotesize Logistic Cluster Membership Model}}
\end{subfigure}

 \caption{{\footnotesize Estimated log-odds of membership in Cluster 1 relative to Cluster 2 based on the SDR-based and logistic cluster membership models for a real estate valuation study.}} \label{fig:Sindian_prob} 
\end{figure}

\begin{figure}[htbp]
\centering
 \includegraphics[width=16cm,height=9cm]{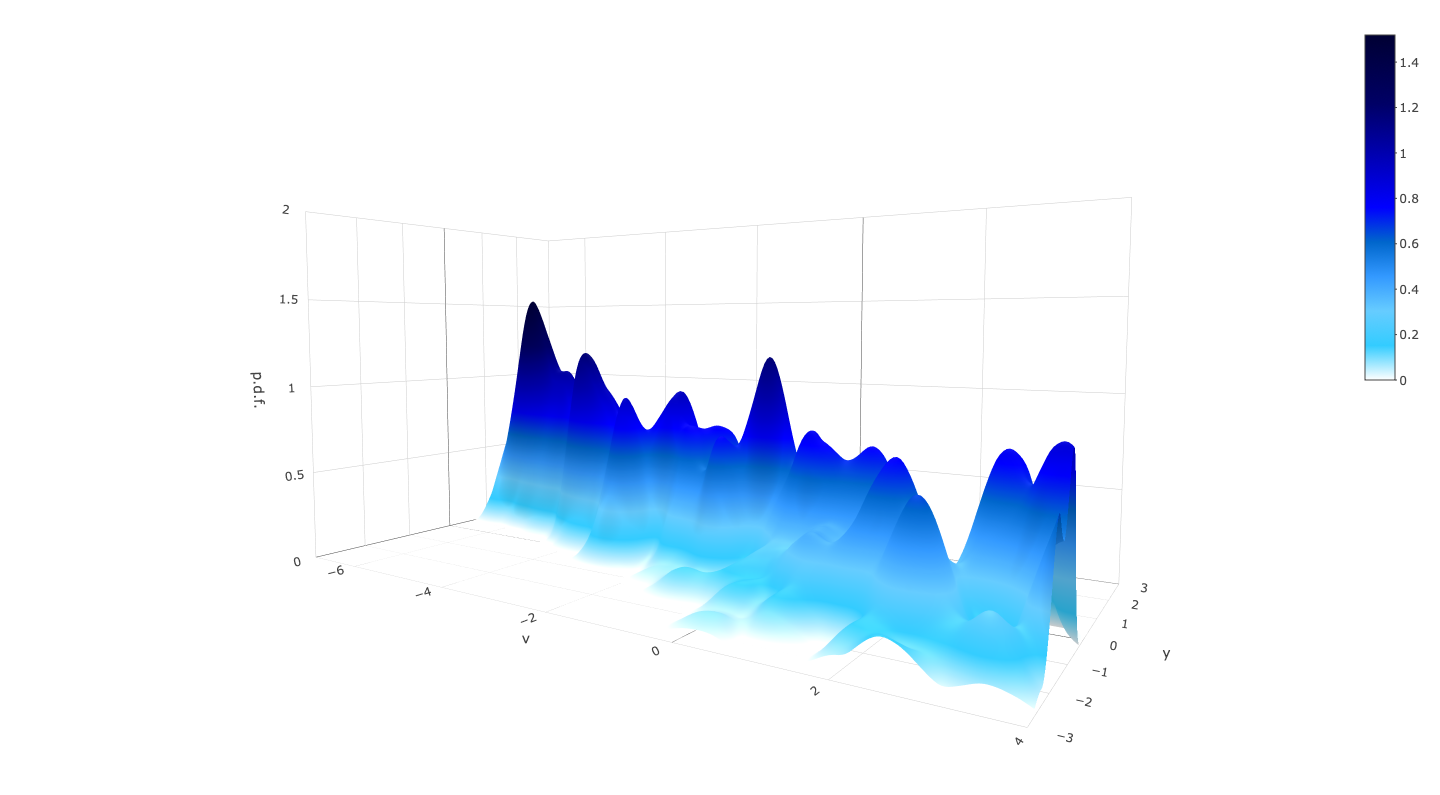}
 \caption{{\footnotesize Estimated conditional probability density function from the clusterwise index probability density model incorporating cluster-specific intercepts and disease effects for the Cleveland Heart Disease study.}} \label{fig:CLE_pdf_D} 
\end{figure}

\begin{figure}[t]
\centering
\begin{subfigure}[b]{\textwidth} 
    \centering
    \includegraphics[width=12cm,height=3.375cm]{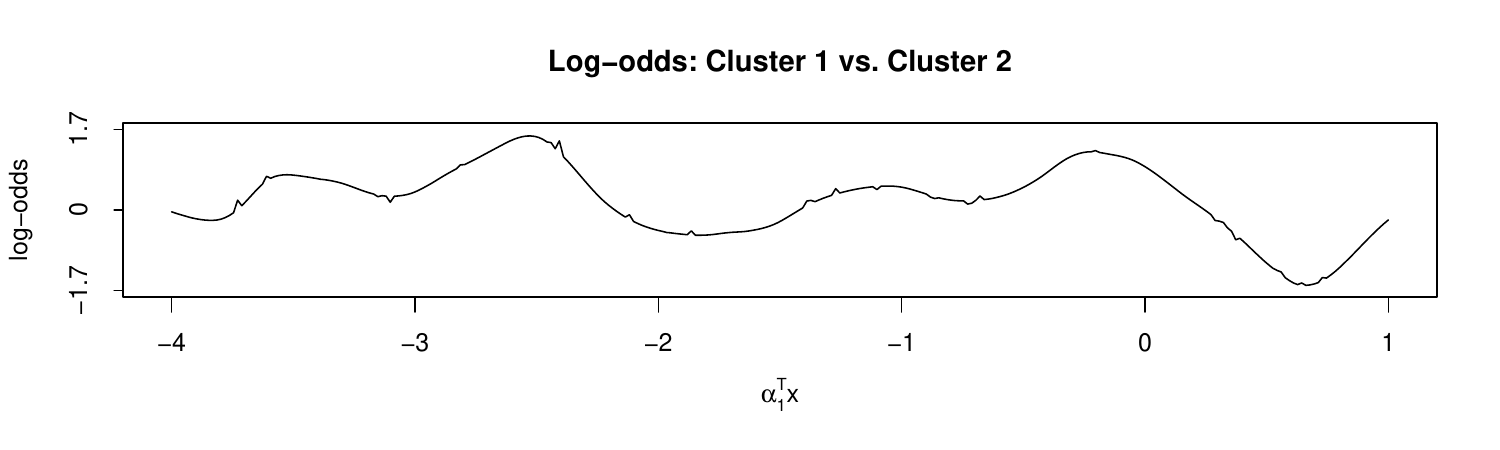} 
    \caption{{\footnotesize SDR-based Cluster Membership Model}}
\end{subfigure}


\begin{subfigure}[b]{\textwidth} 
    \centering
    \includegraphics[width=12cm,height=3.375cm]{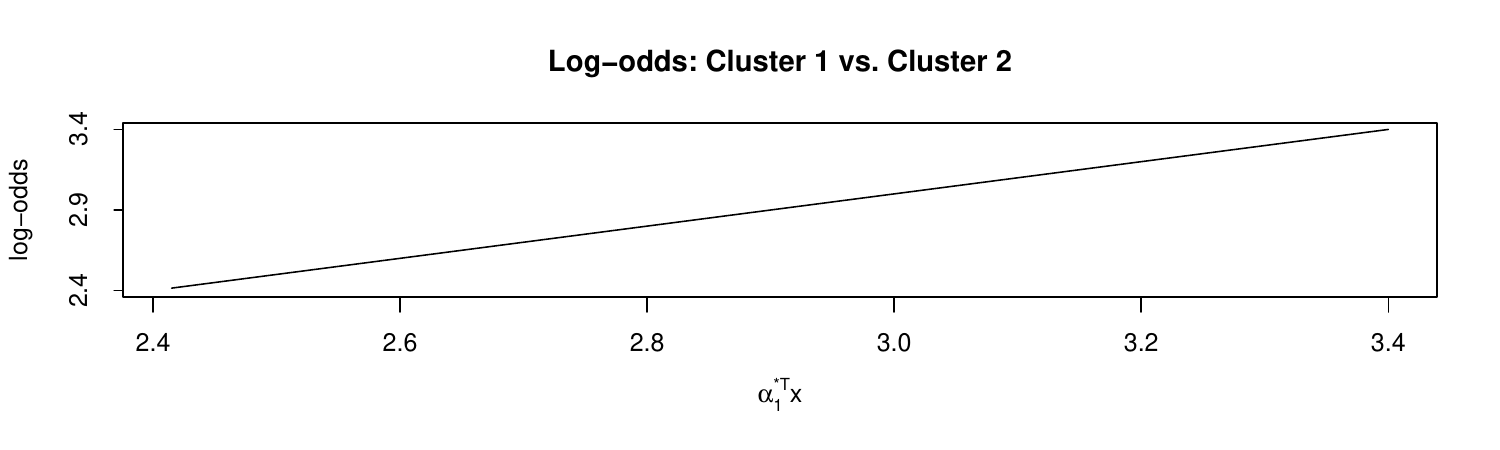}
    \caption{{\footnotesize Logistic Cluster Membership Model}}
\end{subfigure}

 \caption{{\footnotesize Estimated log-odds of membership in Cluster 1 relative to Cluster 2 based on the SDR-based and logistic cluster membership models for the Cleveland Heart Disease study.}} \label{fig:CLE_prob} 
\end{figure}

\subsection{Application in AIDS Clinical Trials Group Study 175} 
\label{subsec:ACTG}

Established in 1987 by the National Institute of Allergy and Infectious Diseases, the AIDS Clinical Trials Group (ACTG) has played a central role in advancing research on HIV/AIDS. One of its most influential studies, ACTG 175, was a landmark randomized clinical trial that substantially shaped the development of antiretroviral therapy protocols. This study assessed the relative effectiveness of four therapy protocols in 2139 adults infected with HIV type I whose initial CD4 T cell counts ranged from 200 to 500 cells per cubic millimeter. Further details on the study design and results can be found in \cite{hammer1996trial}.

We modeled the log-transformed CD4 cell counts measured at $20\pm5$ weeks, denoted as \textit{cd4}, as the response variable. Following \cite{ma2020exploration}, the covariates included treatment indicators (\textit{tr1}, \textit{tr2}, \textit{tr3}), log-transformed baseline CD8 cell counts (\textit{cd8}), age (\textit{age}), weight (\textit{weight}), Karnofsky performance score (\textit{ks}), and gender (\textit{gender}; 0 = female, 1 = male). Treatment assignment was represented by three indicators: (0,0,0) for zidovudine monotherapy, (1,0,0) for zidovudine combined with didanosine, (0,1,0) for zidovudine combined with zalcitabine, and (0,0,1) for didanosine monotherapy. Participants with extreme values in the response variable or covariates were excluded to improve numerical stability, resulting in a final sample of 2,054 adults.
The CID model, incorporating cluster-specific intercepts and treatment effects, was applied to capture latent heterogeneity in treatment response. According to $\textsc{SPIC}_1$ and $\textsc{SPIC}_2$, the number of clusters selected across the RC, HS, SP, and RSP methods was consistently five.
The resulting cluster assignments demonstrated high concordance: the RI value between the RC and RSP methods was 0.929, and the RI value for each pair among the HS, SP, and RSP methods was 1.000. Furthermore, the HS, SP, and RSP methods have the same partition and parameter estimates, while the RC method produced a nearly identical partition with only minor discrepancies.

The RSP method yields the following estimated cluster index coefficients and associated standard errors for the five cluster indices in the CID model:
\begin{table}[H]
  \centering 
  \begin{adjustbox}{max width=\linewidth}
    \begin{tabular}{cccccccccc}
    \toprule
 Cluster & intercept & \textit{tr}1   & \textit{tr}2   & \textit{tr}3   & \textit{cd}8   & \textit{age}   & \textit{weight} & \textit{ks}    & \textit{gender} \\
 \midrule
    1     & 0  & 8.93 (3.084)  & 7.01 (2.265)  & 6.92 (2.397) & 1  & -0.18 (0.055) & 0.42 (0.157) & 1.01 (0.303) & -0.67 (0.264)\\
    2     & -12.45 (4.046)& 3.62 (1.514) & 1.66 (0.705) & 2.02  (0.748)& 1  & -0.18 (0.055) & 0.42 (0.157) & 1.01 (0.303) & -0.67 (0.264)\\
    3     & -6.04 (1.828)& 3.56 (0.988) & 2.37 (0.677) & 2.42 (0.611) & 1  & -0.18 (0.055) & 0.42 (0.157) & 1.01 (0.303) & -0.67 (0.264)\\
    4     & -19.34 (6.621)& 1.41 (0.602) & 1.05 (0.543) & -0.23 (0.317)& 1 & -0.18 (0.055) & 0.42 (0.157) & 1.01 (0.303) & -0.67 (0.264)\\
    5     & 11.92 (4.156)& 6.91 (2.273) & 3.78 (1.475) & 4.42 (1.505)  & 1 & -0.18 (0.055) & 0.42 (0.157) & 1.01 (0.303) & -0.67 (0.264)\\
    \bottomrule
    \end{tabular}%
    \end{adjustbox}
\end{table}%
\noindent The cluster-specific coefficient estimates reveal that Cluster 5 corresponds to the highest \textit{cd4} across all treatment conditions, followed in descending order by Clusters 1, 3, 2, and 4. Clusters 2, 3, and 5 exhibit a consistent treatment ordering---zidovudine$+$didanosine $>$ didanosine $>$ zidovudine $+$zalcitabine $>$ zidovudine---indicating a shared pattern of treatment responsiveness. In contrast, Clusters 1 and 4 deviate from this ordering, reflecting partial heterogeneity in treatment effects. Among the cluster-invariant coefficients, \textit{cd8}, \textit{weight}, and \textit{ks} show significant associations with \textit{cd4}, while \textit{age} and \textit{gender} exhibit significant adverse associations.
In Figure \ref{fig:ACTG_mean}, the estimated mean function increases monotonically with respect to the index value, reflecting a consistent positive association.
Figure \ref{fig:ACTG_pdf} highlights the heterogeneity in distributional forms induced by the interaction between latent cluster structure and treatment effects.

In this data analysis, the structural dimension in the SDR-based cluster membership model was identified as one, with Cluster 5 serving as the reference group. The estimated index coefficients for the CS direction and their standard errors are presented below.
\vspace{-0.07in}
\begin{table}[H]
  \centering
  \footnotesize
    \begin{tabular}{cccccccc}
    \toprule
   \textit{tr}1   & \textit{tr}2   & \textit{tr}3   & \textit{cd}8   & \textit{age}   & \textit{weight} & \textit{ks}    & \textit{gender} \\
    \midrule
    2.12  (0.401)& 3.81 (0.240)& 3.07 (0.254)& -0.24 (0.098)& 1  & 0.18 (0.112)& -0.20 (0.131) & -1.27 (0.331)  \\
    \bottomrule
    \end{tabular}%
\end{table}%
\noindent The estimated direction suggests that the treatment indicators (\textit{tr}1,\textit{tr}2,\textit{tr}3) and age contribute significantly to the index, whereas \textit{cd8} and \textit{gender} contribute in the opposite direction. This contrast in covariate effects provides evidence of a covariate-dependent cluster membership model.
Figure \ref{fig:ACTG_prob} highlights the limitations of the multinomial logistic model in capturing the structure of the cluster membership probabilities.

\begin{figure}[htbp]
\centering
 \includegraphics[width=13cm,height=6.5cm]{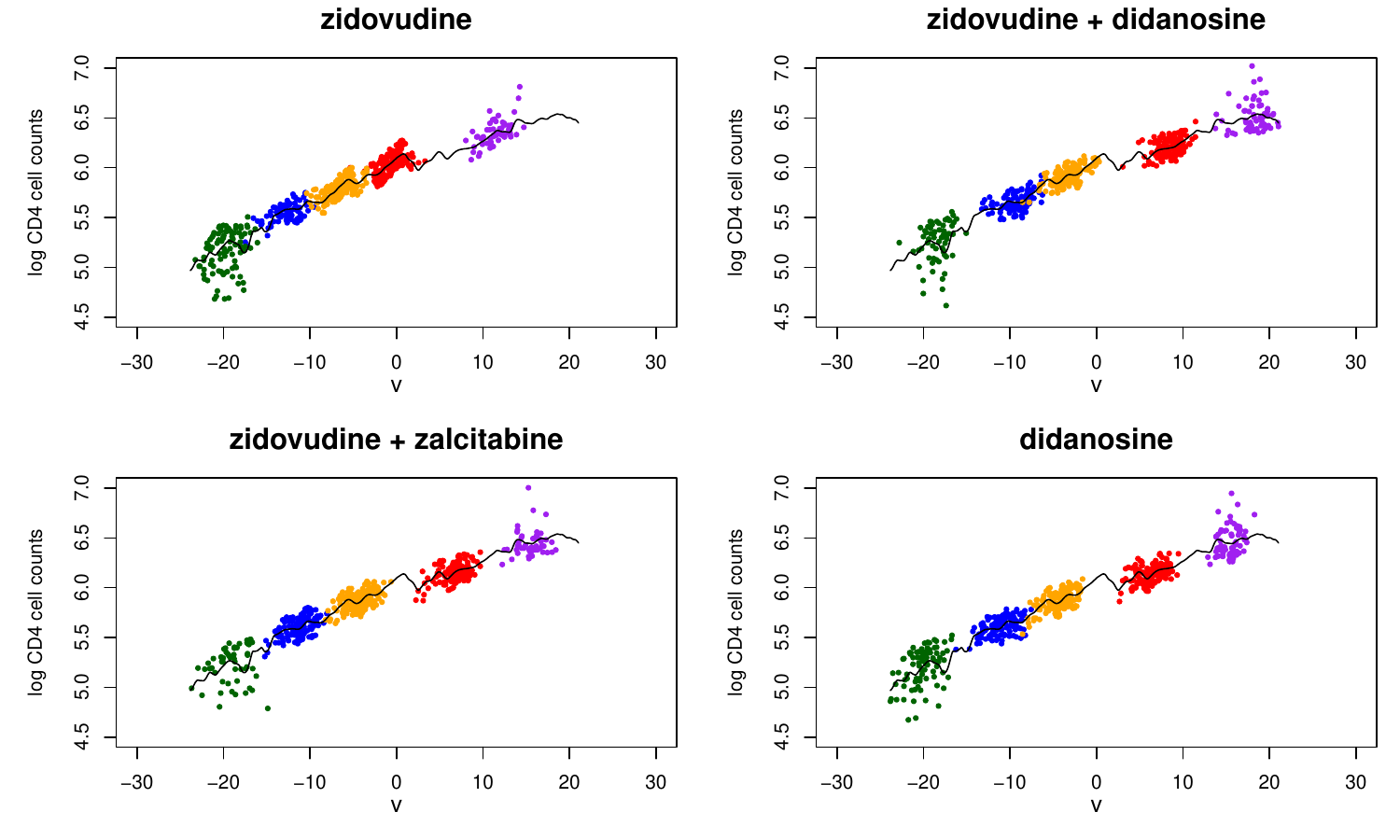}
 \caption{{\footnotesize Scatter plots of log-transformed CD4 cell counts for HIV-infected adults across four therapies from ACTG 175 overlaid with the estimated mean function from the clusterwise index mean model incorporating cluster-specific intercepts and treatment effects. The red, blue, orange, green, and purple dots represent Clusters 1, 2, 3, 4, and 5, respectively.}} \label{fig:ACTG_mean} 
\end{figure}

\begin{figure}[H]
\centering
 \includegraphics[width=16cm,height=9cm]{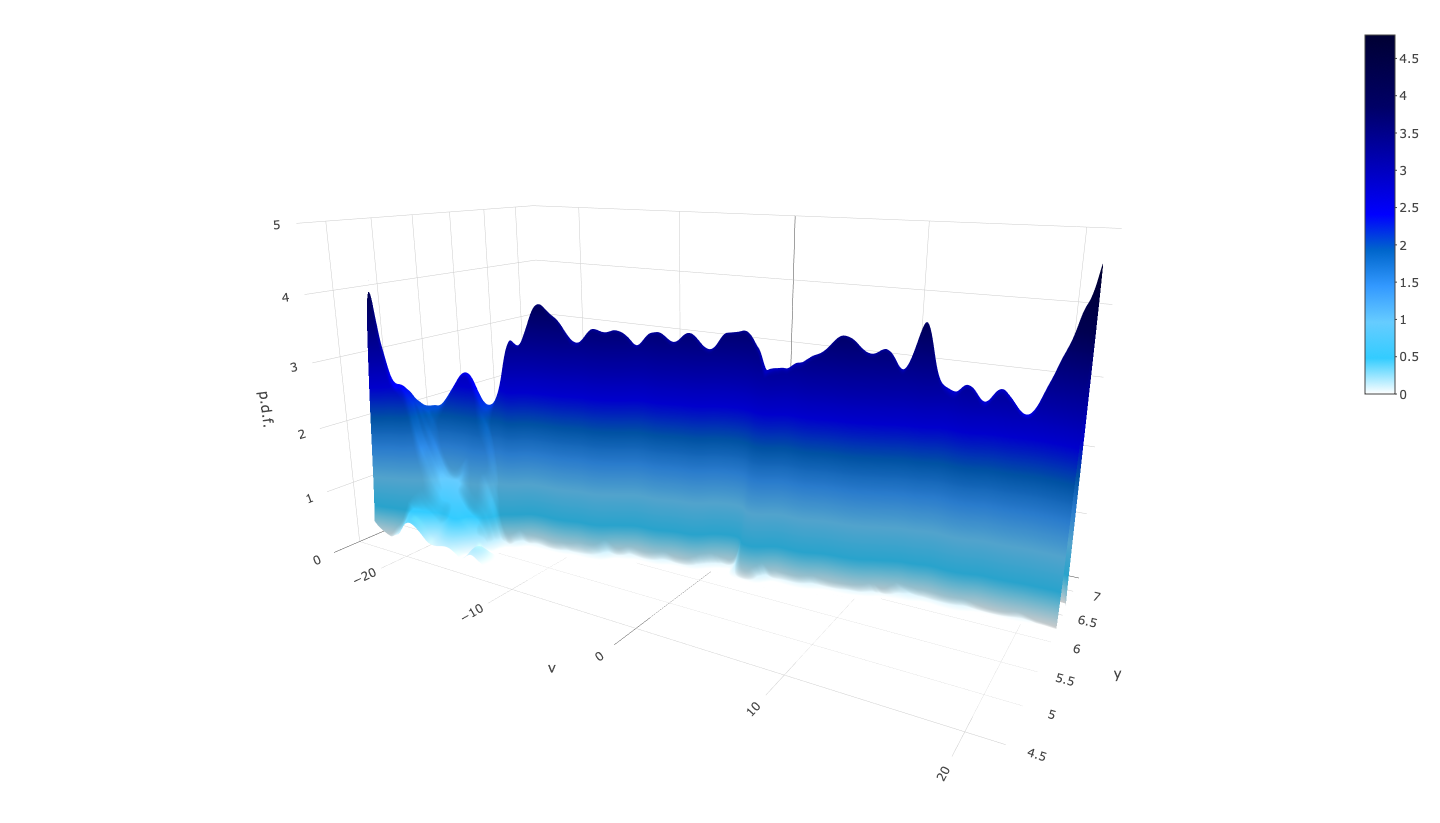}
 \caption{{\footnotesize Estimated conditional probability density function from the clusterwise index probability density model incorporating cluster-specific intercepts and treatment effects for the AIDS Clinical Trials Group study.}} \label{fig:ACTG_pdf} 
\end{figure}

\begin{figure}[H]
\centering
\begin{subfigure}[t]{0.45\linewidth} 
    \centering
    \includegraphics[width=\linewidth]{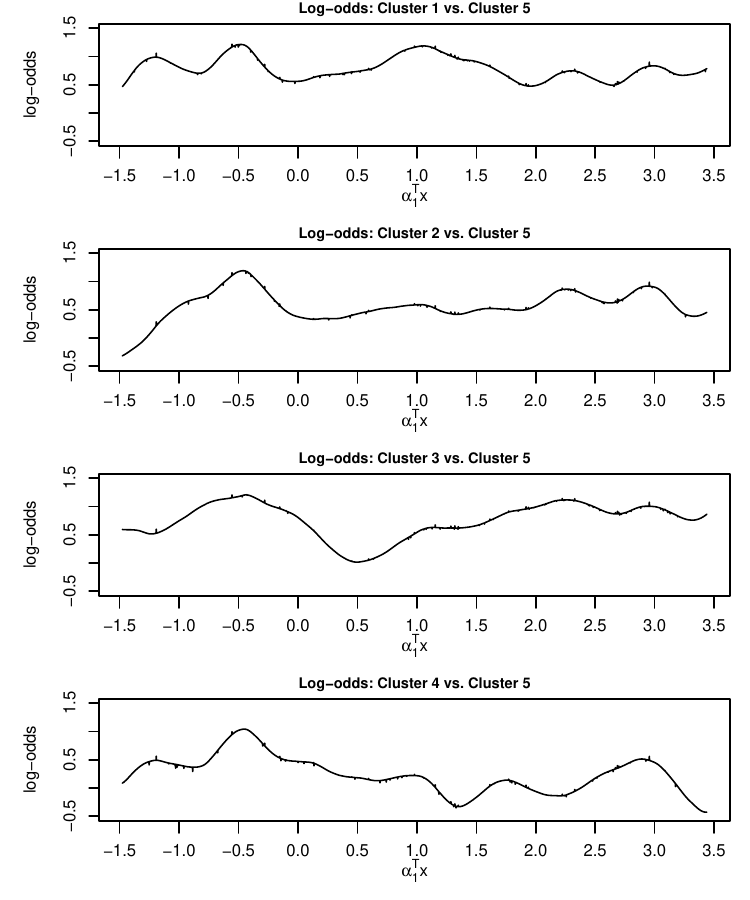} 
    \caption{{\footnotesize SDR-based Cluster Membership Model}}
\end{subfigure}
\begin{subfigure}[t]{0.45\linewidth} 
    \centering
    \includegraphics[width=\linewidth]{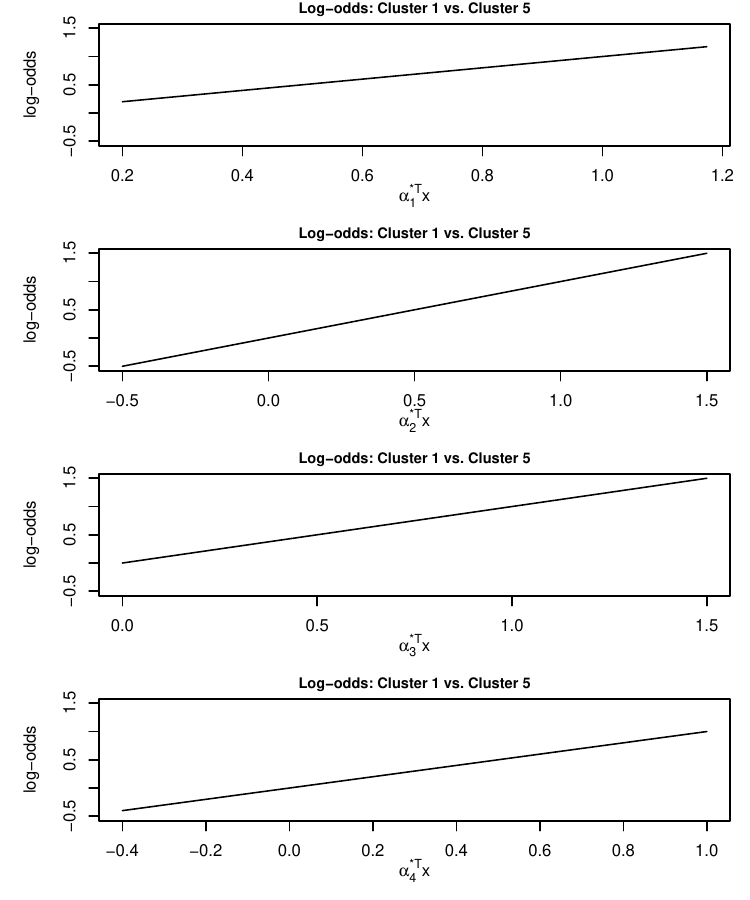} 
    \caption{{\footnotesize Logistic Cluster Membership Model}}
\end{subfigure}

\caption{{\footnotesize Estimated log-odds of membership in Clusters 1, 2, 3, and 4 relative to Cluster 5 based on the SDR-based and logistic cluster membership models for the AIDS Clinical Trials Group study.}}
\label{fig:ACTG_prob}
\end{figure}

\end{document}